\documentclass[11pt,a4paper]{article}
\usepackage{jheppub}

\usepackage{amsmath,amsopn,amssymb,amsfonts}
\usepackage{booktabs}
\usepackage{verbatim}
\usepackage{amsthm}
\def\beq{\begin{equation}}
\def\eeq{\end{equation}}
\def\bea{\begin{eqnarray}}
\def\eea{\end{eqnarray}}
\def\half{\frac{1}{2}}

\def\Tr{{\rm Tr}}

\def\qb {\bar{q}}

\def\cF{\cal F}

\newcommand{\RR}{{\mathbb R}}
\newcommand{\CC}{{\mathbb C}}
\newcommand{\ZZ}{{\mathbb Z}}
\newcommand{\HH}{{\mathbb H}}

\theoremstyle{plain}
\newtheorem{thm}{Theorem}[section]
\newtheorem{corollary}[thm]{Corollary}
\newtheorem{lemma}[thm]{Lemma}
\newtheorem{proposition}[thm]{Proposition}

\theoremstyle{definition}
\newtheorem{definition}[thm]{Definition}
\newtheorem*{remark}{Remark}

\def\CF{{\cal F}}

\def\CN{{\cal N}}
\def\CS{{\cal S}}
\def\CT{{\cal T}}

\def\CC{{\cal C}}

\newcommand{\IR}{\mathbb{R}}
\newcommand{\IC}{\mathbb{C}}
\newcommand{\IZ}{\mathbb{Z}}
\newcommand{\IH}{\mathbb{H}}

\newcommand{\SCFT}{\mathrm{SCFT}}

\def\R{\mathbb{R}}
\def\Z{\mathbb{Z}}

\newcommand{\pa}{\partial}

\def\i{i}

\def\bea{\begin{eqnarray}}
\def\eea{\end{eqnarray}}
\def\be{\begin{equation}}
\def\ee{\end{equation}}
\def\ba{\begin{align}}
\def\ea{\end{align}}
\def\bse{\begin{subequations}}
\def\ese{\end{subequations}}

\newcommand{\bem}{\begin{pmatrix}}
\newcommand{\eem}{\end{pmatrix}}
\newcommand{\bmat}[1]{\left[ \smallmatrix #1 \endsmallmatrix \right]}

\def\qb {\bar{q}}

\def\p{\partial}
\def\tbar{\overline{\tau}}
\def\zbar{\overline{z}}
\def\ubar{\overline{u}}

\def\lp{\left(}
\def\rp{\right)}
\def\lb{\left[}
\def\rb{\right]}

\def\aa{\bar{a}}
\def\bb{\bar{b}}

\def\sl{SL(2,\mathbb{R})}
\def\slc{SL(2,\mathbb{R})/U(1)}

\def\={\;  = \;}
\def\+{\, + \,}

\def\wt{\widetilde}
\def\wh{\widehat}
\def\bar{\overline}

\def\rt2{\sqrt{2}}

\newcommand{\I}{i}
\newcommand{\dd}{\mathrm{d}}

\def\vth{\vartheta}
\def\ve{\varepsilon}
\def\v{\varphi}

\def\t{\tau}

\def\b{\beta}

\def\O{{\Omega}}

\def\th{\theta}\def\vth{\vartheta}

\renewcommand{\th}{\theta}

\def\p{\partial}
\def\tbar{\bar \tau}
\def\zbar{\bar z}
\def\ubar{\bar u}
\def\ch{\mathrm{ch}}

\def\myth1{\vartheta_{1}}

\def\bh{B}

\title{ADE Double Scaled Little String Theories, Mock Modular Forms and Umbral Moonshine}

\preprint{EFI-14-22}

\author[1]{Jeffrey A.~Harvey}
\author[2]{Sameer Murthy}
\author[3]{and Caner Nazaroglu}
\affiliation[1,3]{Enrico Fermi Institute, University of Chicago \\
5640 S Ellis Ave., Chicago, Illinois 60637, USA}
\affiliation[2]{Department of Mathematics, King's College London \\
The Strand, London WC2R 2LS, U.K.}
\emailAdd{j-harvey@uchicago.edu, sameer.murthy@kcl.ac.uk, cnazaroglu@uchicago.edu}

\abstract{
We consider double scaled little string theory on $K3$. These theories are labelled by a positive integer $k \ge 2$ and an $ADE$ root lattice with Coxeter number $k$.  We count BPS fundamental string states 
in the holographic dual of this theory using the 
superconformal field theory
$K3 \times \left( \frac{SL(2,\RR)_k}{U(1)} \times \frac{SU(2)_k}{U(1)} \right) \big/ \IZ_k$. We show that the BPS fundamental string states  that are counted by the second
helicity supertrace of this theory give rise to
weight two mixed mock modular forms. We compute the helicity supertraces using two separate techniques:  a path integral analysis that leads to a modular invariant but non-holomorphic answer,  and a Hamiltonian analysis of the contribution from discrete states which leads to a holomorphic but not modular invariant answer. From a mathematical point of view the Hamiltonian analysis leads to a mixed mock modular form while the path integral gives the completion of this mixed mock modular form. We also  compare these weight two mixed mock modular forms to those that appear in instances of Umbral Moonshine labelled by Niemeier root lattices $X$ that are powers of $ADE$ root lattices and find that they are equal up to a constant factor that we determine. In the course of the analysis we encounter an interesting generalization of Appell-Lerch sums and generalizations of  the Riemann relations of Jacobi theta functions that they obey. 

}
\keywords{modular forms, moonshine, little string theory, NS5-branes}

\begin{document}

\maketitle
\section{Introduction and Motivation}

Little string theory (LST) originated in the study of the dynamics of modes localized on solitonic fivebranes that act as a source for the massless
two form field $B$ that originates in the Neveu-Schwarz sector of superstring theory. These fivebranes are often called NS5-branes and exist in type IIA and IIB string theory as well as in heterotic string theory \cite{Strominger:1990et, Callan:1991dj, Callan:1991ky, Callan:1991at}. As is the case for
D-branes, it is possible to take a limit in which the dynamics on the fivebranes decouples from the bulk, but unlike the corresponding limit for D-branes,
the limit in which the dynamics on fivebranes decouples keeps the energy scale $E$ fixed relative to the string scale $m_s$. Taking the string
coupling $g_s \rightarrow 0$ with $E \sim m_s$ leads to a six-dimensional theory on the fivebranes, dubbed little string theory (LST) in
\cite{Losev:1997hx}, which is a non-trivial interacting
theory with fascinating properties. In particular it is a six-dimensional supersymmetric theory with either $(2,0)$ or $(1,1)$ supersymmetry in
IIA or IIB string theory respectively and has stringy excitations and behavior including T-duality and a Hagedorn density of states. 
Unlike critical string theory, it does not include gravity as there is no massless spin two particle in the spectrum. Evidence that such theories exist and discussions of the properties described above can be found in  \cite{Berkooz:1997cq, Seiberg:1997zk, Losev:1997hx, Giveon:1999px, Giveon:1999tq} with \cite{Aharony:1999ks} and \cite{Kutasov} containing useful reviews.  Some important earlier ideas with important applications to black hole physics appeared in \cite{DVVI, DVVII}. The papers \cite{Israel:2004ir, Eguchi:2004yi,Chang:2014jta} are particularly useful references for many of the aspects of LST that we will utilize in our analysis. 

Techniques for analyzing little string theory are unfortunately limited. Some early analysis was based on discrete light-cone gauge theory
and Matrix theory \cite{Berkooz:1997cq, Aharony:1997an, Aharony:1997th} and there is also an approach using deconstruction \cite{ArkaniHamed:2001ie}, but most recent analyses have utilized a holographic description \cite{Aharony:1998ub, Itzhaki:1998dd, Boonstra:1998mp}  of the theory 
of $k$ fivebranes based on the superconformal field theory (SCFT)
\be \label{chs}
M_6 \times \RR_\phi \times SU(2)_k
\ee
that describes the space-time background sourced by fivebranes in the decoupling limit above~\cite{Callan:1991dj}. Here $\RR_\phi$ is a supersymmetric linear dilaton background,
$SU(2)_k$ is a level $k$ supersymmetric Wess-Zumino-Witten (WZW) model and in the simplest case of fivebranes in flat space $M_6=\RR^{5,1}$ is the free  superconformal field theory
describing the space-time coordinates tangent to the fivebrane world-volume. In this paper we will take $M_6=M_4 \times \RR^{1,1}$ with
$M_4$ a hyperK\"ahler manifold. The utility of this holographic description is limited
by the presence of a strong coupling region: since the dilaton field is linear in the coordinate $\phi$, the theory  has one asymptotic region
in which the string coupling goes to zero and another asymptotic region where the string coupling goes to infinity and thus one needs to be
able to do strong coupling computations to fully study the theory in both regions.

The strong coupling problem can be circumvented by Higgsing the theory, that is by separating the fivebranes so that they are located in  a $\IZ_k$ symmetric fashion on a circle
of radius $r_0$ in a two-dimensional plane in the four-dimensional space transverse to the fivebrane before taking the scaling limit \cite{Sfetsos:1998xd, Giveon:1999px, Giveon:1999tq}. 
The Higgsing breaks the~$SU(2) \times SU(2)$ symmetry of the string theory in the background \eqref{chs} down to~$U(1) \times \IZ_{k}$. If we consider
fivebranes in type IIB string string theory then by $S$-duality we know there is a $U(k)$ gauge theory on the fivebrane broken down to
$U(1)^k$  and that there are $D1$ branes that can stretch between the fivebranes with mass  of order $r_0/g_s$ whose lowest energy states
include the massive gauge bosons from breaking $U(k)$ down to $U(1)^k$. 
One then considers the double scaling limit $g_s \rightarrow 0, r_0 \rightarrow 0$ with $r_0/g_s$ fixed. In this limit the massive gauge bosons
survive as states with finite mass and one can study physics at the string scale rather than taking a low-energy limit. 

Although our description so far involves a $U(k)=A_{k-1} \times U(1)$ gauge theory arising on $k$ coincident fivebranes in type IIB string theory,
there are dual descriptions of this theory which arise when we identify an $S^1$ in one of the directions transverse to the fivebranes. In particular,
T-duality relates the $A_{k-1}$ fivebrane theory to string theory near a $\IC^2/\IZ_k$ singularity \cite{Ooguri:1995wj}. 
In the dual picture, separating the NS5-branes on a circle corresponds to deforming the singularity with a deformation parameter $\mu$. The double scaling limit, $g_s, \mu \to 0$ with $g_s/\mu^{1/k}$ fixed, is then described by the SCFT
\be \label{scft}
\RR^{1,1} \times M^4 \times \left( \frac{SL(2,\RR)_k}{U(1)} \times \frac{SU(2)_k}{U(1)} \right) \big/ \IZ_k
\ee
where $SL(2,\RR)_k/U(1)$ is the supersymmetric, non-compact, ``cigar" conformal field theory and $SU(2)_k/U(1)$ is a coset theory
that describes the $N=2$ minimal model SCFTs.  The infinite throat region that led to a strong coupling region has now been capped
off and as a result it is possible to do reliable perturbative computations on aspects of double scaled little string theory (DSLST) using this SCFT.
The $\IZ_k$ orbifold is present in this construction in order to project onto a set of $U(1)_R$ charges that are consistent with imposition of
the GSO projection needed to obtain a theory with space-time supersymmetry.  Note that the elliptic genus of the SCFT (\ref{scft}) minus the $\RR^{1,1} \times M^4$ factor has been studied in \cite{Ashok:2012qy} in connection with mirror symmetry.

The McKay correspondence
between simply laced ($ADE$) root diagrams and finite subgroups $\Gamma$ of $SU(2)$ suggests that we can define an $ADE$ extension of this construction by
considering string theory on $\IC^2/\Gamma$. This extension to include $D_k$ and $E_6, E_7, E_8$ theories is
supported by the structure of the SCFT (\ref{scft}) which also has an $ADE$ classification of modular invariant partition functions that
are compatible with space-time supersymmetry \cite{Cappelli:1987xt, Cecotti:1992rm, Martinec:1988zu, Vafa:1988uu}.  

Thus the SCFT appearing in (\ref{scft}) is labelled by a hyperK\"ahler manifold $M^4$, a positive integer $k \ge 2$ and a choice of $ADE$
root system or Dynkin diagram  which can be $A_{k-1}$ for any choice of $k$, $D_{1+k/2}$ for $k$ even, and $E_6$, $E_7$ or $E_8$ for $k=12, 18, 30$
respectively\footnote{In the $ADE$ classification of modular invariant partition functions at level $k$ the Coxeter number of the $ADE$ root system must
equal $k$.} and this SCFT provides a holographic description of the corresponding $ADE$ DSLST compactified on $M^4$.

In this paper we will focus on the case $M^4=K3$. Although this choice could be motivated by the desire to study LST in situations with reduced
supersymmetry, our motivation arises from the fact that with this choice there is a natural object, the second helicity supertrace, that counts BPS
states, and that from the analysis in \cite{Harvey:2013mda} for $k=2$ seems to have some connection to the mock modular form appearing
in Mathieu moonshine \cite{Eguchi:2010ej, Cheng:2010pq, Gaberdiel:2010ch, Gaberdiel:2010ca, Eguchi:2010fg}. Given that Mathieu moonshine
has an extension to Umbral moonshine \cite{Cheng:2012tq,Cheng:2013wca} and that Umbral moonshine is classified by root lattices with ADE components, total rank $24$ and
common Coxeter number, it is natural to ask whether the mock modular forms of Umbral moonshine also appear in helicity supertraces of LST on $K3$
at higher values of $k$ and whether there is a connection between the $ADE$ classification of LST and the $ADE$ classification of Umbral Moonshine.
A related analysis for $M^4=T^4$ can be found in \cite{Cheng:2014zpa}.

The second helicity supertrace that we will study is given by
\be \label{helsuptr}
\widehat \chi_2(\tau) =  \Tr \, J_s^2 \, (-1)^{F_s} \, q^{L_0 - c/24} \, \qb^{\tilde{L}_0 - \tilde{c}/24}
\ee
where the trace is over all space-time states, that is including both Ramond (R) and Neveu-Schwarz (NS) sectors of the SCFT
with GSO projection. $J_s$ is the generator of a $U(1)$ rotation in space-time which we take to be rotation about the asymptotic $S^1$ in
the cigar CFT and $F_s$ is the space-time fermion number which is $0$ for bosons and $1$ for fermions.  Standard arguments show that this
quantity receives contributions only from BPS states, and in a theory with a discrete spectrum, is a holomorphic function of $\tau$ . See e.g. \cite{Kiritsis:1997gu} for a review. 

One important fact about the SCFT (\ref{scft}) is that  the $SL(2,\RR)_k/U(1)$ factor describes string propagation on a {\it non-compact }space. As a result the spectrum contains both discrete, normalizable states, as well as a continuum of scattering states. Recent analysis of the elliptic genus of the
$SL(2,\RR)_k/U(1)$ SCFT in (\ref{scft}) \cite{Troost:2010ud, Eguchi:2010cb, Ashok:2011cy} has revealed a tension between modularity and holomorphicity in computations of the elliptic genus.  Contributions from discrete states are holomorphic functions of $\tau$ because  of cancellations between fermions and bosons, but since modular transformations mix characters of discrete states with those of continuum states, the contribution of discrete states alone is not modular invariant. Including the continuum states leads to a modular invariant expression for the elliptic genus, but it is
no longer holomorphic because violations of holomorphy arise from continuum states due to  differences between the density of states of fermions and
bosons. 

This kind of tension between holomorphy and modularity is precisely the defining feature of mock modular forms. Mock modular forms were first
introduced by Ramanujan in the last letter he wrote to Hardy in 1920, but were only understood  in more depth in recent times due to work of Zwegers \cite{zwegers2008mock} and others. For reviews see \cite{MR2605321, MR2555930}.  Recall that a modular form of weight
$k$ is a holomorphic function $f(\tau)$ on the upper half plane $\HH$ obeying
\be
f \left( \frac{a \tau+b}{c \tau +d} \right) = (c \tau+d)^k f(\tau), \qquad \begin{pmatrix} a & b \\ c & d \end{pmatrix} \in SL_2(\ZZ) \,. 
\ee
In the simplest case a weakly holomorphic mock modular form of weight $k$ for $SL_2(\ZZ)$ is a holomorphic function on $\HH$ with at most
exponential growth as $\tau \rightarrow i \infty$ which is part of a pair of holomorphic functions $(h(\tau), g(\tau))$ where $g(\tau)$, called the
shadow of $h(\tau)$, is a modular form of weight $2-k$ and the completion of $h(\tau)$, $\widehat h(\tau)$, given by
\be \label{mockdef}
\widehat h(\tau)= h(\tau) + g^*(\tau) 
\ee
with 
\be
g^*(\tau) = (4i)^{k-1} \int_{- \bar \tau}^{\infty} (z+\tau)^{-k} \overline{g(- \bar z)} dz \, ,
\ee
transforms like a modular form of weight $k$ on $SL_2(\ZZ)$. Note that $h(\tau)$ is holomorphic, but not modular if the shadow $g(\tau)$ is
non-zero, while in this case the completion $\widehat h(\tau)$ is modular but not holomorphic, and obeys the 
equation:
\be    
(4\pi\t_2)^k\,\,\frac{\pa \wh h(\t)}{\pa \bar{\tau}} \= -2\pi i\;\bar{g(\tau)} \, .  
\ee

Mathieu moonshine is the observation that a particular mock modular form of weight $1/2$, $H^{A_1^{24}}$ with $q$ expansion
\be
H^{A_1^{24}}= 2 q^{-1/8} \left( -1 + 45 \, q + 231 \, q^2 + 770 \, q^3 + 2277 \, q^4+ 5796 \, q^5 + \cdots \right)
\ee
appears in the decomposition of the elliptic genus of $K3$ into characters of the $N=4$ superconformal algebra (SCA) and that the
coefficients $45, 231, 770, 2277, 5796$ are dimensions of irreducible representations (irreps) of the sporadic Mathieu group $M_{24}$ while
the higher coefficients have decompositions into small numbers of irreps.  The shadow of $H^{A_1^{24}}$ is $ 24 \eta(\tau)^3$ where
$\eta(\tau)$ is the Dedekind eta function.  Strictly speaking, the pair $(H^{A_1^{24}}, 24 \eta(\tau)^3)$ does not satisfy the definition above
because $\eta(\tau)^3$ acquires a phase (an eighth root of unity) under modular transformations. Thus as is often done for modular forms,
we need to generalize the definition (\ref{mockdef}) to allow for a multiplier system which can be a phase, or more generally for vector
valued modular or mock modular forms with $n$ components, a multiplier $\rho: SL_2(\ZZ) \rightarrow GL_N(\CC)$ that determines the matrix
transformation on the components of the (mock) modular form that accompanies a modular transformation on $\tau$.

We will also need to extend the definition (\ref{mockdef}) to include mixed mock modular forms. If $M_k$ is the space of 
weight $k$ modular forms then we define a mixed mock modular form of weight $k| \ell$ to be a holomorphic function with
polynomial growth at the cusps that has a completion 
\be
\widehat h = h(\tau) + \sum_j f_j \, g_j^* \,,  \qquad f_j \in M_\ell, ~~ g_j \in M_{2-k-\ell} \, ,
\ee
such that $\widehat h$ transforms like a modular form of weight $k$. 
As an example, the product $\eta(\tau)^3 H^{A_1^{24}}$ is a mixed mock modular form of weight $2$ with shadow $24 \, \eta(\tau)^3 \, \overline{ \eta(\tau)^3} $.

In \cite{Harvey:2013mda} it was shown that the second helicity supertrace of the SCFT (\ref{scft}) at $k=2$ is the completion of this mixed mock
modular form, $\chi_2^{k=2}= - (1/2) \, \eta(\tau)^3 \widehat H^{A_1^{24}}(\tau)$. This result suggests a possible connection between the the $k=2$
DSLST and Mathieu moonshine. Mathieu moonshine has been extended to Umbral Moonshine in which vector-valued mock modular forms
$H^X$ are labelled by the root lattices $X$ of the $23$ even, self-dual rank 24 lattices with non-empty root systems, that is by the $23$ Niemeier
root lattices $L^X$, and exhibit moonshine properties for groups $G^X$ which are defined in terms of the Niemeier lattice $L^X$ as
\be
G^X= {\rm Aut}(L^X)/W_X
\ee
where $W_X$ is the Weyl group of $X$.

The Niemeier root lattices $X$ are composed of $ADE$ components with total rank $24$ and equal Coxeter numbers $m(X)$.  Umbral
Moonshine generalizes Mathieu moonshine in that $X=A_1^{24}$ with $m(X)=2$ leads to the moonshine group $G^{A_1^{24}}=M_{24}$
and to the mock modular form $H^{A_1^{24}}$ given above. The weight $3/2$ shadow $\eta(\tau)^3$ appearing for $X=A_1^{24}$ has a
generalization involving the vector-valued weight $3/2$ theta functions
\be
S_{m,r}(\tau) = \sum_{n \in \ZZ} \, (2mn+r) \, q^{(2mn+r)^2/4m}
\ee
with $r=1,2,\cdots m-1$, and for each $X$ labeling an instance of Umbral moonshine we can associate a mixed mock modular form of weight
$2$ given by
\be
\chi_2^X= \sum_r S_{m(X),r}(\tau) \, H^X_r(\tau) \, .
\ee
For $m=2$ we have $S_{2,1}(\tau)= \eta(\tau)^3$.

We thus have on the one hand the SCFT (\ref{scft}), specified by an integer $k \ge 2$ and a choice $Y$ of $ADE$ root diagram with Coxeter number
$k$ with its second helicity supertrace $\widehat \chi_2^Y(\tau)$  (which we now label by $Y$) and we will show later that $\widehat \chi_2^Y(\tau)$ is
the completion of a mixed mock modular form of weight two. On the other hand we have the quantities $\chi_2^X$ labelled by Niemeier lattices
$X$ with an $ADE$ classification and these are also mixed mock modular forms of weight two. 
By comparing shadows we conclude that
the $\widehat \chi_2^Y$ are, up to a numerical factor that we will determine, the completions of the $\chi_2^X$ when the $X$ are powers of a single
$ADE$ component. 

To  summarize, we have two motivations for the computations done in this paper. The first is to explore further the structure of DSLST compactified on
$K3$ by computing the second heliticy supertrace that counts the BPS states of this theory. The second is to further explore possible connections
between DSLST on K3 and Mathieu and Umbral moonshine.

The outline of this paper is as follows. In Section 2 we describe the structure of the superconformal field theory (\ref{scft}) in more detail, discuss the structure of the helicity supertrace and then compute the full, non-holomorphic supertrace by generalizing the techniques used in \cite{Harvey:2013mda} which involve performing
an integral over the holonomy of a $U(1)$ gauge field on the torus. In Section 3 we analyze the holomorphic contribution to the helicity supertrace by
performing an explicit sum over the discrete characters of the superconformal field theory. We compare our results to the holomorphic part of the previous computation and find perfect agreement. The result of these computations involves a sum of the form $\sum_r S_{k,r} h_r$ where the $h_r$ are
vector-valued mock modular forms. We turn in Section 4 to a discussion of the relation between the mock modular forms $h_r$ arising in these computations
and the vector-valued mock modular forms $H^X_r$ of Umbral Moonshine. In Section 5 we conclude and offer some thoughts about possible future directions of research. The appendices contain many technical and interesting mathematical details that appear in this work. These include a generalization of the Riemann identities for Jacobi theta functions to Riemann-like identities involving Appell-Lerch sums and a generalization of the Appell-Lerch sum
that appears prominently in \cite{zwegers2008mock} to a Appell-Lerch like sum depending on the modular parameter $\tau$ as well as on three elliptic variables.

\section{The holographic dual of DSLST}\label{wsheetsec}
We consider the holographic dual of DSLST as described by the SCFT~\eqref{scft} with~$M^{4}=K3$, 
with $\CN=4$ supersymmetry \cite{Kounnas:1993ix, Antoniadis:1994sr, Israel:2004ir} on the worldsheet. 
This SCFT describes string propagation in the background of $k$ NS5-branes
that are wrapped on $K3$, and separated along the transverse~$\IR^{4}$ in a ring structure, in a near-horizon
double scaling limit~\cite{Giveon:1999zm}. 
The geometric picture of the $SL(2)/U(1)$ coset is a semi-infinite cigar \cite{Witten:1991yr, Elitzur:1991cb, Mandal:1991tz}, 
which is asymptotically a linear-dilaton times a circle. 
Translation around the circle is an exact symmetry of the theory, and the conserved $U(1)$ momentum  corresponds to the 
spacetime R-charge. 

The string theory based on the above SCFT has 16 conserved supercharges. 
We would like to compute the second helicity supertrace~\eqref{helsuptr} of fundamental strings in this background, thus
generalizing the calculation of~\cite{Harvey:2013mda} for the~$k=2$ A-type theory. In path-integral language
this is a torus partition function with two insertions of the spacetime R-charge.
We will essentially follow the ideas and calculations of~\cite{Harvey:2013mda} to compute the second
string helicity supertrace. We will briefly sketch the procedure below, mainly focussing on the new
aspects compared to~\cite{Harvey:2013mda} and relegating many details to the appendices.
We begin by considering a string wrapped on the~$S^{1}$ and moving in time.
The full string theory includes the reparametrization ghosts of the $\CN=1$ string worldsheet,
which have the effect of cancelling the oscillator modes of the $\RR \times S^1$ factor in the above SCFT.
This leaves us\footnote{As mentioned in~\cite{Harvey:2013mda}, one can also obtain this SCFT
using a gauge-fixing condition on the string worldsheet.}
with the momentum and winding modes around the~$S^{1}$ and all the fluctuations in the
``internal'' theory with central charge~$c = \wt c = 12$:
\be \label{lcCFT}
K3 \times \Big( \frac{SL(2)_{k}}{U(1)} \times \frac{SU(2)_{k}}{U(1)}\Big)\big/\IZ_{k} \, .
\ee

We shall consider the~$K3$ to be at the~$T^{4}/\IZ_{2}$ orbifold point, but we find that
the final result of our calculation only depends on the elliptic genus of~$K3$ which does not
depend on the moduli of the $K3$ surface. 
A new ingredient compared to the~$k=2$ situation is the supersymmetric~$SU(2)_{k}/U(1)$ theory
with central charge $c = 3-\frac6k$ (which becomes trivial at~$k=2)$.
The characters of this theory are defined in terms of a branching relation using the characters of
the~$SU(2)$ WZW model after splitting off a~$U(1)$ factor~\cite{Kawai:1993jk}. 
We gather some useful information about the characters in Appendix \ref{sec:su2char}.
We treat the non-compact gauged WZW $SL(2)_{k}/U(1)$ model based on the more recent work
of~\cite{Troost:2010ud,Eguchi:2010cb,Ashok:2011cy}.
The coset is expressed as~$SL(2)_{k} \times U(1)^{\IC}/U(1)$ where~$U(1)^{\IC}$ is a complexification of the
gauged~$U(1)$ subgroup, to which one adds a~$(b,c)^\text{cig}$ ghost system of central charge~$c=-2$.
The supersymmetric~$SL(2)_{k}$ consists of a bosonic~$H_{3}^{+}$ WZW model at level $k+2$,
and two free fermions $\psi^{\pm}$ (and their right-moving counterparts).
The coset~$U(1)^{\IC}/U(1)$ is represented by the real boson~$Y$.
Related string worldsheet calculations that combine the two cosets have been discussed 
in~\cite{Israel:2004ir, Eguchi:2004yi} in a different context.

The holonomies of the~$U(1)$ gauge field around the two cycles of the torus are represented by a
complex parameter\footnote{Throughout this paper, we will use the subscripts 1 and 2 on a complex variable
to denote its real and imaginary parts, i.e.~$\t=\t_{1}+i \t_{2}$, $u=u_{1}+i u_{2}$ etc.}~$u = a \t +b$.
It is useful to consider a combination of the boson~$Y$ and the gauge field holonomy which is called~$Y^{u}$,
as in~\cite{Troost:2010ud,Ashok:2011cy}. This field~$Y^{u}$ is a compact real boson at radius~$R=\sqrt{2/k}$, and can be
thought of, in the asymptotic variables, as the angular direction of the cigar. In the exact theory it is described
by the compact level~$k$ CFT~$U(1)_{k}$.
The bosonic~$H_{3}^{+}$, the two fermions, the~$Y^{u}$ boson, and the~$(b,c)^\text{cig}$ ghosts are all
solvable theories and are coupled by the holonomy~$u$ that has to be integrated over the elliptic
curve~$E(\t)=\IC/(\IZ \tau + \IZ)$.

We now write down the partition functions of the various fields entering~\eqref{lcCFT}.
We will keep track of the gauge field holonomy~$u$ as well as a parameter~$z$ which we introduce
as the chemical potential of the spacetime R-charge. This R-charge is the diagonal~$J^{3}$ component
of the~$SU(2) \times SU(2)$ R-symmetry rotations of the CHS solution, and it is identified with the~$U(1)$
momentum around the cigar (see~\cite{Giveon:1999zm, Murthy:2003es} for a discussion). The appearance
of the two potentials~$u$ and~$z$ is governed by the corresponding charges of the above fields.
Our notations and conventions are summarized in Appendix~\ref{Conventions}.

The bosonic $H_{3}^{+} = SL(2,\IC)/SU(2)$ model at level~$k$ contributes:
\be \label{H3plus}
Z_{H_{3}^{+}}(\t,u) \= \frac{(k+2) \sqrt{k}}{\t_{2}^{1/2}} \, e^{{2 \pi u_{2}^{2} / \tau_{2}}} \,  \frac{1}{|\theta_{11}(\tau;u)|^{2}} \, .
\ee
The $(b,c)^{\rm cig}$ ghosts have the contribution:
\be\label{bcpartfn}
Z_{\rm gh}(\t) \= \t_{2} \, |\eta(\t)^{2} |^{2} \ .
\ee
The two left-moving fermions $\psi^{\pm}$ have a contribution in the NS sector\footnote{The prefactor in front of the
usual expression for free fermions arises because of a factor of $k+2$ in the action of these fermions. This prefactor
cancels an equivalent one in the numerator of the bosons in~\eqref{H3plus}.}  \cite{AlvarezGaume:1986es}:
\be\label{psipm}
ch^{\rm NS}_{\rm \psi^{\pm}}(\t) \=  \frac1{\sqrt{k+2}} \,e^{{-\pi u_{2}^{2} / \tau_{2}}} \, e^{{2 \pi i u_{1} u_{2} / \tau_{2}}} \,
\frac{\theta_{00}(\tau;u)}{\eta(\t)} \, ,
\ee
the other~$NS (-1)^{F}, R, R(-1)^{F}$ characters can be derived from this by~$\CN=2$ spectral flow on the worldsheet.
We shall denote these four characters by~$ch^{(ab)}_{\rm \psi^{\pm}}$ where~$a=0,1$ stands 
for $NS, R$, and $b = 0$ and $1$ stands for the trace with an insertion of $1$ and $(-1)^{F}$  respectively.

The characters of the $\CN=2$ minimal model $SU(2)_{k}/U(1)$ are labelled 
by\footnote{In this section we work at a fixed level $k$ and do not display it in the various formulas 
in order to avoid clutter. In Section~\ref{charsec} and the Appendices, there is an explicit superscript 
on the functions~$C^{j}_r$, $ \chi^{j,a}_{r}$, and $ \chi_{j}$ that shows the $k$-dependence
(see~\eqref{Clkrdef}, \eqref{eq:su2u1_defn}, \eqref{defchikj}).}~$(j,r)$, $(a,b)$~\cite{Kawai:1993jk} and have the form:
\beq\label{minchar}
C^{j}_r \bmat{a \\ b}(\tau; v) =  e^{i \pi a b/2} \left( \chi^{j,a}_{r}(\t;v) + (-1)^{b}  \chi^{j,a+2}_{r} (\t;v) \right)\, .
\eeq
Here $v$ is the chemical potential for the $U(1)_R$ symmetry of the minimal model, and in what 
follows we will only need these characters evaluated at $v=0$. 
The explicit form of the functions~$ \chi^{j,a}_{r}$ are presented in Equations~\eqref{eq:su2u1_defn}, \eqref{eq:string_fnc}, 
\eqref{defchikj}. In this section we will only use that they obey the branching relation (see Equation~\eqref{chibranching}): 
\begin{equation} \label{chibranching1}
\sum_{r \in \ZZ_{2k}} \chi_r^{j,a}(\tau; 0) \, \vartheta_{k,r} \lp \tau; \frac{w}{2} \rp 
\= \chi_{j} (\tau ; w) \, \vartheta_{2,a} \lp \tau ; \frac{w}{2} \rp
\end{equation}
where the functions $\chi_{j}$ are defined by:
\be\label{defstringfn}
 \chi_{j} (\tau; z) = - \, \frac{\vartheta_{k,j+1} (\tau ; z/2) - \vartheta_{k,-j-1} (\tau ; z/2)}{\theta_{1}(\tau,z)} \, .
\ee

The~$U(1)_{k}$ fibered over the coset  has the characters:
\beq\label{uonechar}
ch_{k,r} = e^{{-\pi k (u_{2}+z_{2})^{2} / 2 \tau_{2}}} \,  \frac{ \vartheta_{k,r} \left( \tau ;  -(u+z)/2 \right)}{ \eta(\tau)} \, .
\eeq

On tensoring the chiral part of the free fermions, the $U(1)_{k}$ and the~$SU(2)/U(1)$ characters above, 
we obtain: 
\bea\label{zint}
&& C(\t; u,z) \sum_{r \in \IZ_{2k}} ch^{(ab)}_{\rm \psi^{\pm}} ch_{k,-r} \, C^{j}_r \bmat{a \\ b}  \\
& & \qquad =  \frac{\theta_{ab}(\tau ; u)}{\eta^{2}(\tau)}  \sum_{r \in \IZ_{2k}}
\vartheta_{k,r} (\tau ; (u+z)/2) \, \left( \chi^{j,a}_{r}(\t ; 0) + (-1)^{b}  \chi^{j,a+2}_{r} (\t; 0) \right) \, , \cr
& & \qquad =  \frac{\theta_{ab}(\tau; u)}{\eta^{2}(\tau)}  \left(  \chi_{j} ( \tau ; u +z) \, 
 \vartheta_{2,a} ( \tau ;  (u +z)/2) + (-)^{b}   \chi_{j}(\tau ; u+z) \,   \vartheta_{2,a+2} (\tau ; (u +z)/2) \right) \, , \nonumber
\eea
where we have performed the sum over $r$ using the branching relation \eqref{chibranching1}. 
The last line in \eqref{zint} can further be rewritten as
\beq\label{zintfin}
\frac{\theta_{ab} (\tau ; u)
\, \theta_{ab} \left( \tau ; u +z \right)}{\eta^{2}(\tau)} \,  \chi_{j} (\tau ; u+z) .
\eeq
The prefactor~$C(\t;u,z)$ in~\eqref{zint} can be read off from Equations \eqref{psipm}, \eqref{minchar}, \eqref{uonechar},
and we shall keep track of the overall normalisations separately in what follows. The above manipulations show that 
the characters from the~$SU(2)/U(1)$ factor and the~$U(1)$ of the cigar 
recombine to give back the~$SU(2)$ characters~$\chi_{j}$ in the path integral computation of the second helicity supertrace.  

We now add in the rest of the fermionic fields and perform the sum over spin structures with the GSO projection.
We use the~$T^{4}/\IZ_{2}$ oribifold description of the~$K3$. The partition sum naturally splits into two 
pieces arising from the untwisted and twisted sectors of the this orbifold, that we label by a 
subscript~$(r,s)$ with~$r,s=0,1$. This description is fairly standard in the literature, 
we present some details in \S\ref{charsec}. We note that all the K3 characters are all uncharged 
under both~$u$ and~$z$. Including the fermionic K3 contribution to the chiral partition function~\eqref{zintfin}, 
and summing over the worldsheet spin structures, 
we obtain, in the untwisted sector:
\bea\label{zintUntwst}
Z_{(0,0)}^{j} (\tau; u, z) & = & \half \frac{1}{ \eta^{4}(\tau)} \big(\theta_{00}^{2}(\tau;0) \,\theta_{00}(\tau; u) \,
\theta_{00} (\tau; z + u)   - \theta_{01}^{2}(\tau;0) \, \theta_{01}(\tau; u) \, \theta_{01}( \tau; z+u)
\cr
&  & \qquad \qquad
- \theta_{10}^{2}(\tau;0)  \, \theta_{10} ( \tau;u) \, \theta_{10}( \tau; z +u) \big) \,  \chi_{j} (\tau;u+z) \cr
& = & \frac{1}{\eta^{4}(\tau)} \, \theta_{11}^{2} (\tau;  z/2) \, \theta_{11}^{2} (\tau; z/2+u)  \,  \chi_{j} (\tau; u+z) \, ,
\eea
where we have used, in the second line, the identity $R2$ of \cite{mumford1983tata}.
Note that the only~$j$-dependence is in terms of $\chi_{j}(\t;u)$ which is an overall factor.

In a similar fashion, we obtain the twisted sector partition functions:
\bea
Z_{(0,1)}^{j} (\tau; u, z) & = & \frac{1}{\eta^{4}(\tau)} \, \theta_{11}^{2} (\tau;  z/2) \,
\theta_{10}^{2} ( \tau; z/2+u)  \,  \chi_{j} (\tau; u+z) \, , \cr
Z_{(1,0)}^{j} (\tau; u, z) & = & \frac{1}{\eta^{4}(\tau)} \, \theta_{11}^{2} (\tau;  z/2) \,
\theta_{01}^{2} ( \tau; z/2+u)  \,  \chi_{j} (\tau; u+z) \, , \\
Z_{(1,1)}^{j} (\tau; u, z) & = & \frac{1}{\eta^{4}(\tau)} \, \theta_{11}^{2} (\tau;  z/2) \,
\theta_{00}^{2} ( \tau; z/2+u)  \,  \chi_{j} (\tau; u+z) \, .\nonumber
\eea
These equations are the analog of Equations (3.16) of~\cite{Harvey:2013mda}.
We note that the character~$\chi_{j}(\t,u)$ factors out in Equation~\eqref{zintfin}, and
consequently in all the following expressions.

We still need to be tensor in the remaining bosonic fields from the $K3$ and the $SL(2)$, add the right-movers,
and sum over~$j\in \IZ_{2k}$ to obtain a partition function~$Z(\t;u,\ubar,z,\zbar)$ 
which is integrated over~$(u,\ubar)$. The helicity supertrace~\eqref{helsuptr} is then defined as:
\be \label{eq:sec2_del_operator}
\wh{\chi^{k}_{2}}(\t) \=  \frac{1}{4} \int_{E(\t)} \frac{du_{1} du_{2}}{\t_{2}} \,  
\bigl(\frac{1}{2 \pi i}( \p_{z} -\p_{\zbar}) \bigr)^{2} Z(\t;u,\ubar,z,\zbar) \Big |_{z=\zbar=0} \, . 
\ee
The various intermediate steps proceed exactly as in~\cite{Harvey:2013mda}. In the following section we 
shall display more details of these steps as part of a Hamiltonian analysis of the BPS spectrum of the DSLST. 
The final result for the helicity supertrace is: 
\be\label{chi2hat}
\wh{\chi^{k}_{2}}(\t) \=  \frac{\sqrt{k \t_2}}{4} \int_{E(\t)} \frac{du_{1} du_{2}}{\t_{2}} \, e^{{-\pi k u_{2}^{2} / \tau_{2}}} \;
\frac{\eta(\t)^{6}}{\theta_{1}(\tau; u)}   \, \overline{\th_{1}(\t; u)} \, 
\sum_{\ell\in\IZ_{2k}} |\chi_{\ell}(\t; u)|^{2} \; {\cal Z}^{\rm ell}(K3;\t,u) \, .
\ee

It is convenient to introduce the standard notation 
\be
\v_{-2,1}(\t;u) \= \frac{\theta_{1}(\tau;u)^{2}}{\eta(\t)^{6}} \, ,  \qquad 
\v_{0,1}(\t;u) \= \sum_{i=2,3,4} \frac{\th_{i}(\t;u)^{2}}{\th_{i}(\t;0)^{2}} \, , \qquad P(\t;u) \= \frac{\v_{0,1}(\t;u)}{\v_{-2,1}(\t;u)} \, .
\ee 
The functions $\v_{-2,1}$ and $\v_{0,1}$ are elements of the standard basis of weak Jacobi forms (see (\ref{modelljac}) and \cite{eichler1985theory}). 
The elliptic genus of $K3$ is given by ${\cal Z}^{\rm ell}(K3;\t,u) = 2 \v_{0,1}(\t; u)$. 
The function $P(\t;u)$ is a multiple of the Weierstrass $\wp$-function~$P(\t;u) = -\frac{3}{\pi^{2}} \wp(\t,u)$. 

Using this notation and Equation~\eqref{defstringfn}, we can rewrite~\eqref{chi2hat} as 
(c.f. Eqn.~(3.30) of~\cite{Harvey:2013mda}):
\be \label{chi2int}
\wh{\chi}^{k}_{2}(\t) \=  \int_{E(\t)} \frac{du_{1} du_{2}}{\t_{2}} \,
P(\t;u)  \, \bh_{k}(\t;u)  \, ,
\ee
where
\be \label{defHm}
\bh_{k}(\t;u)  \= \frac{ \sqrt{k\t_{2}} }{2} \; e^{-\pi k u_{2}^{2}/\t_{2}} \, \sum_{\ell \in \IZ_{2k}}
|\vth_{k,\ell} (\t; u/2) - \vth_{k,-\ell} (\t; u/2) |^{2} \, .
\ee
In the above equations \eqref{chi2hat}, \eqref{defHm}, the summed variable $\ell$ takes values in $\IZ_{2k}$. It is clear however, from the definition of the theta functions, that $\wh \vth_{k,0}=\wh \vth_{k,k}=0$ and $\wh \vth_{k, \ell}= - \wh \vth_{k,-\ell}$ are non-zero for $\ell=1,2, \ldots k-1$ 
with
\be
\wh \vth_{k,\ell}(\tau; z)= \vth_{k,\ell}(\tau; z)-\vth_{k,-\ell}(\tau; z)
\ee
so
that only~$k-1$ of the $2k$ summed expressions are distinct and non-zero. 
We can therefore restrict the summation so that one has
\be \label{defHmagain}
\bh_{k}(\t;u)  \= \sqrt{k\t_{2}} \; e^{-\pi k u_{2}^{2}/\t_{2}} \, \sum_{r=1}^{k-1}
|\wh \vth_{k,r}(\tau;u/2)|^{2} \, , 
\ee
We shall use this expression in the following sections. 
Note that the function $\bh_{2}(\t;u)$ is equal to the function denoted
by $H(\t,u)$ in \cite{Harvey:2013mda} after using the identity $\theta_1(\t;u)= - \widehat \vartheta_{2,1}(\tau;u/2)$.

\subsection{ADE DSLST}

There is a natural extension of the above analysis to the more general SCFTs labelled by ADE root systems as follows.
For any divisor $d$ of $k$, define
\be
\bh_{(k,d)}(\t;u)  \= \half \sqrt{k\t_{2}} \; e^{-\pi k u_{2}^{2}/\t_{2}} \, \sum_{r, r' \in \IZ_{2k}}
\wh \vth_{k,r} (\t;u/2) \, \overline{ \wh \vth_{k,r'} (\t;u/2)} \, \O^{(k,d)}_{r, r'}\, ,
\ee
where
\be
\O^{(k,d)}_{r, r'} \= 
\begin{cases}  
1 \qquad \text{if} \quad r+r' =0 \bmod 2d \quad \text{and} \quad r-r' =0 \bmod 2k/d \, , \\
0 \qquad \text{otherwise} \, . 
\end{cases}
\ee
Further, for $Y$ an ADE root system, define $\bh_{Y}(\t;u)$ by 
\be
\bh_{Y}(\t;u)  \=\half \sqrt{k\t_{2}} \; e^{-\pi k u_{2}^{2}/\t_{2}} \, \sum_{r, r' \in \IZ_{2k}}
\wh \vth_{k,r} (\t;u/2) \, \overline{ \wh \vth_{k,r'} (\t;u/2)} \, \O^{Y}_{r, r'}\, ,
\ee
where $\O^{Y}_{r, r'}$ are the $(r, r')$ components of the matrices 
defined\footnote{The matrices~$\O^{Y}$ are the same ones as those appearing in Table~5  of~\cite{Cheng:2013wca}, 
that paper uses the notation $\O_{m}(d)$ for~$\O^{(m,d)}$ .}  
in Table 1.

\begin{table}
\centering
\begin{tabular}{ccc}
\hline
$ Y $& $m(Y)$  &$ \O^{Y}$ \\
\hline
$A_{m-1}$ & $m$ & $ \O^{(m,1)} $		\vspace{0.2em}\\
$D_{{m}/{2}+1}$ & $m$ &  $\O^{(m,1)} + \O^{(m,m/2)} $	\vspace{0.2em}\\
$E_6 $& $12$& $ \O^{(12,1)} + \O^{(12,4)} + \O^{(12,6)}$	\vspace{0.2em}\\
$E_7$ & $18$ &  $\O^{(18,1)} + \O^{(18,6)} + \O^{(18,9)} $	\vspace{0.2em}\\
$E_8$ & $30$ &$\O^{(30,1)} + \O^{(30,6)} + \O^{(30,10)}+ \O^{(30,15)}$
\vspace{0.1em}\\
\hline
\end{tabular}
\caption{\label{ADE1}{Coxeter numbers and  matrices $\Omega^Y$ for ADE root systems Y.
	}}
	 \label{tab:ADE}
\end{table}

In this notation, the function \eqref{defHmagain} is  
\be
\bh_{k} \= \bh_{(k,1)}  \= \bh_{A_{k-1}} \, . 
\ee
Next, define:
\be \label{chi2md}
\wh\chi^{(k,d)}_{2}(\t) \=  \int_{E(\t)} \frac{du_{1} du_{2}}{\t_{2}} \, P(\t;u)   \, \bh_{(k,d)}(\t;u)  \, ,
\ee
and
\be \label{chi2X}
\wh \chi^{Y}_{2}(\t) \=  \int_{E(\t)} \frac{du_{1} du_{2}}{\t_{2}} \, P(\t;u)   \, \bh_{Y}(\t;u)  \, .
\ee

From calculations parallel to the one in the last subsection, the functions~$\wh \chi^{Y}_{2}(\t)$ are 
the second helicity supertraces evaluated in the holographic dual of DSLST of type Y=$A$, $D$ or $E$.

\subsection{The helicity supertraces, their completions and their shadows}

The modular properties of the functions~$\wh\chi^{Y}_{2}(\t)$ can be deduced by following arguments similar to those 
presented in~\cite{Harvey:2013mda}, as we now briefly discuss.

The function~$P(\t;u)$ is a Jacobi form of weight 2 and index 0, i.e. it is invariant under elliptic transformations. We can check
that the function~$\bh_{k}(\t;u)$ is invariant under the full Jacobi group. The measure~$du_{1}du_{2}/\t_{2}$ is also invariant
under the elliptic transformations. This means that the integral~\eqref{chi2int} over the coset~$E(\t)$ is indeed well-defined,
and further, it transforms like a modular form of weight 2. 

The pole of the $P$-function at $u=0$ necessitates some care in the definition of~\eqref{chi2int}. As in~\cite{Harvey:2013mda},
we define the integral as a limit:
\be \label{chimsing}
\wh \chi_{k} (\t) \=   \lim_{\varepsilon \rightarrow 0}  \int_{E^{\varepsilon}} \bh_{k}(\tau;u) \, P(\tau;u) \, \frac{du_{1} du_{2}
}{\tau_2} \, ,
\ee 
where $E^{\ve}(\t)$ is the torus with a small disk of size~$\ve$ centered at the origin removed from it. 
The $\tbar$-derivative of~$\wh \chi_{k} (\t) $ can be computed easily using the trick 
in~\cite{Harvey:2013mda}.  We first notice 
that $\p_{\tbar} \, B(\t; a\t+b) = \dfrac{i}{2\pi k} \, \p_{\ubar}^{2} \, B(\t;u)|_{u=a\t+b}$. Using this heat equation,
the $\tbar$-derivative of~$\wh \chi_{k} (\t) $ reduces to a contour integral around the origin:
\bea \label{Ftbar3}
\p_{\tbar} \, \wh \chi_k(\t)  &\=& \frac{1}{4 \pi k} \oint_{\p D^{\ve}}  \p_{\bar u} \, \bh_{k}(\t;u) \, P(\t;u)   \, \frac{du}{\t_{2}} \, \cr
& = & - \frac{1}{4\pi k \t_{2}} \frac{3}{\pi^{2}} (2 \pi i) \, {\rm Res}_{u \rightarrow 0} \biggl( \partial_{\bar u} \bh_{k}(\tau;u) \frac{1}{u^2} \biggr) \cr
& = &  \frac{3i}{\sqrt{k\tau_2}} \sum_{r \in \IZ_{2k}} \, S_{k,r} (\t) \,  \overline{S_{k,r} (\tau)}\, ,
\eea
where the weight $3/2$ modular forms~$S_{k,r} (\t)$ are defined in \eqref{defSkr}. 
In other words, $\wh \chi_{k}$ is a mixed mock modular form of weight~$2$ and 
shadow~$-\, \dfrac{3}{\sqrt{\pi k}} \sum_{r \in \IZ_{2k}}S_{k,r} (\t) \,  \overline{S_{k,r} (\tau)}$ \footnote{This corrects an error in
the normalization of equations (A.12) and (A.13) of \cite{Harvey:2013mda}.}

It is now also clear from their definitions that $\wh \chi^{(k,d)}$ and  $\wh \chi^{Y}$ are mixed mock modular form of 
weight~$2$ and shadows proportional to 
\be
\sum_{r, r' \in \IZ_{2k}} S_{k,r} (\t) \, \overline{S_{k,r'} (\tau)} \, \O^{(k,d)}_{r,r'}\, , \quad \text{and, respectively} \;
\sum_{r, r' \in \IZ_{2k}} S_{k,r} (\t) \, \overline{S_{k,r'} (\tau)}\, \O^{Y}_{r,r'}\, .
\ee

The integral~\eqref{chimsing} is analyzed in~\cite{MurZag} and explicit expressions for the 
mixed mock modular forms~$\chi^{(k,d)}(\t)$ are found in terms of elementary number-theoretic 
sums as follows. Define, for~$d|k$,
\be
{{\cF}}_2^{(k,d)}(\tau) = \left( d \sum_{\substack{r,s \\ kr>d^2 s>0}} - \frac{k}{d} \sum_{\substack{r,s \\ d^2 r>ks>0}} \right) \, s \, q^{rs} \, ,
\ee
and for $Y$ an ADE root system define $\CF_2^Y$ to be the linear combination of ${\cF}_2^{k,d}$ using the matrices of Table 1 as above. Thus
\bea
{\cF}_2^{A_{k-1}} &=& {\cF}_2^{(k,1)} \, , \\
{\cF}_2^{D_{k/2+1}} &=& {\cF}_2^{(k,1)}+{\cF}_2^{(k,k/2)} \, , \\
{\cF}_2^{E_6} &=& {\cF}_2^{(12,1)}+{\cF}_2^{(12,4)}+{\cF}_2^{(12,6)} \, , \\
{\cF}_2^{E_7} &=& {\cF}_2^{(18,1)}+{\cF}_2^{(18,6)}+{\cF}_2^{(18,9)} \, , \\
{\cF}_2^{E_8} &=& {\cF}_2^{(30,1)}+{\cF}_2^{(30,6)}+{\cF}_2^{(30,10)}+{\cF}_2^{(30,15)} \, .
\eea
The result of~\cite{MurZag} is
\be
\chi^{(k,d)}_{2}(\tau) \= \left( \frac{k}{d} - d  \right) \, E_{2}(\tau) - 24 \, \CF_{2}^{(k,d)}(\tau) \, ,
\ee
which further implies
\be \label{mzres}
\chi^{Y}_{2} (\tau) \= \mathrm{rk}(Y) \, E_{2}(\tau) - 24 \, \CF_{2}^{Y} (\tau) \, .
\ee
Here the second Eisenstein series is given by
\be
E_2(\tau) = 1 -24 \, \sum_{n=1}^\infty \, \sigma_1(n) \, q^n \, , 
\ee
with $\sigma_1(n)$ is the sum of the divisors of $n$.

In the following section we will show, using a Hamiltonian analysis, that the result~\eqref{mzres} 
also arises as the discrete part of the helicity supertrace of the string background described by~(\ref{scft}).
Our Hamiltonian analysis complements the functional integral analysis of the present section, and provides 
us with an independent point of view on the associated physics, as well as a independent derivation 
of the mixed mock modular forms~\eqref{chi2X} that we now briefly summarize. 
We begin with the the identity~\eqref{eq:phi01_decompose}, which can be rewritten in the form 
\be
P(\t;u) = \frac{12} {\mathrm{rk}(Y)} \,\eta(\tau)^3 \,
\Big[ - \mu^{Y}(\tau;z,z,0)  + \frac{1}{3} \sum_{w \in \Pi_2} \mu^{Y}(\tau;w,w,0) \Big] \, , 
\ee
where $\Pi_2  = \{ \frac{1}{2}, \frac{\tau}{2}, \frac{\tau+1}{2} \}$ and $\mu^{Y}(\tau;v,u,w)$ is a multi-variable Appell-Lerch sum defined
in Appendix \ref{sec:Appell_Lerch}. 
Substituting this into \eqref{chi2X}, we obtain
\be \label{chisplit}
\wh \chi_2^Y(\tau)= \int_{E(\tau)} \frac{du_1 du_2}{\tau} \left( \frac{12}{\mathrm{rk}(Y)} \eta(\tau)^3 \left[ \frac{1}{3} \sum_{w \in \Pi_2} \mu^{Y}(\tau;w,w,0)
- \mu^{Y}(\tau;u,u,0) \right] B_Y(\tau;u) \right) \, .
\ee
The first term in square brackets in \eqref{chisplit} is independent of $u$, and the integral of~$B_Y(\tau;u)$ 
can be performed using a slight modification of the technique used in Appendix A of~\cite{Harvey:2013mda} 
to yield
\be \label{intBY}
\int_{E(\tau)} \frac{du_1 du_2}{\tau_2} \, B_Y(\tau;u) = \mathrm{rk}(Y) \, .
\ee
This term, as in the analogous computation for $k=2$ in~\cite{Harvey:2013mda}, 
gives the holomorphic contribution to~$\wh \chi_2^Y(\tau)$, while the second term in square brackets 
in~\eqref{chisplit} provides the non-holomorphic completion. 
From~\eqref{chisplit} and~\eqref{intBY}, we find that the holomorphic contribution
to the second helicity supertrace is given by
\be \label{intres}
\chi_2^Y(\tau) = 4 \, \eta(\tau)^3 \sum_{w \in \Pi_2} \mu^Y(\tau;w,w,0) \, . 
\ee
It is indeed a nontrivial fact that \eqref{intres}  is equal to \eqref{mzres} and gives the weight two mock modular form whose completion is $\wh \chi_{2}^{Y} (\t) $. This
follows from the following argument, that we have also used earlier. 
By the corollary \eqref{cor:mu_mod}, the expression \eqref{intres} for $\chi_2^Y(\tau)$ has a weight two completion
\begin{equation}
\wh \chi_{2}^{Y} (\t) \= \chi_2^Y(\tau) - \frac{3}{\sqrt{\pi k}} \sum_{j,j' \in \ZZ_{2k}} \O^Y_{j,j'} S_{k,j}(\t)
S^*_{k,j'}(\t).
\end{equation}
Comparing this with Equation~\eqref{Ftbar3}, we deduce that both the expressions \eqref{mzres} 
and \eqref{intres} are mixed mock modular forms with the same shadow. 
Therefore their difference has vanishing shadow and must be a holomorphic modular form of weight two. 
But there are no such modular forms, so their difference must vanish.

\section{Evaluation of the BPS index via characters \label{charsec}}
 In this section, we will isolate the contributions to the helicity index, $\chi^{Y}_{2}(\t)$, that arise  from the discrete characters of the $\sl$ affine algebra. In this way we will explicitly verify that modes localized at the tip of the cigar account for the holomorphic part of the helicity index. In other words, we will compute
 \begin{equation}
 \chi^{Y}_{2,dis} (\t) 
 				= \Tr \, (-1)^{F_s} \, q^{L_0 - c/24} \, \qb^{\tilde{L}_0 - \tilde{c}/24} \, 
 							(J_3 - \wt{J}_3)^2,
 \end{equation}
where the trace omits the sum over momentum and winding modes associated with $S^1$ as before but now also omits contributions associated with continuous characters of $\sl$. 

The charge $J^3$ and its right moving counterpart $\wt{J}_3$ are global charges coming from the level $k$ super $\sl$ algebra.
In the path integral formulation of Section~\ref{wsheetsec}, the associated currents are gauged via an integral over the variable $u$. This equates these charges with the left and right momenta of the $Y^u$ boson ($J^3 \to p_L$ and $- \wt{J}^3 \to p_R$).
The level $k$ super $\slc$ algebra is built out of a bosonic $\slc$ algebra at level $k+2$ and two free fermions giving a further level $-2$ contribution. Therefore, $J^3 =\, :\psi^+ \psi^- : \, + \, j^3$, where $\psi^\pm$ are the free fermions and $j^3$ is the global charge associated with the bosonic $\slc$ coset descending from a bosonic parent $\sl$ algebra at level $k+2$. Further details on both the bosonic and the supersymmetric $\slc$ cosets and their characters can be found in \cite{Dixon:1989cg, Sfetsos:1991wn, Israel:2004xj}. See also \cite{Fotopoulos:2004ut} for a nice summary.

The momentum around the cigar corresponds to $J^3 - \tilde{J}^3$ which is the unbroken $U(1)$ of the full $\CN = (4,4)$ $\SCFT$, $\lb SL(2,\IR)_{k}/U(1) \times SU(2)_{k}/U(1) \rb / \Z_k$ \cite{Kounnas:1993ix,  Antoniadis:1994sr}. $J^3 + \tilde{J}^3$, on the other hand, is the winding number which is only conserved modulo $\Z_k$. 

If we put the $\lb SL(2,\IR)_{k}/U(1) \times SU(2)_{k}/U(1) \rb / \Z_k$ and $T_4/\Z_2$ factors together, we get\footnote{See 
Equation 4.29 in \cite{Israel:2004ir}.}:
\begin{align}
\chi^{Y}_{2,dis}(\t) = \half &\sum_{r,s \in \Z_2} \ \half \sum_{a,b \in \Z_2} (-1)^{a+b+a b} \ \half \sum_{\aa,\bb \in \Z_2} (-1)^{\aa+\bb}  
\ Z_{T_4/\Z_2}\bmat{ r & a & \aa \\ s & b & \bb }(\t,\tbar)
 \notag \\
&\times \sum_{2l, 2l' = 0}^{k-2} \,  \sum_{2j'=1}^{k+1} a(2j') \sum_{2p, 2\bar{p} \in \ZZ}
 				\wh{\O}^Y_{2l+1,2l'+1}   \ 
 				C^{2l,k}_{-2p-a} \bmat{ a \\ b}(\t;0)   \ 
 				\bar{C}^{2l',k}_{-2\bar{p}-\aa} \bmat{ \aa \\ \bb}(\tbar;0)
 				\notag \\
\label{eq:helicity_discrete}
&\times \lb  \lp p+ \frac{a}{2} \rp - \lp \bar{p} + \frac{\aa}{2} \rp \rb^2 \,
					 \ch_d^k \bmat{ a \\ b}(j',p;\t,0) \ 
						\bar{\ch}_d^k \bmat{ \aa \\ \bb}(j',\bar{p};\tbar,0),
\end{align}
where the sum over $r,s$ goes over $\Z_2$ twists of the $T_4/\Z_2$ orbifold and the sums over $a, b, \aa, \bb$ perform the GSO projections (for Type IIA string).
The bosonic and fermionic oscillator contributions from $\R_t \times S_1$ factor are canceled by the $(b,c,\beta,\gamma)$ ghosts as before. $\wh{\Omega}^Y_{2l+1,2l'+1}$ is the ADE invariant matrix\footnote{See appendix \ref{sec:su2char} for details.} corresponding to the simply laced root system $Y$ having  Coxeter number $k$. In terms of $\O^Y_{r,r'}$ defined in the table \ref{tab:ADE} we have 
\begin{equation}
\wh{\Omega}^Y_{r,r'} = \Omega^Y_{r,r'} - \Omega^Y_{r,-r'}.
\end{equation}

\begin{itemize}
\item \textbf{$T_4 / \Z_2$ contribution:}
\begin{equation}
Z_{T_4/\Z_2}\bmat{ r & a & \aa \\ s & b & \bb }(\t,\tbar) =
Z_{T_4/\Z_2}^{bos}\bmat{ r \\ s }(\t,\tbar)\
Z_{T_4/\Z_2}^{frm}\bmat{ r & a \\ s & b }(\t) \
\bar{Z}_{T_4/\Z_2}^{frm}\bmat{ r & \aa \\ s & \bb}(\tbar),
\end{equation}
where
\begin{equation}
Z_{T_4/\Z_2}^{frm} \bmat{ r & a \\ s & b }(\t) = \frac{\th\bmat{ a+r \\ -b-s }(\t;0)\  \th\bmat{ a-r \\ -b+s }(\t;0) }{\eta(\t)^2},
\end{equation}
\begin{equation}
Z_{T_4/\Z_2}^{bos}\bmat{ 0 \\ 0 }(\t,\tbar) = \frac{\Theta^{4,4}(\t,\tbar)}{\eta(\t)^4 \  \bar{\eta}(\tbar)^4},
\end{equation}
and for $(r,s) \neq (0,0)$
\begin{equation}
Z_{T_4/\Z_2}^{bos}\bmat{ r \\ s }(\t,\tbar) = \frac{16 \ \eta(\t)^2 \  \bar{\eta}(\tbar)^2}
						{\th\bmat{ 1+r \\ 1-s}(\t;0)\   \th\bmat{1-r \\ 1+s}(\t;0) \     
				\bar{\th}\bmat{1+r \\ 1-s}(\tbar;0) \  \bar{\th}\bmat{1-r \\ 1+s}(\tbar;0)   }.
\end{equation}

\item \textbf{$SU(2)_{k}/U(1)$ contribution:}

The expressions $C^{l,k}_r \bmat{ a \\ b }$ are the supersymmetric $SU(2)/U(1)$ characters. Further details are given in Appendix \ref{sec:su2char}.

\item  \textbf{$SL(2,\IR)_{k}/U(1) $ contribution:}

Discrete $SL(2,\IR)_{k}/U(1) $ characters that appear in $\chi^{Y}_{2,dis}(\t)$ are
\begin{equation}\label{eq:sec3_chi_d_sl2r}
 \ch_d^k \bmat{ a \\ b}(j',p;\t,x) = 
 \frac{q^{\frac{-(j'-1/2)^2 + (p+a/2)^2}{k}} \, y_x^{2(p+a/2)/k}}
 {1+(-1)^b \, y_x \, q^{p+a/2-j'+1/2}} \, 
 \frac{\th_{ab}(\t;x)}{\eta(\t)^3}.
\end{equation}
The $x$ variable counts the $R$ charge of the $\CN=2$ superconformal algebra formed out of the $SL(2,\IR)_{k}/U(1) $ coset. This expression assumes $p - j' \in \ZZ$ and otherwise we take it to be zero.

These are states that descend from the discrete series of representations for $\sl$ algebra, 
$\wh{\mathcal{D}}_j^{\pm}$, by gauging  $J^3$. For these two representations $j' \in \RR^+$ and states at the ground level have $j^3 \in \pm ( j' + \ZZ_0^+ )$.  Note that for $p \geq j'$ we can find a state with
\begin{equation}
\Delta  = \frac{-j' (j'-1) + (p+a/2)^2}{k} +\frac{a^2}{8}  \quad \mbox{and} \quad
R = \frac{2(p+a/2)}{k} + \frac{a}{2},
\end{equation}
at the lowest $q$-power level of $\ch_d^k \bmat{ a \\ b}(j',p;\t,z)$. These are precisely the conformal weight and R-charge for the $\slc$ primary $V^{sl,\eta = a/2}_{j=j'-1,m =p}$. This notation for the vertex operator follows \cite{Chang:2014jta}.

One last factor that appears in the equation (\ref{eq:helicity_discrete}) is\footnote{
This factor reflects an ambiguity in the character decomposition as 
\begin{equation}
\ch_d^k \bmat{ a \\ b}\lp 1/2,n+1/2;\t,x \rp  
+ \ch_d^k \bmat{ a \\ b}\lp (k+1)/2,n+(k+1)/2;\t,x \rp 
\end{equation}
can be rewritten in terms of a continuous character (see \cite{Eguchi:2004yi}). In any case, we will see that 
the $j' = 1/2$ and $(k+1)/2$ terms do not contribute to $\chi^{Y}_{2,dis}(\t)$.
} 
\begin{equation}
a(2j') = \begin{cases}
         1/2 &\mbox{if } 2j'=1 \mbox{ or } j'=k+1 \, , \\
         1 &\mbox{if } 2j'=2,\ldots,k.
        \end{cases}
\end{equation}

\end{itemize}

To account for the $\lb  \lp p+ \frac{a}{2} \rp - \lp \bar{p} + \frac{\aa}{2} \rp \rb^2 $ of the equation (\ref{eq:helicity_discrete}) we will compute a sum in which we replace 
$\lb  \lp p+ \frac{a}{2} \rp - \lp \bar{p} + \frac{\aa}{2} \rp \rb^2 $ factor with 
\begin{equation}\label{eq:sec3_exp_replacement}
e^{2 \pi \i z (p-j' + (1+a)/2)} \  e^{- 2 \pi \i \zbar (\bar{p}-j' + (1+\aa)/2)}.
\end{equation}
We will finally get back to $\chi^Y_{2,dis}$ by acting
\begin{equation}\label{eq:sec3_z_diff_operator}
\lb \frac{1}{2 \pi \I} \lp \partial_z + \partial_{\bar{z}} \rp \rb^2 \Bigg|_{z=\bar{z}=0} \,
\end{equation}
on the expression we get.\footnote{
The sign difference compared to \eqref{eq:sec2_del_operator} is because we are keeping track of $J_3$ and $\wt{J}_3$ here instead of $p_L$ and $p_R$ as we did there.
} Note that $j' - \half$ factors cancel once we get to $\chi^Y_{2,dis}$; however, we have included them in the equation (\ref{eq:sec3_exp_replacement}) to get nice modular properties in the $z$ dependent function we will compute.

Let us start our computation by focusing on
\begin{equation}\label{eq:sec3_Bfnc_1}
\sum_{2p \in \ZZ} C^{2l,k}_{-2p-a} \bmat{ a \\ b}(\t;t)  \ 
				 \ch_d^k \bmat{ a \\ b}(j',p;\t,x) \ 
				 y_z^{p-j' + (1+a)/2}  \, .
\end{equation}
We first note that 
$C^{2l,k}_{-2p-a} \bmat{ a \\ b}(\t;t) = (-1)^{a b}\,  C^{2l,k}_{2p+a} \bmat{ a \\ b}(\t;-t)$. Since $ \ch_d^k \bmat{ a \\ b}(j',p;\t,x)$ is zero unless $p - j' \in \ZZ$ we define $n = p - j' + a$ and the expression in (\ref{eq:sec3_Bfnc_1}) becomes
\begin{equation}
(-1)^{a b} \sum_{n \in \ZZ} C^{2l,k}_{2n+2j'-a} \bmat{ a \\ b}(\t;-t)  \ 
				 \ch_d^k \bmat{ a \\ b}(j',j'+n-a;\t,x) \ 
				 y_z^{n + (1-a)/2}  \, ,
\end{equation}
or, using the equation (\ref{eq:sec3_chi_d_sl2r} )for $ \ch_d^k \bmat{ a \\ b}(j',p;\t,x)$, 
\begin{equation}
 (-1)^{a b} \sum_{n \in \ZZ} C^{2l,k}_{2n+2j'-a} \bmat{ a \\ b}(\t;-t)  \ 
				  \frac{q^{\frac{-\lp j-\frac{1}{2} \rp^2 +\lp j'+n-\frac{a}{2} \rp^2}{k}} 
				  						\, y_x^{2 \lp j'+n-\frac{a}{2} \rp/k}}
 {1+(-1)^b \, y_x \, q^{n+\frac{1-a}{2}}} \, 
 \frac{\th_{ab}(\t;x)}{\eta(\t)^3} \ 
				 y_z^{n + \frac{1-a}{2}}  \,.
\end{equation}
If we introduce variables $v,u,w$ by
\begin{equation}
x=v, \quad t = v-w, \quad z = u+\lp 1 -\frac{2}{k} \rp w - v \, , 
\end{equation}
we obtain
\begin{align}
 (-1)^{a b} \, \frac{\th_{ab}(\t;v)}{\eta(\t)^3} 
 \sum_{n \in \ZZ}  y_v^{(2j'-1)/k} \ 
				  &\frac{q^{\lp n+\frac{1-a}{2} \rp^2/k} \, q^{(2j'-1) \lp n+\frac{1-a}{2} \rp/k}  }
 {1+(-1)^b \, y_v \, q^{n+\frac{1-a}{2}}} \, \notag \\ 
				 &\times \, \lp \frac{y_u y_w^{1-2/k}}{y_v^{1-2/k}} \rp^{n+\frac{1-a}{2}}  \,
				 C^{2l,k}_{2n+2j'-a} \bmat{ a \\ b}(\t;w-v) .
\end{align}
This is exactly equal to
\begin{align}
(-1)^{a b} &\ B^{k,2l+1,2j'-1}_{a b}(\tau; v,u,w) & \notag \\
&=  (-1)^{a b} \ B^{k,2l+1,2j'-1}_{a b} \lp \tau; x,z+\frac{2}{k}\, x + \lp 1 - \frac{2}{k} \rp t,
	x-t   \rp.
\end{align}
The expression $B^{k,j,j'}_{a b}$ is defined in Appendix \ref{sec:Riemann_reln} as 
\begin{equation}
B^{k,j,j'}_{a b}(\tau; v,u,w) =
\frac{\th_{ab}(\tau; v)\, \th_{ab}(\tau;u)}{\eta(\t)^3} \   
 \mu^{k,j,j'} \lp \tau; v + \t_{ab}, u +  \t_{ab}, w \rp,
\end{equation}
where $\tau_{ab} \equiv (a-1) \t/2 + (b-1)/2$ and
\begin{equation}
\mu^{k,j,j'}(\tau;v,u,w) = \frac{y_v^{j'/k}}{\theta_{11}(\tau;u)}
        \sum_{n \in \ZZ}
        \frac{q^{n^2/k} \, q^{j'n/k}}{1-y_v \, q^{n}} \lp \frac{y_u\, y_w^{1-2/k}}{y_v^{1-2/k}} \rp^n
        C^{j-1,k}_{2n + j'}\bmat{1 \\ 1}(\tau;w-v).
\end{equation}
Some algebraic and modular properties of $\mu^{k,j,j'}(\tau;v,u,w)$ are worked out in detail in Appendix \ref{sec:Appell_Lerch}. Moreover, as shown in Appendix \ref{sec:Riemann_reln}, 
$B^{k,j,j'}_{a b}(\tau; v,u,w)$ satisfies analogues of the Riemann relations for theta functions.
This will be an important ingredient in simplifying $\chi^{Y}_{2,dis}(\t)$. Moreover, Appendix \ref{sec:Riemann_reln} 
tells us that $B^{k,j,j'}_{a b}(\tau; v,u,w)$ is an entire function of $v$, $u$ and $w$ and $B^{k,j,j'}_{11}(\tau; 0,0,0) =  \delta_{j,j'}$.

To summarize this discussion, our goal now is to compute
\begin{align}
\half &\sum_{r,s \in \Z_2} \    
\frac{Z_{T_4/\Z_2}^{bos}\bmat{ r \\ s }(\t,\tbar)}{\eta(\tau)^2 \bar{\eta}(\tbar)^2}
\sum_{l, l' = 1}^{k-1} \,  \sum_{j'=0}^{k} a(j'+1)  \, \wh{\O}^Y_{l,l'}   \ 
 \notag \\
&\times\half \sum_{a,b \in \Z_2} (-1)^{a+b}  \ 
 				\eta(\tau)^2 \, Z_{T_4/\Z_2}^{frm}\bmat{ r & a \\ s & b }(\t) \ 
 				B^{k,l,j'}_{a b}(\tau; 0,z,0)
 				\notag \\
\label{eq:sec3_helicity_exp}
&\times \half \sum_{\aa,\bb \in \Z_2} (-1)^{\aa+\bb+\aa \bb} \ 
 \bar{\eta}(\tbar)^2 \, 
						\bar{Z}_{T_4/\Z_2}^{frm}\bmat{ r & \aa \\ s & \bb}(\tbar) \ 
						\bar{B}^{k,l',j'}_{\aa \bb}(\tbar; 0,\zbar,0)
\end{align}
using the machinery developed in the appendices. Note that we have redefined 
$2l+1, 2l'+1, 2j' -1 \to l,l',j'$ compared to the equation (\ref{eq:helicity_discrete}).

\noindent \textbf{GSO Projections:}
We will start simplifying the equation (\ref{eq:sec3_helicity_exp}) by performing the GSO projections first, i.e. by computing its second and third lines:
\begin{equation}
\half \sum_{a,b \in \Z_2} (-1)^{a+b}  \ 
 				\th\bmat{ a+r \\ -b-s }(\t;0)\  \th\bmat{ a-r \\ -b+s }(\t;0) \ 
 				B^{k,l,j'}_{a b}(\tau; 0,z,0) \, , 
\end{equation}
and
\begin{equation}
\half \sum_{\aa,\bb \in \Z_2} (-1)^{\aa+\bb+\aa \bb} \ 
 \bar{\th}\bmat{ \aa+r \\ -\bb-s }(\tbar;0)\  \bar{\th}\bmat{ \aa-r \\ -\bb+s }(\tbar;0) \ 
						\bar{B}^{k,l',j'}_{\aa \bb}(\tbar; 0,\zbar,0).
\end{equation}

\noindent \textbf{Untwisted Sector (r,s)=(0,0):}

First we study the holomorphic side, keeping in mind that $\th_{11}(\t;0)= 0$:
\begin{align}
 \half & \Big( \th_{00}(\t;0) \, \th_{00}(\t;0) \, B^{k,l,j'}_{0 0}(\tau; 0,z,0)
 - \th_{01}(\t;0) \, \th_{01}(\t;0) \, B^{k,l,j'}_{0 1}(\tau; 0,z,0)  \notag\\
 &- \th_{10}(\t;0) \, \th_{10}(\t;0) \, B^{k,l,j'}_{1 0}(\tau; 0,z,0)
 \Big).
\end{align}
According to the identity $(\wt{R5})$ of Section \ref{sec:Riemann_reln} this is equal to
\begin{equation}
 \th_{11}(\t;z/2) \, \th_{11}(\t;-z/2) \, B^{k,l,j'}_{1 1}(\tau; -z/2,z/2,0).
\end{equation}
Similarly, for the anti-holomorphic side we get
\begin{equation}
\bar{\th}_{11}(\tbar;\zbar/2) \, \bar{\th}_{11}(\tbar;-\zbar/2) \, 
\bar{B}^{k,l',j'}_{1 1}(\tbar; -\zbar/2,\zbar/2,0) .
\end{equation}

Due to the $\th_{11}$ terms, $(r,s) = (0,0)$ contribution to the equation (\ref{eq:sec3_helicity_exp}) behaves as $z^2 \, \zbar^2$ as $z,\zbar \to 0$. Therefore the untwisted sector does not give any contribution to $\chi^Y_{2,dis}$ which we get after the action of the operator (\ref{eq:sec3_z_diff_operator}).

\noindent\textbf{Twisted Sectors:}

We start with the $(r,s)=(1,0)$ sector. From the holomorphic side we get 
\begin{align}
 \half & \Big( \th_{10}(\t;0) \, \th_{10}(\t;0) \, B^{k,l,j'}_{0 0}(\tau; 0,z,0)
 - \th_{00}(\t;0) \, \th_{00}(\t;0) \, B^{k,l,j'}_{1 0}(\tau; 0,z,0)  \notag\\
 &- \th_{01}(\t;0) \, \th_{01}(\t;0) \, B^{k,l,j'}_{1 1}(\tau; 0,z,0)
 \Big)\notag \\
 &\quad  = 
 - \, \th_{11}(\t;z/2) \, \th_{11}(\t;-z/2) \, B^{k,l,j'}_{0 1}(\tau; -z/2,z/2,0)
\end{align}
using the identity $(\wt{R13})$.
We can similarly employ the identity $(\wt{R11})$ and find the contribution of the anti-holomorphic side as:
\begin{equation}
- \, \bar{\th}_{01}(\tbar;\zbar/2) \, \bar{\th}_{01}(\tbar;-\zbar/2) \, 
\bar{B}^{k,l',j'}_{1 1}(\tbar; -\zbar/2,\zbar/2,0) .
\end{equation}
The crucial point is that the only nontrivial contributions after the action of 
\begin{equation}
\lb \frac{1}{2 \pi \I} \lp \partial_z + \partial_{\bar{z}} \rp \rb^2 \Bigg|_{z=\bar{z}=0} 
\end{equation}
come from 
$\partial_z^2$ acting on the $\th_{11}(\t; z/2) \, \th_{11}(\t; -z/2)$ term in the holomorphic contribution. Taking derivatives and setting $z=0$, these two theta functions give $\half \eta(\t)^6$. Setting
$z$ and $\zbar$ to zero for the rest, we get the $(r,s) = (1,0)$ contribution to 
$\chi^Y_{2,dis}$ as
\begin{align}
\half \, 
&\frac{Z_{T_4/\Z_2}^{bos}\bmat{ 1 \\ 0 }(\t,\tbar)}{\eta(\tau)^2 \bar{\eta}(\tbar)^2}
\sum_{l, l' = 1}^{k-1} \,  \sum_{j'=0}^{k} a(j'+1)  \, \wh{\O}^Y_{l,l'}  \,
\lp - \, \frac{\eta(\t)^6}{2} B^{k,l,j'}_{0 1}(\tau; 0,0,0) \rp 
\notag \\
&\quad \times \lp - \, \bar{\th}_{01}(\tbar; 0) \, \bar{\th}_{01}(\tbar; 0) \, 
\bar{B}^{k,l',j'}_{1 1}(\tbar; 0,0,0)   \rp.
\end{align}
Substituting $Z_{T_4/\Z_2}^{bos}\bmat{ 1 \\ 0 }$ and noting that 
$\bar{B}^{k,l',j'}_{1 1}(\tbar; 0,0,0)$ is $\delta_{l', j'}$ in the relevant range we can rewrite this as
\begin{equation}
4 \,  \eta(\t)^3 \sum_{l,l'=1}^{k-1} \wh{\O}^Y_{l,l'} \ \mu^{k,l,l'}(\t;\t/2,\t/2,0).
\end{equation}

We can repeat the same arguments for the other two twisted sectors as well. Using Riemann-like relations $(\wt{R8})$ and $(\wt{R9})$ we find the contribution of $(r,s) = (0,1)$ sector to be
\begin{equation}
4 \,  \eta(\t)^3 \sum_{l,l'=1}^{k-1} \wh{\O}^Y_{l,l'} \ \mu^{k,l,l'}(\t;1/2,1/2,0).
\end{equation}
$(\wt{R15})$ and $(\wt{R16})$, on the other hand, gives the $(r,s) = (1,1)$ sector contribution to be
\begin{equation}
4 \,  \eta(\t)^3 \sum_{l,l'=1}^{k-1} \wh{\O}^Y_{l,l'} \ 
\mu^{k,l,l'}\lp \t;\frac{\t+1}{2},\frac{\t+1}{2},0 \rp.
\end{equation}

In summary, the result of this section is that 
\begin{equation}
\chi_{2,dis}^Y(\tau) = 4 \eta(\tau)^3 \sum_{w \in \Pi_2} 
\sum_{l,l'=1}^{k-1} \wh{\O}^Y_{l,l'}  \mu^{k,l,l'}(\tau;w,w,0),
\end{equation}
where $\Pi_2  = \{ \frac{1}{2}, \frac{\tau}{2}, \frac{\tau+1}{2} \}$. We notice
 that this is exactly
$\chi_{2}^Y(\tau)$,  the holomorphic part of the helicity supertrace.

\section{Connections to Umbral Moonshine}

The upshot of the preceding two sections is that for each choice of $k$ and an $ADE$ root system $Y$ with Coxeter number $k$, the
second helicity trace $\widehat \chi_2^Y$ is the completion of a mixed mock modular form of weight $2$ with shadow proportional to
\be
\sum_{r, r' \in \IZ_{2k}}S_{k,r} (\t) \,  \overline{S_{k,r'} (\tau)}\, \O_{Y}^{rr'}\, .
\ee
The holomorphic part of $\widehat \chi_2^Y$ is a mixed mock modular form of weight two. In Section \ref{wsheetsec} we found that the holomorphic
part of $\wh \chi_2^Y $ is given by
\be \label{chiy2}
\chi^{Y}_{2} (\tau) \= -\text{rk}(Y) \, E_{2}(\tau) + 24 \, \CF_{2}^{Y}(\tau) \, 
\ee
or by
\be
\chi_{2}^Y(\tau) = 4 \eta(\tau)^3 \sum_{w \in \Pi_2} \mu^Y(\tau;w,w,0) \, ,
\ee
where the latter expression is indeed equal to the contribution to $\wh{\chi}_2^Y$ from the discrete characters of the SCFT as we worked out in Section \ref{charsec}.
Although it is not obvious, these two expressions are in fact identical. 
We proved their equality in Section \ref{wsheetsec} by asserting that they are 
weight two mixed mock modular forms with equal shadows. That means
their difference is a true modular form of weight two for $SL(2,\ZZ)$, but there are no such forms so their difference must vanish.

We now want to compare our result for $\chi^Y_2$ to the weight two mixed mock modular forms that appear in Umbral Moonshine.
Recall that in \cite{Cheng:2013wca} an instance of Umbral Moonshine was associated to the root system $X$ of each Niemeier lattice with
a non-vanishing root system. The root systems $X$ are uniquely characterized by having components consisting of ADE root systems with equal
Coxeter numbers $m(X)$ and with total rank $24$. For each $X$ there exists a vector-valued weight $1/2$ mock modular form $H^X_r$, 
$ r=1, \cdots m(X)$
which exhibits moonshine for a finite group $G^X= {\rm Aut} (L^X)/W_X$ where $L^X$ is the Niemeier lattice associated to $X$, ${\rm Aut} (L^X)$ its 
automorphism group and $W_X$ the Weyl group of $X$.  For each $X$ there exists a weight one, index $m(X)$ mock Jacobi form given by
\be
\psi_{1,m}^X(\tau,z) = \sum_r H^X_r (\t) \, \widehat \vartheta_{m,r} (\t,z) \, .
\ee
The first coefficient in the Taylor series expansion of $\psi_{1,m}^X$ about $z=0$ is given by
\be
\chi_2^X = \sum_r H^X_r S_{m,r} \, .
\ee

While the Niemeier root lattices $X$ have an ADE classification, it is distinct from the ADE classification of the the SCFT in equation (\ref{scft})
since a mixture of $ADE$ components is allowed for $X$ while the classification of the SCFT only allows a single, distinct, A,D or E component
\footnote{Attempts to combine the  partition functions of SCFT's with different $ADE$ components leads either to partition functions that  do not have integer multiplicities for states or partition functions that correspond to SCFT's that do not have a unique vacuum state.}. In comparing our results to those of Umbral Moonshine
we will therefore consider only $X$ which are powers of single $A$, $D$, or $E$ components, namely 
\be \label{purexes}
X=A_1^{24}, A_2^{12}, A_3^8, A_4^6, A_6^4, A_8^3, A_{12}^2, A_{24}, D_4^6, D_6^4, D_8^3, D_{12}^2, D_{24}, E_6^4, E_8^3 \, .
\ee
To each of these $X's$ we associate a single $ADE$ root system $Y$ via $X=Y_{{\rm rk}(Y)}^{24/{\rm rk}(Y)}$ so that
\be
Y=A_1, A_2, A_3, A_4, A_6, A_8, A_{12}, A_{24}, D_4, D_6, D_8, D_{12}, D_{24}, E_6, E_8
\ee
respectively. For these cases the results of \cite{Cheng:2013wca} show that $\chi_2^X$ is a weight two mixed mock modular form with
shadow proportional to
\be
 \sum_{r, r' \in \IZ_{2k}} S_{k,r} (\t) \, \overline{S_{k,r'} (\tau)}\, \O^{Y}_{rr'}\, .
\ee

Since the shadows of $\chi_2^X$ and $\chi_2^Y$ are equal up to a constant, $\chi_2^X$ and $\chi_2^Y$ must also be equal up to
a constant. We determine the constant by the following argument and then verify this by direct comparison of $q$ expansions. In
DSLST the massless fields consist of $\mathrm{rk}(Y)$ fields in a multiplet which is equivalent to a vector multiplet of $N=2$ supersymmetry
in four dimensions. Each vector multiplet contributes $+1$ to the second helicity supertrace, so the leading term in $\chi_2^Y(\tau)$ should
be $\mathrm{rk}(Y) q^0$. On the other hand, the mock modular forms of Umbral Moonshine are characterized by the fact that the only
polar term occurs in $H^X_r$ for $r=1$ with $H^X_1= -2 q^{-1/4m(X)} + O(q^{1-1/4m(X)})$. As a result the leading term in $\sum_r S_{m,r} H^X_r$
is $-2$ and we therefore deduce that 
\be \label{xyrel}
\chi_2^Y(\tau) =  - \, \frac{\mathrm{rk}(Y)}{2} \, \chi_2^X
\ee
which can also be verified by comparing the $q$ expansions of both sides using the explicit
$q$-series for $H^X_r$ given in \cite{Cheng:2013wca}. For $X=A_1^{24}$ with $Y=A_1$ and $\mathrm{rk}(Y)=1$ the relation
(\ref{xyrel}) reduces to that found in \cite{Harvey:2013mda}.

The holomorphic part of the second helicity supertrace $\chi_2^Y$ can of course be computed for any 
single $ADE$ root lattice $Y$ and will be a mixed mock modular form of weight two, so it is natural to ask 
whether there is anything special about those $Y$ which appear in Umbral Moonshine via Niemeier 
lattices $X$. For the Niemeier lattices with pure $ADE$ root systems built out of a component~$Y$, 
the rank 24 condition for the rank of the Niemeier lattices implies that $\mathrm{rk}(Y)$ divides $24$, and 
therefore we can write $X=Y^{24/\mathrm{rk}(Y)}$. We will call such cases the Umbral $Y$. 
For these cases we know that $\chi_2^X$ has a $q$ expansion with coefficients that are even integers. 
From the formula~\eqref{chiy2}, we deduce that all the coefficients in the $q$ expansion of $\chi_2^Y$ 
are divisible by $\mathrm{rk}(Y)$. 

As remarked above, this divisibility can be understood for the coefficient of $q^0$ since 
there are $\mathrm{rk}(Y)$ massless states contributing to~$\chi_2^Y$, but it is not clear
why this should be true for the coefficients of higher powers of $q$ that count massive BPS states. 
In fact this divisibility by $\mathrm{rk}(Y)$ does not occur for non-Umbral choices of~$Y$. 
For example, while for~$Y=A_{3}$, the $q$ expansion of $\chi_2$ is given by:
 \be
\chi_2^{Y=A_3}= 3 - 96 \, q - 288 \, q^2 - 384 \, q^3 - 576 \, q^4 - 360 \, q^5 + \cdots \, 
\ee
where all coefficients are divisible by~3, the $q$ expansion of $\chi_2^{Y=A_7}$ which is a 
non-Umbral example is given by
\be
\chi_2^{Y=A_7}= 7 - 192 \, q - 576 \, q^2 - 768 \, q^3 - 1344 \, q^4 - 1152 \, q^5 + \cdots \, 
\ee
where the coefficients $7, 192$ are relatively prime so there is no common factor that can be factored out.
For the convenience of the reader we give low order terms in the $q$ expansions of the $\chi_2^X$ for the choices of $X$ appearing in the list (\ref{purexes})
\begin{table}[h]
\centering
\begin{tabular}{ c c c c c c c c c c c }
\toprule
$X \backslash~ n $& $0 $& $1$ &$ 2 $& $3$ & $4$ & $5$ & $6 $& $7$ & $8$ & $9$ \\
\hline
$A_1^{24}$ & $-2$ & $96$  & $192$  & $144$  & $384$  & $240$  & $768$  & $336$  & $768$  & $720$  \\
$A_2^{12}$ &$ -2$ &$ 72$ &$ 216$ & $216$ & $336$ & $240$ & $648$ &$ 336$ & $816$ & $648$ \\
$A_3^8 $ & $-2$ & $64$ &   $192$ &  $256$ & $ 384$ & $240$ & $608$ & $336$ & $768$ & $624$ \\
$A_4^6$ &   $-2 $ & $60$ &  $180$ & $240$ & $420$ & $300$ & $588$ & $336$ & $744$ & $612$   \\
$A_6^4$ &  $ -2$ & $56$  & $168$  & $224$  & $392$  & $336$  & $672$  & $392$  & $720$  & $600$    \\
$A_8^3$ &   $ -2$ & $54$  & $162$  & $216$  & $378$  & $324$  & $648$  & $432$  & $810$  & $648$      \\
$A_{12}^2$ &  $ -2$ & $52$  & $156$  & $208$  & $364$  & $312$  & $624$  & $416$  & $780$  & $676$   \\
$A_{24} $& $ -2$ & $50$  & $150$  & $200$  & $350$  & $300$  & $600$  & $400$  & $750$  & $650$     \\
$D_4^6$ &   $-2$ & $36$  & $192$  & $252$  & $384$  & $372$  & $612$  & $336$  & $768$  & $540$    \\
$D_6^4$ &    $-2$ & $40$  & $120$  & $248$  & $384$  & $360$  & $616$  & $472$  & $768$  & $704$    \\
$D_8^3$ &    $ -2$ & 42  & 126  & 168  & 384  & 354  & 618  & 462  & 768  & 696  \\
$D_{12}^2$ & $ -2$ & 44  & 132  & 176  & 308  & 264  & 620  & 452  & 768  & 688   \\
$D_{24}$ & $-2$  &$46$  &$138$  &$184$  &$322$  &$276$  & $552$ &$368$  &$690$  & $598$             \\
$E_6^4$ & $-2$ & 16  & 128  & 216  & 352  & 376  & 648  & 488  & 800  & 648     \\
$E_8^3$ &  $-2$ & $ - 6$  & 78  & 102  & 246  & 300  & 474  & 384  & 768  & 648   \\
\bottomrule
\end{tabular}
\caption{Coefficient of $q^n$ in the $q$ expansion of $\chi_2^X$ for pure $ADE$ Umbral $X$}
\label{tab:xqexp}
\end{table}

\section{Conclusions and Open Problems}

In this paper we performed a path integral evaluation of the second heliticty supertrace of the SCFT describing the holographic dual of
ADE little string theory on $K3$ and showed that the answer has the form of the completion of a mixed mock modular form of weight $2$.
We identified the holomorphic part of the supertrace as the contribution from localized states on the non-compact, ``cigar" component of the SCFT
by using a Hamiltonian analysis to sum only over these normalizable contributions and showed that this matched the holomorphic part of the
previous computation.  

Our results have a very suggestive connection to the mock modular forms appearing in Umbral Moonshine in that the $ADE$ label
$Y$ which are Umbral in the sense that $X=Y^{24/\mathrm{rk}(Y)}$ is the root lattice of a Niemeier lattice lead to second helicity supertraces
$\chi_2^Y$ with special divisibility properties.  In trying to refine this connection the most pressing problem is to identify the symmetries of the SCFT (\ref{scft}) that preserve space-time supersymmetry and in particular to understand to what degree these extend the known such symmetries that act on the K3 component and have a relation to the Umbral symmetry groups $G^X$. 

It is clear that we have not identified a physical criterion that would single out precisely the Umbral choices of $Y$. The computation we have
done goes through for any choice of $Y$ and we see no sign of any inconsistency or instability of the theory for non-Umbral choices.
This problem also afflicts other related attempts to find a relation between Umbral Moonshine and explicit non-compact Conformal Field
Theories. For example, the decomposition of the elliptic genus given in the Appendix of \cite{Cheng:2014zpa} in terms of Umbral Jacobi forms
$\psi^X(\tau,z)$ can be extended to an infinite class of non-Umbral Jacobi forms.

To some degree our results are inevitable mathematical consequences of modular invariance in the context
of non-compact SCFT. It is nonetheless interesting that the precise mock modular forms appearing have such a close relation
to mock modular forms that have appeared in to context of Umbral Moonshine. Our results makes these mock modular forms manifest
in a physical system.  The string theory computation also leads to other findings 
such as the divisibility condition $rk(Y)|24$ -- this was not guaranteed by modular invariance.  The significance of the divisibility criterion
is not entirely clear to us. If one divides $\chi_2^Y$ by $rk(Y)/24$ then one still obtains a possible second helicity supertrace in that it still
has integer coefficients in the $q$ expansion of the holomorphic part. The question then is whether it is possible to associate this supertrace
to some physical background in string theory. One possibility is that it should be interpreted as a possible helicity supertrace of the
long sought for theory of a single fivebrane. The problem with this idea is that it is not clear why a background corresponding to $k$
fivebranes in the double scaling limit should factorize in this way.

We noted earlier that one can view the SCFT ( \ref{scft}) as describing the behavior of string theory near a singularity $\CC^2/\Gamma$
with $\Gamma$ as finite subgroup of $SU(2)$. We can further view this singularity as a local singularity in a global $K3$ manifold
which can develop singularities of type $ADE$ with rank up to $20$. If we do this then we are considering type II string theory
on $K3 \times K3$ with the second $K3$ factor developing a local $ADE$ singularity and then focusing attention at the behavior
near the singularity. In the global context this background has a tadpole in the antisymmetric tensor field $B$ \cite{Sethi:1996es} due to a space-time
coupling of the form 
\be
- \int B \wedge X_8(R) \, .
\ee
Furthermore, for geometric compactification on an eight-dimensional manifold $M$ the integral of $X_8$ is proportional to the
Euler number of $M$, 
\be
\int_{M} X_8(R)=  \frac{\chi(M)}{24} \, ,
\ee
As a result, for $K3 \times K3$ which has Euler number $24^2$, there is a tadpole which can be cancelled by the addition of $24$
fundamental strings tangent to the remaining $\RR^{1,1}$ (or $R \times S^1$ is we compactly the spatial direction as discussed earlier).
As we approach a singularity in $K3$ to obtain the SCFT used here we lose the global structure and there is no tadpole condition since the
flux can escape off to infinity. Nonetheless, the SCFT we have used does {\it not} have a flux of the $B$ field at infinity and thus we expect
that it describes a situation in which the local contribution to $X_8(R)$ coming from the singularity is cancelled by a contribution from
fundamental strings. It would be interesting to understand the role of these background fundamental strings in more detail.

\section*{Acknowledgments}

We thank Don Zagier for valuable discussions. JH and CN acknowledge the support of NSF grant 1214409.


\newpage
\appendix

\section{Definitions and Conventions \label{Conventions}}

\noindent We use variables: $q\equiv e^{2 \pi \i \tau}$ for $\tau \in \HH$, $y \equiv e^{2 \pi \i z}$ for $z \in \IC$, and similarly we use $y_z$, $y_x$, $y_u$, ...etc. for $e^{2 \pi \i z}$,
$e^{2 \pi \i x}$, $e^{2 \pi \i u}$, ..., respectively, where $z,x,u,\ldots \in \IC$.

\subsection{Basic modular, elliptic and Jacobi functions} \label{modelljac}

\noindent a) The Dedekind $\eta$ function is defined as:
\be
\eta(\tau)  \equiv q^{1/24} \prod_{n=1}^\infty (1-q^n) = q^{1/24} \sum_{n \in \ZZ} q^{n(3n-1)/2}.
\ee
Under the generators of modular transformations $\eta$ transforms as
\begin{align}
\eta(\tau+1) &= e^{\i \pi/12} \, \eta(\tau), \\
\eta(-1/\tau) &= e^{-\i \pi/4} \, \tau^{1/2} \, \eta(\tau).
\end{align}

\noindent b) The Jacobi theta function $\theta_{ab}$ with $a,b \in \{0,1\}$ is:
\be
\theta_{ab}(\tau; z) \equiv 
\theta \bmat {a\\b}(\tau; z) \equiv 
\sum_{n \in \ZZ} q^{(n+a/2)^2/2} \, e^{2 \pi i (z+b/2)(n+a/2)}.
\ee
We use the notation $ \theta \bmat {a\\b}(\tau,z)$  for $a,b \notin \{0,1\}$ as well using the infinite series above.

\noindent i) Product formulae and other conventions:
\begin{align} 
 \theta_{00}(\tau;z) &= \theta_3(\tau;z) = \prod_{n=1}^\infty (1-q^n) (1+y q^{n-1/2})(1+y^{-1} q^{n-1/2})  \notag \\
&= \sum_{m \in \ZZ} q^{m^2/2} y^m,  
 \end{align}
 \begin{align}
 \theta_{01}(\tau;z)  &=  \theta_4(\tau;z) = \prod_{n=1}^\infty (1-q^n) (1-y q^{n-1/2})(1-y^{-1} q^{n-1/2})  \notag \\
&= \sum_{m \in \ZZ} (-1)^m q^{m^2/2} y^m,  
 \end{align}
 \begin{align}
 \theta_{10}(\tau;z)  &=   \theta_2(\tau;z) =q^{1/8} y^{1/2} \prod_{n=1}^\infty (1-q^n)(1+y q^n)(1+y^{-1} q^{n-1})   \notag \\
   &= \sum_{m \in \ZZ} q^{(m+1/2)^2/2}
y^{m+1/2},
 \end{align}
 \begin{align} 
  \theta_{11}(\tau;z)  &= -\i\, \theta_1(\tau;z) = \i\, q^{1/8} y^{1/2} \prod_{n=1}^\infty (1-q^n) (1-y q^n) (1-y^{-1} q^{n-1})    \notag \\
   &= \i \sum_{m \in \ZZ} (-1)^m q^{(m+1/2)^2/2} y^{m+1/2}.
\end{align}
The conventions here for $\theta_{00}, \theta_{01}, \theta_{10}, \theta_{11}$ are consistent with \cite{mumford1983tata} and the conventions for
$\theta_i$, $i=1,2,3,4$ are consistent with \cite{Cheng:2012tq}.

\noindent ii) Transformation under shifts :
\begin{align}
\theta_{00} \lp \tau;  z+\frac{a}{2}\, \tau + \frac{b}{2} \rp &= q^{-a^2/8}\, y^{-a/2}\, e^{-i \pi ab/2}\, \theta_{ab}(\tau;z), \\
\theta_{11} \lp \tau; z+\frac{a-1}{2}\,\tau + \frac{b-1}{2} \rp &= e^{i \pi b(1-a)/2}\, q^{-(1-a)^2/8}\, y^{(1-a)/2}\, \theta_{ab}(\tau;z).
\end{align}
We also note the following relations for reference:
\begin{align}
 & \theta_{00}(\tau; z+1/2) =  \theta_{01}(\tau; z), \notag\\
& \theta_{01}(\tau; z+1/2) =  \theta_{00}(\tau; z), \notag \\
 & \theta_{10}(\tau; z+1/2) =  \theta_{11}(\tau; z),  \notag\\
& \theta_{11}(\tau; z+1/2) =  -\theta_{10}(\tau; z).
\end{align}
\begin{align}
 & \theta_{00}(\tau; z+\tau/2) = q^{-1/8}\, y^{-1/2} \, \theta_{10}(\tau; z),  \notag \\
 &\theta_{01}(\tau; z+\tau/2) = -\i\, q^{-1/8}\, y^{-1/2} \, \theta_{11}(\tau; z), \notag \\
 & \theta_{10}(\tau; z+\tau/2) = q^{-1/8}\, y^{-1/2} \, \theta_{00}(\tau; z),  \notag \\
 &\theta_{11}(\tau; z+\tau/2) =  -\i\, q^{-1/8}\, y^{-1/2} \, \theta_{01}(\tau; z).
\end{align}
\begin{align}
 & \theta_{00}(\tau; z+(\tau+1)/2) = -\i\, q^{-1/8}\, y^{-1/2} \, \theta_{11}(\tau; z) , \notag \\ 
& \theta_{01}(\tau; z+(\tau+1)/2) = q^{-1/8}\, y^{-1/2}\,  \theta_{10}(\tau; z), \notag \\
 & \theta_{10}(\tau; z+(\tau+1)/2) = -\i\, q^{-1/8}\, y^{-1/2} \, \theta_{01}(\tau; z),\notag \\
& \theta_{11}(\tau; z+(\tau+1)/2) =  -q^{-1/8} \,y^{-1/2}\,  \theta_{00}(\tau; z).
\end{align}

\noindent iii) In \cite{zwegers2008mock}, $\vartheta(z)$ is defined as:
\be
\vartheta(z) \equiv \vartheta(z;\tau)= \theta_{11}(\tau;z).
\ee
It obeys the following properties: 
\begin{equation}\label{eq:theta_ell}
\theta_{11}(\tau;z+1) = - \, \theta_{11}(\tau;z) \mbox{ and }
\theta_{11}(\tau; z+\tau) = -\,  q^{-1/2} \, y^{-1} \, \theta_{11}(\tau;z),
\end{equation}
and moreover, $z \to \theta_{11}(\tau;z)$ is the unique entire function satisfying these two properties up to an overall multiplicative constant.
\begin{equation}
\theta_{11}(\tau;-z) = - \, \theta_{11}(\tau;z), 
\end{equation}
and more generally $\theta_{ab}(\tau; -z) = (-1)^{ab} \, \theta_{ab}(\tau; z)$.

\noindent iv) Under modular transformations:
\begin{align}
\label{eq:theta11_modT}
\theta_{11}(\tau+1;z)&= e^{\pi \i/4} \, \theta_{11}(\tau;z), \\
\label{eq:theta11_modS}
\theta_{11}(-1/\tau;z/\tau) &= e^{-3 \pi \i/4} \, \tau^{1/2} \, e^{\pi \i z^2/\tau} \, \theta_{11}(\tau;z). 
\end{align}
\begin{equation}
\frac{1}{2 \pi \i} \,  \theta_{11}'(\tau;0)  = \i \, \eta(\tau)^3.
\end{equation}

\noindent c) Let
\begin{align}
 & x_0 = \frac{1}{2}(x+z+v+u)  \, ,  \qquad \qquad z_0 = \frac{1}{2}(x+z-v-u) \, , \cr
 & v_0 = \frac{1}{2}(x-z+v-u)  \, , \qquad \qquad u_0 = \frac{1}{2}(x-z-v+u) .
\end{align}
Then, we have the following Riemann theta relations \cite{mumford1983tata}.
\begin{align} \nonumber
(R5): + \theta_{00} \theta_{00} \theta_{00} \theta_{00}
- \theta_{01} \theta_{01} \theta_{01} \theta_{01}
- \theta_{10}\theta_{10} \theta_{10}\theta_{10}  
+ \theta_{11} \theta_{11} \theta_{11} \theta_{11} & = 
2 \theta_{11} \theta_{11} \theta_{11} \theta_{11}, \\  \nonumber
(R8):- \theta_{01} \theta_{01} \theta_{00} \theta_{00} 
+\theta_{00} \theta_{00} \theta_{01} \theta_{01}
 -\theta_{11} \theta_{11} \theta_{10} \theta_{10}
 + \theta_{10} \theta_{10} \theta_{11} \theta_{11}  &=  
 -2 \theta_{11} \theta_{11} \theta_{10} \theta_{10}, \\ \nonumber 
(R9): - \theta_{01} \theta_{01} \theta_{00} \theta_{00}
+ \theta_{00} \theta_{00} \theta_{01} \theta_{01} 
 +\theta_{11} \theta_{11} \theta_{10} \theta_{10}
- \theta_{10} \theta_{10} \theta_{11} \theta_{11}  &= 
 -2 \theta_{10} \theta_{10} \theta_{11} \theta_{11}, \\ \nonumber 
(R11): - \theta_{10} \theta_{10} \theta_{00} \theta_{00}
- \theta_{11} \theta_{11} \theta_{01} \theta_{01}
+\theta_{00} \theta_{00} \theta_{10} \theta_{10} 
+ \theta_{01} \theta_{01} \theta_{11} \theta_{11}   &= 
2 \theta_{01} \theta_{01} \theta_{11} \theta_{11}, \\ \nonumber
(R13): - \theta_{10} \theta_{10} \theta_{00} \theta_{00}
 +\theta_{11} \theta_{11} \theta_{01} \theta_{01}
+\theta_{00} \theta_{00} \theta_{10} \theta_{10} 
- \theta_{01} \theta_{01} \theta_{11} \theta_{11} &= 
   2 \theta_{11} \theta_{11} \theta_{01} \theta_{01}, \\ \nonumber   
(R15): -\theta_{11} \theta_{11} \theta_{00} \theta_{00}
- \theta_{10} \theta_{10} \theta_{01} \theta_{01} 
+ \theta_{01} \theta_{01} \theta_{10} \theta_{10}
+\theta_{00} \theta_{00} \theta_{11} \theta_{11}   &= 
2 \theta_{00} \theta_{00} \theta_{11} \theta_{11}, \\ \nonumber
(R16):  -\theta_{11} \theta_{11} \theta_{00} \theta_{00}
+ \theta_{10} \theta_{10} \theta_{01} \theta_{01}
- \theta_{01} \theta_{01} \theta_{10} \theta_{10} 
+\theta_{00} \theta_{00} \theta_{11} \theta_{11}  &= 
 -2 \theta_{11} \theta_{11} \theta_{00} \theta_{00}.
\end{align}
In these relations, the arguments of the theta functions on the left hand side are 
$(\tau;x)$, $(\tau;z)$, $(\tau;v)$, $(\tau;u)$ in that order, and arguments for the theta functions on the right hand
side are $(\tau;x_0)$, $(\tau;z_0)$, $(\tau;v_0)$, $(\tau;u_0)$ again in that order.

\noindent d) Level $k$ Jacobi theta functions $\vartheta_{k,r}(\tau;z)$ are defined as
\be
\vartheta_{k,r}(\tau;z)
 \equiv \sum_{\substack{n \in \ZZ \\ n \, \equiv \, r \bmod 2k  }} q^{n^2/4k} \, y^{n}
= \sum_{n \in \ZZ+r/2k} q^{k n^2} \, y^{2kn}.
\ee

Note that there is another theta function definition, which is used commonly in affine algebra characters
\be
\Theta_{r,k}(\tau;z)
 \equiv  \sum_{n \in \ZZ+r/2k} q^{k n^2} \, y^{kn}.
\ee
It is related to the Jacobi theta function as
\be
\vartheta_{k,r}(\tau;z) = \Theta_{r,k}(\tau;2z).
\ee
In this paper, we exclusively use Jacobi theta functions, $\vartheta_{k,r}$, even in cases involving affine $SU(2)$ characters.

We also define a combination that appears often
\be \label{defthhat}
\wh \vartheta_{k,r}(\tau;z) = \vartheta_{k,r}(\tau;z) - \vartheta_{k,-r}(\tau;z).
\ee

From the Taylor expansion of the Jacobi theta function we also have the weight $3/2$ unary theta function
\be \label{defSkr}
S_{k,r}(\tau) =  \left. \frac{1}{2 \pi \i} \frac{\partial}{\partial z} \vartheta_{k,r}(\tau;z) \right\vert_{z=0} = \sum_{n \in \ZZ} (2kn+r) q^{(2kn+r)^2/4k}.
\ee

\noindent i) Identities:
\begin{align}
\vartheta_{k,r+2k}(\tau;z) &= \vartheta_{k,r}(\tau;z), \\
\vartheta_{k,-r}(\tau;z) &= \vartheta_{k,r}(\tau;-z), \\
\vartheta_{k,r}(\tau;z) + \vartheta_{k,r+k}(\tau;z) &= \vartheta_{k/2,r}(\tau/2;z), \\
S_{k,r+2k}(\tau) &= S_{k,r}(\tau), \\
S_{k,-r}(\tau) &= - \, S_{k,r}(\tau). 
\end{align}

\noindent ii) Elliptic transformations:
\begin{align}
\vartheta_{k,r}(\tau;z+1) &= \vartheta_{k,r}(\tau;z), \\
\vartheta_{k,r}(\tau;z+\tau) &= q^{-k} \, y^{-2k} \, \vartheta_{k,r}(\tau;z), \\
\vartheta_{k,r}(\tau;z+1/2) &= (-1)^r \, \vartheta_{k,r}(\tau;z), \\
\vartheta_{k,r}(\tau;z+\tau/2) &= q^{-k/4} \,  y^{-k} \, \vartheta_{k,k+r}(\tau;z).
\end{align}

\noindent iii) Modular transformations:
\begin{align}
\vartheta_{k,r}(\tau+1;z) &= e^{\pi \i r^2/2k} \, \vartheta_{k,r}(\tau;z), \\
\label{eq:theta_kr_modS}
\vartheta_{k,r}(-1/\tau; z/\tau) &= \sqrt{-\i \tau} \, e^{2\pi \i z^2 k/\tau} 
			\sum_{r'=-k+1}^k {\CS}^{(k)}_{r,r'} \, \vartheta_{k,r'}(\tau;z), \\
S_{k,r}(\tau+1) &= e^{\pi \i r^2/2k} \, S_{k,r}(\tau), \\
S_{k,r}(-1/\tau) &= e^{- \pi \i /4} \, \tau^{3/2}  
			\sum_{r'=-k+1}^k {\CS}^{(k)}_{r,r'} \, S_{k,r'}(\tau) \notag  \\
			&= ( - \i \tau)^{3/2} \, \sum_{r'=1}^{k-1} \wh{{\CS}}^{(k)}_{r,r'} \, S_{k,r'}(\tau), 
\end{align}
where
\begin{equation}
{\CS}^{(k)}_{r,r'} = \frac{1}{\sqrt{2k}} \, e^{-\pi \i r r'/k} 
\quad \mbox{and} \quad
\wh{{\CS}}^{(k)}_{l,l'}  \equiv \sqrt{\frac{2}{k}} \, \sin \frac{\pi l l'}{k} 
			= \i \lp {\CS}^{(k)}_{l,l'} - {\CS}^{(k)}_{l,-l'} \rp.
\end{equation}

\noindent e) Jacobi Forms 

The modular and elliptic transformation laws of Jacobi theta functions provide a template for the definition of the broader class of functions
known as Jacobi forms. Following \cite{eichler1985theory} we say that a holomorphic function $\phi: \HH \times \CC \rightarrow \CC$ is a Jacobi
form of weight $k$ and index $m$ if it transforms under the Jacobi group $SL_2(\ZZ) \ltimes \ZZ^2$ as
\bea
\phi(\tau,z) &=& (c \tau+d)^{-k} \exp \left( - 2 \pi i m \frac{c z^2}{c \tau+d} \right) \phi \left( \frac{a \tau+b}{c \tau+d}, \frac{z}{c \tau+d} \right), ~~\begin{pmatrix} a & b \\ c & d \end{pmatrix} \in SL_2(\ZZ)  \\
\phi(\tau,z) &=& \exp \left( 2 \pi i m ( \lambda^2 \tau+ 2 \lambda z) \right) \phi(\tau, z+ \lambda \tau + \mu ), ~~ \lambda, \mu \in \ZZ \, .
\eea
Invariance under the transformations $\tau \rightarrow \tau+1$ and $z \rightarrow z+1$ implies that we can write a Fourier expansion in the
form
\be
\phi(\tau,z) = \sum_{n,r \in \ZZ} c(n,r) q^n y^r
\ee
and the elliptic properties imply that the coefficients $c(n,r)$ depend only on the discriminant $r^2-4mn$ and on $r ~{\rm mod}~2m$.  A weak Jacobi
form is a Jacobi form whose coefficients obey $c(n,r)=0$ whenever $n<0$.  In our analysis an important role is played by the weak Jacobi forms
$\varphi_{0,1}(\tau,z)$ and $\varphi_{-2,1}(\tau,z)$ with (weight,index) of $(0,1)$ and $(-2,1)$ respectively given by
\begin{align} \label{phzeroone}
\varphi_{0,1}(\tau;z) &= 4 \lp \frac{\th_{00}(\tau; z)^2}{\th_{00}(\tau; 0)^2}
											+ \frac{\th_{01}(\tau; z)^2}{\th_{01}(\tau; 0)^2}
											+\frac{\th_{10}(\tau; z)^2}{\th_{10}(\tau; 0)^2}  \rp
\notag \\
&= \frac{y^2 + 10 y + 1}{y} + 2 \, \frac{(y-1)^2 (5 y^2 - 22 y +5)}{y^2} \, q + \ldots
\end{align}
\begin{align} \label{phtwoone}
\varphi_{-2,1}(\tau;z) &= - \, \frac{\th_{11}(\tau; z)^2}{\eta(\tau)^6}  \notag \\
			&= \frac{(y-1)^2}{y} - 2 \, \frac{(y-1)^4}{y^2} \, q +  \ldots
\end{align}

\subsection{\texorpdfstring{$SU(2)_k/U(1)$}{SU(2)/U(1)} characters \label{sec:su2char}}

\noindent a) Affine $SU(2)$ Characters:

The character of the spin $l-1$ representation of affine $SU(2)$ algebra at level $k-2$ is
\begin{equation} \label{defchikj}
\chi_{l-1}^{(k-2)} (\tau;z) = \frac{\vartheta_{k,l} (\tau;z/2) -  \vartheta_{k,-l} (\tau;z/2) }{\vartheta_{2,1} (\tau;z/2) -  \vartheta_{2,-1} (\tau;z/2) } 
= \i \, \frac{\wh \vartheta_{k,l}(\tau;z/2)}{\theta_{11}(\tau;z)},
\end{equation}
where $l$ is an integer in the range $1 \leq l \leq k-1$. It is an entire and even function of the variable~$z$.

\noindent i) Elliptic transformations:
\begin{align}
\label{eq:chi_ell1}
\chi_{l-1}^{(k-2)} (\tau;z+1) &= (-1)^{l+1} \, \chi_{l-1}^{(k-2)} (\tau;z) , \\ 
\label{eq:chi_elltau}
\chi_{l-1}^{(k-2)} (\tau;z+\tau) &=
				 q^{-(k-2)/4} \, y^{-(k-2)/2} \, \chi_{k-l-1}^{(k-2)} (\tau;z).
\end{align}

\noindent ii) Modular transformations:
\begin{align}
\chi_{l-1}^{(k-2)} (\tau+1;z) &= 
e^{\pi \i l^2/2k} \, e^{-\pi \i/4} \, \chi_{l-1}^{(k-2)} (\tau;z), \label{eq:chi_modT} \\ 
\chi_{l-1}^{(k-2)} (-1/\tau;z/\tau) &= e^{\pi \i z^2 (k-2)/2\tau} \, \sum_{l' = 1}^{k-1} 
\wh{{\CS}}^{(k)}_{l,l'}    \chi_{l'-1}^{(k-2)} (\tau;z). \label{eq:chi_modS} 
\end{align}

\noindent b) Modular invariant combinations:

The modular transformation for $\vartheta_{k,r}$ involves matrices
\begin{equation}
{\CS}^{(k)}_{r,r'} = \frac{1}{\sqrt{2k}} \, e^{-\pi \i r r'/k} \quad \mbox{and} \quad
{\CT}^{(k)}_{r,r'} = e^{\pi \i r^2/2k} \, \delta_{r,r'}.
\end{equation}
for $r,r' \in \ZZ_{2k}$. Both of these matrices are symmetric and unitary.

For each divisor $d$ of $k$ we define
$2k \times 2k$ matrices
\begin{align}\label{eq:OmegaDefn}
\Omega^{(k,d)}_{r,r'} = \begin{cases} 1 &\text{if $r+r' = 0$ mod $2 d$ and $r-r' = 0$ mod ${2k}/{d}$,} \\
0 &{\rm otherwise}.
\end{cases}
\end{align}
These matrices then satisfy
\begin{equation}\label{Omega_S}
{\CS}^{(k) \dagger} \, \Omega^{(k,d)} \,  {\CS}^{(k)} = \Omega^{(k,d)} 
\quad \mbox{and} \quad
\CT^{(k) \dagger} \, \Omega^{(k,d)} \, \CT^{(k)} = \Omega^{(k,d)} .
\end{equation}
That means, in particular, that the form of 
$\sum_{r,r' \in \ZZ_{2k}} \vartheta_{k,r}^* \, \Omega^{(k,d)}_{r,r'} \, \vartheta_{k,r'}$ is preserved under the modular transformations. Also, note that $ \Omega^{(k,d)}_{r,r'} =  \Omega^{(k,d)}_{r',r} $.

The modular transformations of  $\chi_{l-1}^{(k-2)}$, $\wh{\vartheta}_{k,r}$ and 
$S_{k,r}$ on the other hand involves the matrix
\begin{equation}\label{eq:Shat_S_reln}
\wh{{\CS}}^{(k)}_{l,l'}	= \i \lp {\CS}^{(k)}_{l,l'} - {\CS}^{(k)}_{l,-l'} \rp,
\qquad  l,l' = 1, \ldots, k-1.
\end{equation}
It satisfies
\begin{equation}
\wh{{\CS}}^{(k)}_{l,l'} = \wh{{\CS}}^{(k)}_{l',l},
\quad
\wh{{\CS}}^{(k)*}_{l,l'} = \wh{{\CS}}^{(k)}_{l,l'},
\quad
\wh{{\CS}}^{(k)}_{l,k-l'} = (-1)^{l+1} \,\wh{{\CS}}^{(k)}_{l',l}
\quad \mbox{and} \quad
\sum_{r = 1}^{k-1}  \wh{{\CS}}^{(k)}_{l,r} \,   \wh{{\CS}}^{(k)}_{r,l'}  = \delta_{l,l'}.
\end{equation}

It is natural to define then the $(k-1) \times (k-1)$ matrices
\begin{equation}
\wh{\Omega}^{(k,d)}_{r,r'} = {\Omega}^{(k,d)}_{r,r'} - {\Omega}^{(k,d)}_{r,-r'}
 \qquad \mbox{for } r,r'=1, \ldots k-1.
\end{equation}
From (\ref{eq:OmegaDefn}) we immediately deduce that 
\begin{equation}
\wh{\Omega}^{(k,d)}_{r,r'} = \wh{\Omega}^{(k,d)}_{k-r,k-r'}.
\end{equation}
Moreover, $\wh{\Omega}^{(k,d)}_{r,r'} = 0$ unless $r^2 - {r'}^2 \in 4 k \ZZ$ and this gives
\begin{equation}\label{eq:Omega_commutingT}
\wh{\Omega}^{(k,d)}_{r,r'} e^{\pi \i (r^2 - {r'}^2)/2k} = \wh{\Omega}^{(k,d)}_{r,r'}.
\end{equation}
Finally, using (\ref{Omega_S}) we see that
\begin{equation}\label{eq:Omega_commutingS}
\wh{{\CS}}^{(k) \dagger} \, \wh{\Omega}^{(k,d)} \, \wh{{\CS}}^{(k)} = \wh{\Omega}^{(k,d)}.
\end{equation}

\noindent c) String Functions:

The function $z \to \chi_{l}^{(k)} (\tau;2z)$ satisfies $\chi_{l}^{(k)} (\tau;2(z+1)) = \chi_{l}^{(k)} (\tau;2z)$ and $\chi_{l}^{(k)} (\tau;2(z+\tau)) = q^{-k} y^{-2k}\chi_{l}^{(k)} (\tau;2z)$. In other words, under elliptical transformations it transforms like an index $k$ Jacobi form and hence has a theta decomposition. This follows from the fact  that there is a $U(1)$ symmetry we can gauge in the level $k$ affine $SU(2)$ algebra so that the characters have a branching relation
\begin{equation}\label{eq:string_fnc}
\chi_{l}^{(k)} (\tau;z) = \sum_{r \in \ZZ_{2k} } c^{(k)}_{l,r}(\tau) \, \vartheta_{k,r}(\tau;z/2),
\end{equation}
where we define string functions $c^{(k)}_{l,r}(\tau)$ using this decomposition.
 
From( \ref{eq:string_fnc}) and from the properties of $\chi_{l}^{(k)} (\tau;z)$ it is easy to show that
\begin{equation}
c^{(k)}_{l,r}(\tau) = c^{(k)}_{l,r+2k}(\tau) = c^{(k)}_{l,-r}(\tau) = c^{(k)}_{k-l,k-r}(\tau).
\end{equation}

Since $\chi_{l}^{(k)} (\tau;z+1) = (-1)^l \, \chi_{l}^{(k)} (\tau;z)$ and 
$\vartheta_{k,r}(\tau, (z+1)/2) = (-1)^r \, \vartheta_{k,r}(\tau; z/2)$ 
\begin{equation}\label{eq:strfnc_zero}
c^{(k)}_{l,r}(\tau)=0 \quad \mbox{if } l-r \neq 0 \, (\mbox{mod }2). 
\end{equation}

\noindent d) Supersymmetric Coset:

The supersymmetric coset theory $SU(2)_k/U(1)$ gives rise to $N=2$ minimal models with central charge
\be
\wh{c} = \frac{c}{3} = 1-2/k, \qquad k=2,3,4,\ldots
\ee
We will need characters in the R and NS sector. There is also a T sector in which the two supercharges have different periodicity.
We use the conventions of \cite{Israel:2004ir} where (adding superscript $k$ to their notation to specify the level)
\begin{equation}\label{eq:su2u1_defn}
\chi_r^{l,s,k}(\tau; z) = \sum_{n \in \ZZ_{k-2}} c^{(k-2)}_{l,r-s+4n}(\tau) \,
 \vartheta_{2k(k-2),2r+k(4n-s)}(\tau; z/2k),
\end{equation}
where $s=0,2,1,3$ denotes NS sector states with $(-1)^F= 1,-1$ and R sector states with $(-1)^F=1,-1$ respectively.
$ c^{(k-2)}_{l,r}(\tau) $ are the $SU(2)_{k-2}$ string functions defined through (\ref{eq:string_fnc}).
$\chi_r^{l,s,k}(\tau; z)$ satisfies the branching relation
\begin{equation} \label{chibranching}
\chi_{l}^{(k-2)} (\tau;w) \, \vartheta_{2,s} \lp \tau; \frac{w-z}{2} \rp =
\sum_{r \in \ZZ_{2k}} \chi_r^{l,s,k}(\tau; z) \, \vartheta_{k,r} \lp \tau; \frac{w}{2} - \frac{z}{k} \rp.
\end{equation}
Specifically for $k=2$, $\chi_{l}^{(k-2=0)} (\tau;z)=1$ and this branching relation implies
\begin{equation}\label{eq:k2_SU2}
\chi_r^{l=0,s,k=2}(\tau; z) = \delta_{r,s \,(\mbox{\small mod }4) }.
\end{equation}

Finally, we define supersymmetric $SU(2)_k/U(1)$ characters by combining $\chi_r^{l,s,k}(\tau; z)$ as
\begin{equation} \label{Clkrdef}
C^{l,k}_r \bmat{a \\ b}(\tau; z) =  e^{i \pi a b/2} 
\left( \chi^{l,a,k}_r(\tau; z)  + (-1)^b \chi^{l,a+2,k}_r(\tau; z)  \right), 
\end{equation}
where $a=0,1$ denotes NS and R sectors in that order and $b=0,1$ states whether the trace
defining the character is taken with or without a $(-1)^F$ insertion, respectively. Using (\ref{eq:su2u1_defn}) we can also write it as 
\begin{equation}\label{eq:CdefnSum}
C^{l,k}_r \bmat{a \\ b}(\tau; z) =  e^{i \pi a b/2} \sum_{n \in \ZZ_{2(k-2)}} (-1)^{b n} \, 
c^{(k-2)}_{l,r+2n-a}(\tau)  \,
 \vartheta_{2k(k-2),2r+k(2n-a)}(\tau; z/2k).
\end{equation}
Note that from (\ref{eq:strfnc_zero}) it follows that
\begin{equation}\label{eq:Cfnc_zero}
C^{l,k}_r \bmat{a \\ b}(\tau; z)=0 \quad \mbox{if } l+r+a \neq 0\, (\mbox{mod }2). 
\end{equation}

\noindent i) Identities:

Using (\ref{eq:CdefnSum}) we can show
\begin{align}
\label{eq:su2u1_prop1}
C^{l,k}_{-r} \bmat{a \\ b}(\tau; z) &= (-1)^{a b} \, C^{l,k}_r \bmat{a \\ b}(\tau; -z), \\
\label{eq:su2u1_prop2}
C^{l,k}_{r+2k} \bmat{a \\ b}(\tau; z) &= C^{l,k}_r \bmat{a \\ b}(\tau; z), \\
\label{eq:su2u1_prop3}
C^{k-2-l,k}_{r+k} \bmat{a \\ b}(\tau; z) &= (-1)^b \, C^{l,k}_r \bmat{a \\ b}(\tau; z).
\end{align}
\begin{equation}\label{eq:su2u1_shift00}
C^{l,k}_{r+a} \bmat{0 \\ 0} \lp \tau; z+\frac{a}{2}\, \tau + \frac{b}{2} \rp = 
q^{-a^2(k-2)/8k}\, y^{-a(k-2)/2k}\, e^{i \pi ab/2}\, e^{i \pi b(r+a)/k}\, 
C^{l,k}_{r} \bmat{a \\ b} (\tau; z),
\end{equation}
\begin{align}
C^{l,k}_{r+a-1} \bmat{1 \\ 1} \lp \tau; z+\frac{a-1}{2}\,\tau + \frac{b-1}{2} \rp  = 
-\,&q^{-(1-a)^2(k-2)/8k}\, y^{(1-a)(k-2)/2k}\, e^{i \pi (a+1)b/2}\, \notag \\
\label{eq:su2u1_shift11}
&\times e^{i \pi (b-1)(r+a-1)/k}\, 
C^{l,k}_{r} \bmat{a \\ b} (\tau; z).
\end{align}

$C^{l,k}_r \bmat{1 \\ 1}(\tau; z)$ is particularly important for our discussion. It satisfies the branching relation 
\begin{equation}\label{eq:su2_branch}
\chi_{l}^{(k-2)} \lp \tau;w + \frac{2z}{k} \rp \, \theta_{11} \lp \tau; w - \frac{k-2}{k}z \rp =
\sum_{r \in \ZZ_{2k}} C^{l,k}_r \bmat{1 \\ 1}(\tau; z) \, \vartheta_{k,r} \lp \tau; w/2 \rp.
\end{equation}

Formally $\chi_{l+2k}^{(k-2)} (\tau;z) = \chi_{l}^{(k-2)} (\tau;z) $ and 
$\chi_{-l-2}^{(k-2)} (\tau;z) = - \chi_{l}^{(k-2)} (\tau;z)$. This  in turn leads to
\begin{align}
C^{l+2k,k}_r \bmat{1 \\ 1}(\tau; z) &= C^{l,k}_r \bmat{1 \\ 1}(\tau; z), \\
C^{-l-2,k}_r \bmat{1 \\ 1}(\tau; z) &= - \,C^{l,k}_r \bmat{1 \\ 1}(\tau; z), \\
C^{-1,k}_r \bmat{1 \\ 1}(\tau; z) &= C^{k-1,k}_r \bmat{1 \\ 1}(\tau; z) = 0.
\end{align}

Setting $z\to 0$ in the branching relation (\ref{eq:su2_branch}) we get
\begin{equation}
C^{l-1,k}_r \bmat{1 \\ 1}(\tau; 0) = \i \lp \delta_{r,l \, (\mbox{\small mod }2k) } - 
																	\delta_{-r,l \, (\mbox{\small mod }2k) } \rp.
\end{equation}

\noindent ii) Elliptical Transformations:

Using (\ref{eq:CdefnSum}) one can easily obtain
\begin{align}
C^{l,k}_r \bmat{a \\ b}(\tau; z+1) &= e^{2 \pi \i r/k}\,  e^{-\pi \i a} \, C^{l,k}_r \bmat{a \\ b}(\tau; z), \\
\label{eq:Cfnc_elltau}
C^{l,k}_r \bmat{a \\ b}(\tau; z+\tau) &=  q^{-(k-2)/2k} \, 
											 y^{-(k-2)/k} \, (-1)^b \,  C^{l,k}_{r-2} \bmat{a \\ b}(\tau; z).
\end{align}

\noindent iii) Modular Transformations:

Using (\ref{eq:chi_modT}),(\ref{eq:chi_modS}) and (\ref{eq:su2_branch}) one can work out the behavior of 
$C^{l,k}_r \bmat{1 \\ 1}(\tau; z)$ under modular transformations as
\begin{align}
\label{eq:Cfnc_modT}
C^{l-1,k}_r \bmat{1 \\ 1}(\tau+1; z) &= e^{\pi \i (l^2 - r^2)/2k} \, C^{l-1,k}_r \bmat{1 \\ 1}(\tau; z) \\
\label{eq:Cfnc_modS}
C^{l-1,k}_r \bmat{1 \\ 1}(-1/\tau; z/\tau) &= -\i \, e^{\pi \i z^2 (k-2)/k\tau} \, 
		\sum_{l' =1}^{k-1} \sum_{r' \in \ZZ_{2k}} \wh{{\CS}}^{(k)}_{l,l'}  \,  {\CS}^{(k)*}_{r,r'} \,
		C^{l'-1,k}_{r'} \bmat{1 \\ 1}(\tau; z).
\end{align}

\noindent iv) Explicit Expressions:

According to \cite{Ravanini:1987yg} we can write 
$C^{l,k}_r \bmat{0 \\ 0}(\tau; z)$ as 
\begin{equation}
C^{l-1,k}_r \bmat{0 \\ 0}(\tau; z) = q^{\frac{l^2 - r^2}{4k}}  y^{\frac{r}{k}} 
\frac{\theta_{00}(\tau;z)}{\eta(\tau)^3}
\sum_{n \in \ZZ} q^{k n^2 + l n} \lp 
\frac{1}{1 + y^{-1} q^{k n + \frac{l+r}{2}}}
+ \frac{1}{1 + y q^{k n +\frac{l-r}{2}}}  - 1
 \rp
\end{equation} 
or as
\begin{equation}
C^{l-1,k}_r \bmat{0 \\ 0}(\tau; z) = 
\frac{ \i \, q^{(l^2 - r^2-k^2)/4k}\,   y^{r/k}\,  \theta_{00}(\tau;z) \, \theta_{11}(k \tau; l \tau)}
{\theta_{10}(k \tau;z - \frac{l+r}{2} \tau) \, 
\theta_{10}(k \tau;z+\frac{l - r}{2}\tau)}\,
\frac{\eta(k \tau)^3}{\eta(\tau)^3}
\end{equation}
whenever $l-1 = r\, (\mathrm{ mod }\,2)$. We find then explicit expressions for $C^{l-1,k}_r \bmat{a \\ b}(\tau; z)$ using equation (\ref{eq:su2u1_shift00}). From these expressions, it is easy to check that $z \to C^{l-1,k}_r \bmat{a \\ b}(\tau; z)$ is an entire function.

\section{Appell-Lerch sums}\label{sec:Appell_Lerch}

In \cite{zwegers2008mock} Zwegers defines a  function $\mu(u,v;\tau)$  as
\be
\mu(u,v;\tau) \equiv \frac{e^{i \pi u}}{\vartheta(v,\tau)} \sum_{n \in \ZZ} 
\frac{(-1)^n \, q^{n(n+1)/2} \, e^{ 2 \pi i n v}}{1-q^n \, e^{2 \pi i u}}.
\ee

We will define now a generalization \footnote{Multivariable generalizations of Appell-Lerch sums which are distinct from those presented here are discussed
in \cite{zwegers2010multivariable}.} :
\begin{definition}
 \be\label{eq:mu_defn1}
\mu^{k,j,j'}(\tau;v,u,w) \equiv \frac{y_v^{j'/k}}{\theta_{11}(\tau;u)}
        \sum_{r \in \ZZ_k} \sum_{n \in \ZZ + \frac{r}{k}}
        \frac{q^{k n^2 + j' n}}{1-y_v \, q^{k n}} \lp \frac{y_u^k\, y_w^{k-2}}{y_v^{k-2}} \rp^n
        C^{j-1,k}_{2r + j'}\bmat{1 \\ 1}(\tau;w-v)
\ee
or equivalently 
\be\label{eq:mu_defn2}
\mu^{k,j,j'}(\tau;v,u,w) = \frac{y_v^{j'/k}}{\theta_{11}(\tau;u)}
        \sum_{n \in \ZZ}
        \frac{q^{n^2/k} \, q^{j'n/k}}{1-y_v \, q^{n}} \lp \frac{y_u\, y_w^{1-2/k}}{y_v^{1-2/k}} \rp^n
        C^{j-1,k}_{2n + j'}\bmat{1 \\ 1}(\tau;w-v)
\ee
where $\tau \in \IH$, $w \in \IC$ and $v,u \in \IC - (\IZ \tau + \IZ)$. The superscripts, $k$, $j$ and $j'$, are integers and $k \geq 2$. 
\end{definition}

We will usually deal with $j$ and $j'$ in the range $\{ j,j'=1, \ldots, k-1 \}$, but we will not restrict them for now, aside from requiring them to be integers, and use the expression above to define it for general $j$ and $j'$. An important property is that
\begin{equation}\label{eq:mufnc_zero}
\mu^{k,j,j'}(\tau;v,u,w)=0 \quad \mbox{if } j \neq j' (\mbox{mod }2). 
\end{equation}
This follows from the property (\ref{eq:Cfnc_zero}) of  $C^{j-1,k}_{2n + j'}\bmat{1 \\ 1}$ function.
For the rest of this section we will implicitly assume $j = j' \, (\mbox{mod }2)$ if we do not state otherwise.

\begin{definition}
 We define
\begin{equation}
\mu^{(k,d)}(\tau;v,u,w) = 
		\sum_{j,j'=1}^{k-1} \wh{\Omega}^{(k,d)}_{j,j'}  \ \mu^{k,j,j'}(\tau;v,u,w)
\end{equation}
for $d | k$ and
\begin{equation}
\mu^{Y}(\tau;v,u,w) = 
		\sum_{j,j'=1}^{k-1} \wh{\Omega}^Y_{j,j'} \ \mu^{k,j,j'}(\tau;v,u,w)
\end{equation}
for a simply laced root system $Y$ with Coxeter number $k$.
\end{definition}

In the following, we will work out some properties for these functions generalizing the properties of $\mu(u,v;\tau)$ as worked out in \cite{zwegers2008mock}.

\begin{remark}
 $\mu^{2,1,1}(\tau;v,u,w)$ is equal to $\i \, \mu(v,u;\tau)$ of \cite{zwegers2008mock} and to $\mu(v,u;\tau)$ of \cite{Eguchi:2008gc}. This can be seen explicitly from \eqref{eq:mu_defn2} using $C^{0,2}_{2n+1}\bmat{1 \\ 1}(\tau;w-v) = \i \, (-1)^n$ (see \ref{eq:k2_SU2}).
\end{remark}

\subsection{Basic Properties and Behavior Under Elliptical Transformations}

\begin{proposition}\label{prop:mu_analyticity}
The function $ \mu^{k,j,j'}(\tau;v,u,w) $ is holomorphic on the set it is defined ($\tau \in \IH$, $w \in \IC$ and $v,u \in \IC - (\IZ \tau + \IZ)$). Moreover both $u \to \mu^{k,j,j'}(\tau;v,u,w)$ and $v \to \mu^{k,j,j'}(\tau;v,u,w)$ are meromorphic functions with at most simple poles at $u = a_1 \tau + b_1$ and $v = a_2 \tau + b_2$ respectively with $a_1,a_2,b_1,b_2 \in \IZ$. In particular, the residue of the pole at $v=0$ for the function $v \to \mu^{k,j,j'}(\tau;v,u,w)$ is $-\frac{1}{2 \pi \i} \, \frac{C^{j-1,k}_{j'}\bmat{1 \\ 1}(\tau;w)}{\theta_{11}(\tau;u)}$.
\end{proposition}
\begin{proof}
 That $\mu^{k,j,j'}(\tau;v,u,w) $ is holomorphic for $\tau \in \IH$, $w \in \IC$ and $v,u \in \IC - (\IZ \tau + \IZ)$ is easy to see from \eqref{eq:mu_defn1} using the convergence of the sum over $n$ and also the holomorphicity of $C^{j-1,k}_{j'}\bmat{1 \\ 1}(\tau;x)$ for $\tau \in \IH$ and $x \in \IC$. That $u \to \mu^{k,j,j'}(\tau;v,u,w)$ has at most simple poles on $\IZ \tau + \IZ$ is immediate from the fact that $\theta_{11}(\tau;u)$ has simple zeros on $\IZ \tau + \IZ$. $v \to \mu^{k,j,j'}(\tau;v,u,w)$
 may have poles at $v$'s satisfying $1 - y_v \, q^{-a_1} = 0$ for an integer $a_1$ (see \eqref{eq:mu_defn2}). This happens on $\IZ \tau + \IZ$ where each location gives a simple pole on these locations unless there is a zero contribution from the $C^{j-1,k}_{j'}\bmat{1 \\ 1}$ term . In particular, $v=0$ pole comes from the $n=0$ term of \eqref{eq:mu_defn2}. The associated residue is
 \begin{align}
 \lim_{v \to 0} v \, \mu^{k,j,j'}(\tau;v,u,w) &= \frac{1}{\theta_{11}(\tau;u)} \lim_{v \to 0} \lp \frac{v}{1-e^{2 \pi \i v}} \rp
 C^{j-1,k}_{j'}\bmat{1 \\ 1}(\tau;w) \notag \\
 \label{eq:mu_v_residue}
 &= -\frac{1}{2 \pi \i} \frac{C^{j-1,k}_{j'}\bmat{1 \\ 1}(\tau;w)}{\theta_{11}(\tau;u)}.
 \end{align}
\end{proof}

\begin{proposition}\label{prop:mu_algell1}
\hfill
 \begin{description}
  \item[(a)] $\mu^{k,j,j'}(\tau;v,u+1,w) = -\mu^{k,j,j'}(\tau;v,u,w)$,
  \item[(b)] $\mu^{k,j,j'}(\tau;v+1,u,w) = -\mu^{k,j,j'}(\tau;v,u,w)$,
  \item[(c)] $\mu^{k,j,j'}(\tau;v,u,w+1) = -\, e^{2 \pi \i j'/k} \mu^{k,j,j'}(\tau;v,u,w)$,
  \item[(d)] $\mu^{k,j,j'}(\tau;v+\tau,u+\tau,w) = \mu^{k,j,j'}(\tau;v,u,w)$,
  \item[(e)] $\mu^{k,j,j'}(\tau;-v,-u,-w) = \mu^{k,k-j,k-j'}(\tau;v,u,w)$.
 \end{description}
\end{proposition}
\begin{proof}
 \hfill
 
 (a) follows trivially from the fact that $\theta_{11}(\tau;u+1) = - \theta_{11}(\tau;u)$.
 
 For (b), we first point out that
\be
 C^{j-1,k}_{2n + j'}\bmat{1 \\ 1}(\tau;w-v-1) = - e^{-2 \pi \i (2n+j')/k} \, C^{j-1,k}_{2n + j'}\bmat{1 \\ 1}(\tau;w-v).
\ee
The $e^{-2 \pi \i (2n+j')/k}$ factor is then canceled by $e^{2 \pi \i j'/k}$ produced by $y_v^{j'/k}$ and
$e^{4 \pi \i n/k}$ produced by $y_v^{-(1-2/k)n}$ in the sum.

The proof for (c) is quite similar to the proof for (b). $ C^{j-1,k}_{2n + j'}$ produces a $-\, e^{2 \pi \i (2n+j')/k}$ factor.
$e^{-4 \pi \i n/k}$ produced by $y_w^{(1-2/k)n}$ in the sum cancels a portion of this and we are left with the phase factor $-\, e^{2 \pi \i j'/k}$ as stated in the proposition.

For (d) we simply write
\begin{align}
\mu^{k,j,j'}(\tau;v+\tau,u+\tau,w) = &\frac{y_v^{j'/k} \, q^{j'/k}}{- q^{-1/2}\, y_u^{-1} \, \theta_{11}(\tau;u)}
        \sum_{n \in \ZZ}
        \frac{q^{n^2/k} \, q^{j' n/k}}{1-y_v \, q^{n+1}} \, q^{2n/k} \lp \frac{y_u\, y_w^{1-2/k}}{y_v^{1-2/k}} \rp^n \notag \\
        & \times \lp -  q^{-(k-2)/2k} \lp \frac{y_w}{y_v}\rp^{1-2/k} C^{j-1,k}_{2n + j'+2}\bmat{1 \\ 1}(\tau;w-v) \rp.
\end{align}
By rearranging various factors we get
\be
      \frac{y_v^{j'/k}}{ \theta_{11}(\tau;u)}  \sum_{n \in \ZZ}
        \frac{q^{(n+1)^2/k} \, q^{j'(n+1)/k}}{1-y_v \, q^{n+1}} \lp \frac{y_u\, y_w^{1-2/k}}{y_v^{1-2/k}} \rp^{n+1}
        C^{j-1,k}_{2(n+1) + j'}\bmat{1 \\ 1}(\tau;w-v)
\ee
which is just $\mu^{k,j,j'}(\tau;v,u,w)$ after shifting the dummy summation variable $n \to n-1$.

Finally, $\mu^{k,j,j'}(\tau;-v,-u,-w)$ is
\begin{align}
& \frac{y_v^{-j'/k}}{-\theta_{11}(\tau;u)}
        \sum_{n \in \ZZ}
        \frac{q^{n^2/k} \, q^{j'n/k}}{1-y_v^{-1} \,  q^{n}} \lp \frac{y_u\, y_w^{1-2/k}}{y_v^{1-2/k}} \rp^{-n}
        C^{j-1,k}_{2n + j'}\bmat{1 \\ 1}(\tau;v-w) \\
        &= - \frac{y_v^{-j'/k}}{\theta_{11}(\tau;u)}
        \sum_{n \in \ZZ}
        \frac{q^{n^2/k} \,  q^{-j' n/k}}{1-y_v^{-1} \, q^{-n}} \lp \frac{y_u\, y_w^{1-2/k}}{y_v^{1-2/k}} \rp^{n}
        C^{j-1,k}_{-2n + j'}\bmat{1 \\ 1}(\tau;v-w),
\end{align}
where we changed the summation variable $n \to -n$  on the second line. Using
(\ref{eq:su2u1_prop1}), (\ref{eq:su2u1_prop3}),
and $\frac{1}{1-y_v^{-1} \, q^{-1}} = -\frac{y_v \, q^n}{1-y_v \, q^n}$, (e) follows.

\end{proof}

\begin{proposition}\label{prop:mu_algell2}
 \hfill
 \begin{description}
  \item [(a)] 
  \begin{align*}
\mu^{k,j,j'}(\tau;v,u,w) \, +& y_v \, y_u^{-1} \, q^{-1/2} \, \mu^{k,j,j'}(\tau;v,u+\tau,w) =  \notag \\
&q^{-j'^2/4k} \lp \frac{y_v}{y_u y_w^{1-2/k}} \rp^{j'/2} \chi_{j-1}^{(k-2)}(\tau; w+u-v)
\end{align*}

\item [(b)]
\begin{align*}
\mu^{k,j,j'}(\tau;v,u,w) \, +& \, y_w^{(k-2)/k} \, q^{(k-2)/2k} \, \mu^{k,j,j'+2}(\tau;v,u,w+\tau) =  \notag \\
&q^{-j'^2/4k} \lp \frac{y_v}{y_u y_w^{1-2/k}} \rp^{j'/2} \chi_{j-1}^{(k-2)}(\tau; w+u-v)
\end{align*}

\item [(c)]
\begin{align*}
\mu^{k,j,j'}(\tau;v,u,w) \, +& \, \mu^{k,k-j,k+j'}(\tau;v,u,w) =  \notag \\
&q^{-j'^2/4k} \lp \frac{y_v}{y_u y_w^{1-2/k}} \rp^{j'/2} \chi_{j-1}^{(k-2)}(\tau; w+u-v)
\end{align*}

 \end{description}

\end{proposition}
\begin{proof}
 \hfill
 
 We start with (a) by writing $\mu^{k,j,j'}(\tau;v,u+\tau,w)$ as
\begin{equation}\label{eq:muplustau}
\frac{y_v^{j'/k}}{- q^{-1/2} y_u^{-1} \theta_{11}(\tau;u)}
        \sum_{n \in \ZZ}
        \frac{q^{n^2/k} q^{j' n/k}}{1-y_v q^{n}} \lp \frac{q y_u\, y_w^{1-2/k}}{y_v^{1-2/k}} \rp^n
        C^{j-1,k}_{2n + j'}\bmat{1 \\ 1}(\tau;w-v).
\end{equation}
Using $\frac{q^n}{1- y_v q^n} = y_v^{-1} \lp \frac{1}{1-y_v q^{n}} - 1\rp$ we can rewrite this as
\begin{equation}
  - \frac{q^{1/2} y_u}{y_v} \Big[ \mu^{k,j,j'}(\tau;v,u,w)  - 
\frac{y_v^{j'/k}}{ \theta_{11}(\tau;u)}
        \sum_{n \in \ZZ}
        q^{n^2/k} q^{j' n/k} \lp \frac{ y_u\, y_w^{1-2/k}}{y_v^{1-2/k}} \rp^n
        C^{j-1,k}_{2n + j'}\bmat{1 \\ 1}(\tau;w-v) \Big].
\end{equation}
This tells us that we can express the combination 
\begin{equation}
q^{j'^2/4k} \lp \frac{y_u y_w^{1-2/k}}{y_v} \rp^{j'/2}  \lb   \mu^{k,j,j'}(\tau;v,u,w) +  y_v y_u^{-1} q^{-1/2} \mu^{k,j,j'}(\tau;v,u+\tau,w) \rb 
\end{equation}
as
\begin{equation}
\frac{1}{ \theta_{11}(\tau;u)}
        \sum_{n \in \ZZ}
        q^{\lp n + j'/2 \rp^2/k} \lp \frac{ y_u\, y_w^{1-2/k}}{y_v^{1-2/k}} \rp^{n+j'/2}
        C^{j-1,k}_{2n + j'}\bmat{1 \\ 1}(\tau;w-v).
\end{equation}
Using (\ref{eq:su2u1_prop2}) and separating the sum over $n \in \ZZ$ as a sum over $s \in \ZZ$ and $p \in \ZZ_k$ by writing $n = ks +r$ we obtain
\begin{align}
\frac{1}{ \theta_{11}(\tau;u)}
       &\sum_{r \in \ZZ_k}  \sum_{s \in \ZZ}
        q^{k \lp s + \frac{2r+j'}{2k} \rp^2} \lp 
        \frac{ y_u\, y_w^{1-2/k}}{y_v^{1-2/k}} \rp^{k \lp s+\frac{2r+j'}{2k} \rp}
        C^{j-1,k}_{2r + j'}\bmat{1 \\ 1}(\tau;w-v). \\
        &= \frac{1}{ \theta_{11}(\tau;u)}
       \sum_{r \in \ZZ_k}  
       \vth_{k, 2r+j'}\lp \tau; \frac{1}{2} \lp u + (w-v)(1-2/k)  \rp \rp
        C^{j-1,k}_{2r + j'}\bmat{1 \\ 1}(\tau;w-v).
\end{align}
We can extend the sum over $\ZZ_k$ to a sum over $\ZZ_{2k}$ using (\ref{eq:Cfnc_zero}) as
\begin{equation}
\frac{1}{ \theta_{11}(\tau;u)}
       \sum_{r \in \ZZ_{2k}}  
       \vth_{k, r}\lp \tau; \frac{1}{2} \lp u + (w-v)(1-2/k)  \rp \rp
        C^{j-1,k}_{r}\bmat{1 \\ 1}(\tau;w-v).
\end{equation}
Finally the branching relation (\ref{eq:su2_branch}) tells us this is just $\chi_{j-1}^{(k-2)}(\tau; w+u-v)$ finishing the proof for (a).

Note that from (\ref{eq:chi_elltau} )and part (d) of proposition (\ref{prop:mu_algell1}) we find a similar identity for $v \to v+\tau$ transformation.  
 \begin{align*}
\mu^{k,j,j'}(\tau;v,u,w) +& y_u y_v^{-1} q^{-1/2} \mu^{k,j,j'}(\tau;v+\tau,u,w) =  \notag \\
&q^{-(k-j')^2/4k} \lp \frac{y_u y_w^{1-2/k}}{y_v} \rp^{(k-j')/2} \chi_{k-j-1}^{(k-2)}(\tau; w+u-v)
\end{align*}

For part (b) we use  ( \ref{eq:Cfnc_elltau}) and follow a proof very similar to that we had in part (a).

Finally, by part (a) 
\begin{equation}
-\mu^{k,j,j'}(\tau;v,u,w) + q^{-j'^2/4k} \lp \frac{y_v}{y_u y_w^{1-2/k}} \rp^{j'/2} \chi_{j-1}^{(k-2)}(\tau; w+u-v)
\end{equation}
is $y_v y_u^{-1} q^{-1/2} \mu^{k,j,j'}(\tau;v,u+\tau,w) $. By (\ref{eq:muplustau}) this combination is
\begin{equation}
- \frac{y_v^{(k+j')/k}}{\theta_{11}(\tau;u)}
        \sum_{n \in \ZZ}
        \frac{q^{n^2/k} q^{j' n/k} q^n }{1-y_v q^{n}} \lp \frac{y_u\, y_w^{1-2/k}}{y_v^{1-2/k}} \rp^n
        C^{j-1,k}_{2n + j'}\bmat{1 \\ 1}(\tau;w-v).
\end{equation}
Using (\ref{eq:su2u1_prop3}), we see that this is just $\mu^{k,k-j,k+j'}(\tau;v,u,w)$ proving part (c).

\end{proof}

We now give a lemma generalizing Proposition 1.4.7 of \cite{zwegers2008mock}. It will be an important ingredient both in working out the modular properties of the $\mu$ functions defined here and in the proof of the Riemann relations of Section \ref{sec:Riemann_reln}.

\begin{lemma}\label{lem:mu_z_shift}
 \begin{align}
 \mu^{k,j,j'}(\tau;v+z,&u+z,w) - \mu^{k,j,j'}(\tau;v,u,w)  \notag\\ 
 &= \frac{\i \eta(\tau)^3 \  C^{j-1,k}_{j'}\bmat{1 \\ 1}(\tau;w) \  \theta_{11}(\tau;z) \ \theta_{11}(\tau;v+u+z)}
 {\theta_{11}(\tau;v) \ \theta_{11}(\tau;u) \ \theta_{11}(\tau;v+z) \  \theta_{11}(\tau;u+z)},
\end{align}
for $u,v,u+z,v+z \notin \IZ \tau + \IZ$. 
\end{lemma}
\begin{proof}
 Let us define
 \begin{equation}
f(z) \equiv \frac{ \theta_{11}(\tau;v+z) \  \theta_{11}(\tau;u+z) 
 \lb \mu^{k,j,j'}(\tau;v+z,u+z,w) - \mu^{k,j,j'}(\tau;v,u,w) \rb}
 {\theta_{11}(\tau;z) \ \theta_{11}(\tau;v+u+z)}.
 \end{equation}
 Using (\ref{eq:theta_ell}) and parts (a), (b) and (d) of proposition \ref{prop:mu_algell1} it is easy to check that
 \begin{equation}
 f(z+1) = f(z) \qquad  \mathrm{and} \qquad f(z+\tau) = f(z).
 \end{equation}
 In other words, $f(z)$ is a meromorphic function of $z$ which is doubly periodic. 
 
 In fact, parts (a), (b) and (d) of proposition \ref{prop:mu_algell1} tell us that $\lb \mu - \mu \rb$ term alone is a meromorphic function in $z$ which is doubly periodic. By proposition \ref{prop:mu_analyticity}
 \begin{equation}
 z \to \lb \mu^{k,j,j'}(\tau;v+z,u+z,w) - \mu^{k,j,j'}(\tau;v,u,w) \rb 
 \end{equation}
 part has simple poles at $z = -v + a_1 \tau + b_1$ and at $z = -u + a_2 \tau + b_2$ with $a_1,a_2,b_1,b_2 \in \ZZ$ (or double poles if $u = v \, (\mathrm{mod } \, \ZZ \tau + \ZZ)$ but we will assume that this is not the case since generalization of our arguments to this case is straightforward). Moreover, this $\lb \mu - \mu \rb$ combination has zeros on $z \in \ZZ \tau + \ZZ$.
 
 The theta function prefactor in $f(z)$ preserves the double periodicity and cancels the poles at 
 $-u+\ZZ \tau + \ZZ$ and at $-v+\ZZ \tau + \ZZ$. In return, the theta functions in the denominator have simple zeros at $\ZZ \tau + \ZZ$ and at $-u-v+\ZZ \tau + \ZZ$ which would be poles for $f(z)$.
 However, the zeros of the denominator at   $\ZZ \tau + \ZZ$ are canceled by the zeros of the $\lb \mu - \mu \rb$ term on the same set. In summary, $f(z)$ is a doubly periodic function with at most one simple pole per fundamental parallelogram (at $-u-v+\ZZ \tau + \ZZ$).  
 
 A doubly periodic function with at most one simple pole per fundamental parallelogram is  a constant. Therefore,
 \begin{equation}\label{eq:mudiff}
 \mu^{k,j,j'}(\tau;v+z,u+z,w) - \mu^{k,j,j'}(\tau;v,u,w) =
 \frac{\theta_{11}(\tau;z) \ \theta_{11}(\tau;v+u+z)}
 {\theta_{11}(\tau;v+z) \  \theta_{11}(\tau;u+z)}\,
 g(\tau; v,u,w),
 \end{equation}
 for some meromorphic function $g(\tau; v,u,w)$.
 
 By proposition \ref{prop:mu_analyticity} the residue of the left hand side of equation( \ref{eq:mudiff}) at $z=-v$  is
 \begin{equation}
 -\frac{1}{2 \pi \i} \frac{C^{j-1,k}_{j'}\bmat{1 \\ 1}(\tau;w)}{\theta_{11}(\tau;u-v)}.
 \end{equation}
  For the right hand side the residue at the same position is 
 \begin{equation}
 \frac{\theta_{11}(\tau;-v) \ \theta_{11}(\tau;u)}
 {\theta_{11}(\tau;u-v)}  \lp \frac{1}{2 \pi \i} \frac{1}{\i \eta(\tau)^3} \rp g(\tau; v,u,w).
 \end{equation}
 Equating both expressions gives $g(\tau; v,u,w)$ and completes the proof.
 \end{proof}

\begin{corollary}
 Taking $z = -u-v$ in Lemma \ref{lem:mu_z_shift} gives 
 \begin{equation}
  \mu^{k,j,j'}(\tau;-u,-v,w) = \mu^{k,j,j'}(\tau;v,u,w).
 \end{equation}
 Then, using part (e) of Proposition \ref{prop:mu_algell1} we get 
 \begin{equation}\label{eq:mu_Z2symm}
  \mu^{k,k-j,k-j'}(\tau;u,v,-w) = \mu^{k,j,j'}(\tau;v,u,w).
 \end{equation}
\end{corollary}

\subsection{Modular behavior of \texorpdfstring{$\mu$}{mu} functions}
The behavior of $\mu^{k,j,j'}$ under $\tau \to \tau +1$ is easy to work out from its definition.
\begin{proposition}\label{prop:mu_modT}
 \begin{equation}
\mu^{k,j,j'}(\tau+1;v,u,w)  =  
e^{- \pi \i /4} \, e^{2 \pi \i \lb j^2 - j'^2 \rb /4k} \, \mu^{k,j,j'}(\tau;v,u,w) 
\end{equation}
\end{proposition}
\begin{proof}
 \hfill
 
 This easily follows from (\ref{eq:theta11_modT}) and (\ref{eq:Cfnc_modT}).
\end{proof}

The transformation under $\tau \to -1/\tau$, $v \to v/\tau$, $u \to u/\tau$ and 
$w \to w/\tau$ is not as simple and
we need to develop some tools first for this discussion.

\begin{definition}
 For $\tau \in \IH$ and $u \in \IC$ we define
 \begin{equation}
  h_{k,r}(\tau; u) \equiv \i \int\limits_{\IR} \dd x\  \frac{e^{2 \pi \i\, k \tau x^2 - 4 \pi k u x}}
  {1 - e^{2 \pi x - \pi \i r/k} } \, .
 \end{equation}
 We also define
 \begin{align}
  \wh{h}_{k,r}(\tau; u) &\equiv 2 \lb h_{k,r}(\tau;u/2) - h_{k,-r}(\tau;u/2) \rb, \notag \\
     &= \i \int\limits_{\IR} \dd x\ e^{\pi \i \, k \tau x^2/2} \, e^{- \pi k u x}
   \lp  \frac{1}{1 - e^{\pi x - \pi \i r/k} } - \frac{1}{1 - e^{\pi x + \pi \i r/k} } \rp.
 \end{align}
 for $r=1,\ldots,k-1$.
\end{definition}
\begin{remark}
For $r=1,2,\ldots,2k-1$, the definition here coincides with the definition of $h_r(u;\tau)$ in
 Proposition 3.3.6 of \cite{zwegers2008mock}. According to our definition we find
$h_{k,r} = h_{k,2k+r}$. Also, we note that $\wh{h}_{2,1}(\tau; u)$ is exactly $h(u;\tau)$ 
defined in Definition 1.1 of \cite{zwegers2008mock}.
 \end{remark}

\begin{proposition}\label{prop:hath_ell}
 $u \to \wh{h}_{k,r}(\tau,u)$ are entire functions satisfying the following properties:
 \begin{description}
  \item[(a)] $ \wh{h}_{k,r}(\tau; u)  + (-1)^{r+1} \,  \wh{h}_{k,r}(\tau; u+1) = \displaystyle\frac{2}{\sqrt{-\i \tau}} 
      \sum\limits_{p=1}^{k-1} \wh{{\CS}}^{(k)}_{r,p} \,  e^{\i \pi k (u+p/k)^2/2\tau}$.
  \item[(b)] $ \wh{h}_{k,r}(\tau; u)  + e^{-\pi \i k u - \pi \i k \tau/2} \  \wh{h}_{k,k-r}(\tau; u+\tau) =
  2\, e^{-\pi \i r u - \pi \i r^2 \tau/ 2k}$ for $r = 1, \ldots, k-1$.
 \item[(c)] $\wh{h}_{k,r}(\tau; u) = \wh{h}_{k,r}(\tau; -u)$.
\end{description}
\end{proposition}
\begin{proof} \hfill

\begin{description}
\item[(a)] We write $\wh{h}_{k,r}(\tau; u)  + (-1)^{r+1} \,  \wh{h}_{k,r}(\tau; u+1)$ as
\begin{equation}
\i \int\limits_{\IR} \dd x\ e^{\pi \i \, k \tau x^2/2} \, e^{- \pi k u x} \lp 1+ (-1)^{r+1}  e^{- \pi k x} \rp \lp  \frac{1}{1 - e^{\pi x - \pi \i r/k} } - \frac{1}{1 - e^{\pi x + \pi \i r/k} } \rp.
\end{equation}
Note that 
\begin{equation}
\frac{1+ (-1)^{r+1}  e^{- \pi k x}}{1 - e^{\pi x - \pi \i r/k}}
	 = - e^{- \pi x + \pi \i r/k} \, \frac{1- \lp  e^{-\pi x +\pi \i r/k} \rp^k}{1 - e^{-\pi x +\pi \i r/k}}
	 = - \sum_{p=1}^k  \lp  e^{-\pi x +\pi \i r/k} \rp^p
\end{equation}
and similarly
\begin{equation}
\frac{1+ (-1)^{r+1}  e^{- \pi k x}}{1 - e^{\pi x + \pi \i r/k}}
	= - \sum_{p=1}^k  \lp  e^{-\pi x -\pi \i r/k} \rp^p.
\end{equation}
$p=k$ terms cancel each other and for $\wh{h}_{k,r}(\tau; u)  + (-1)^{r+1} \,  \wh{h}_{k,r}(\tau; u+1) $ we eventually get
\begin{align}
		\i \int\limits_{\IR} \dd x\ e^{\pi \i \, k \tau x^2/2} \, 
		e^{- \pi k u x} \, &\sum_{p=1}^{k-1} 
		 e^{-\pi x p} \lp  e^{-\pi \i r p /k} - e^{\pi \i r p /k} \rp \\ &= 
		 2  \sum_{p=1}^{k-1} \sin \frac{\pi r p}{k} 
		  \int\limits_{\IR} \dd x\ e^{\pi \i \, k \tau x^2/2} e^{- \pi k (u+p/k) x}.
\end{align}
Now we can easily compute the integral to find our result.

\item[(b)] We start by noting that
\begin{equation}
- \i \int\limits_{\IR + \i} \dd x\ e^{\pi \i \, k \tau x^2/2} \, e^{- \pi k u x}
   \lp  \frac{1}{1 - e^{\pi x - \pi \i r/k} } - \frac{1}{1 - e^{\pi x + \pi \i r/k} } \rp
\end{equation}
is simply 
\begin{equation}
e^{-\pi \i k u - \pi \i k \tau/2} \  \wh{h}_{k,k-r}(\tau; u+\tau).
\end{equation}
Using this, we can express $\wh{h}_{k,r}(\tau; u)  + e^{-\pi \i k u - \pi \i k \tau/2} \  \wh{h}_{k,k-r}(\tau; u+\tau) $ as
\begin{equation}
 \i \lp\int \limits_{\IR} -  \int\limits_{\IR + \i} \rp \dd x \ e^{\pi \i \, k \tau x^2/2} \, 
 e^{- \pi k u x}
   \lp  \frac{1}{1 - e^{\pi x - \pi \i r/k} } - \frac{1}{1 - e^{\pi x + \pi \i r/k} } \rp.
\end{equation}
This integral can be easily computed as
\begin{equation}
2 \pi \i \,   \underset{x = \i r / k}{\mathrm{Res}} 
\lp \frac{\i \, e^{\pi \i \, k \tau x^2/2}\, e^{- \pi k u x}}{1 - e^{\pi x - \pi \i r/k} } \rp = 
2\, e^{-\pi \i r u - \pi \i r^2 \tau/ 2k}. 
\end{equation}

\item[(c)] Changing the integration variable $x \to -x$ we get
\begin{equation}\label{eq:B10c}
  \wh{h}_{k,r}(\tau; -u) 
     = \i \int\limits_{\IR} \dd x\ e^{\pi \i \, k \tau x^2/2} \, e^{- \pi k u x}
   \lp  \frac{1}{1 - e^{-\pi x - \pi \i r/k} } - \frac{1}{1 - e^{-\pi x + \pi \i r/k} } \rp.
\end{equation}
Using
\begin{equation}
\lp \frac{1}{1 - e^{-\pi x - \pi \i r/k} } -1 \rp + \lp 1 - \frac{1}{1 - e^{-\pi x + \pi \i r/k} }  \rp
 = - \frac{1}{1 - e^{\pi x +  \pi \i r/k} } + \frac{1}{1 - e^{\pi x - \pi \i r/k} } 
\end{equation}
gives us $\wh{h}_{k,r}(\tau; u)$ from equation (\ref{eq:B10c}).
\end{description}
\end{proof}

\begin{remark}
$u \to \wh{h}_{k,r}(\tau,u)$ functions are the only entire functions satisfying parts (a) and (b) of proposition \ref{prop:hath_ell}. To see this suppose there are two entire functions, 
$\wh{h}^{(1)}_{k,r}(\tau; u)$ and $\wh{h}^{(2)}_{k,r}(\tau; u)$, satisfying these two properties. Then their difference 
$\wh{f}_{k,r}(\tau; u) \equiv \wh{h}^{(1)}_{k,r}(\tau; u) - \wh{h}^{(2)}_{k,r}(\tau; u)$
satisfies
\begin{equation}
\wh{f}_{k,r}(\tau; u)  - (-1)^{r} \,  \wh{f}_{k,r}(\tau; u+1) =0
\quad \mathrm{and} \quad
\wh{f}_{k,r}(\tau; u)  + e^{-\pi \i k u - \pi \i k \tau/2} \  \wh{f}_{k,k-r}(\tau; u+\tau) = 0
\end{equation}
 for $r = 1, \ldots, k-1$. The second equation then gives
 \begin{equation}
 \wh{f}_{k,r}(\tau; u)  - e^{-2\pi \i k u - 2\pi \i k \tau} \  \wh{f}_{k,r}(\tau; u+2\tau) = 0.
 \end{equation}
 By considering $\wh{f}_{k,r}(\tau; u + m + 2n\tau)$ with integer $m,n$ and restricting u to a parallelogram formed by $1$ and $2\tau$ we see that $u \to \wh{f}_{k,r}(\tau; u)$ is a bounded function which goes to zero as $u \to i \infty$.  Liouville's theorem implies then  $\wh{f}_{k,r}(\tau; u)=0$.
\end{remark}

\begin{definition}
 For $u \in \IC$ and $\tau \in \IH$, \cite{zwegers2008mock} defines
 \begin{equation}
  R_{k,r}(\tau;u) = \sum_{n \,\equiv\, r\, \mathrm{mod} 2k} \lb \mathrm{sgn} \lp n+\frac{1}{2} \rp - 
  E \lp (n + 2 k u_2/\tau_2) \sqrt{\tau_2/k} \rp    \rb e^{- \pi \i n^2 \tau /2k - 2 \pi \i n u},
 \end{equation}
 where
 \begin{equation}
 E(z) \equiv 2 \int\limits_0^z \dd u \, e^{-\pi u^2}
 \end{equation}
 We also define
 \begin{equation}
  \wh{R}_{k,r}(\tau; u) \equiv  R_{k,r}(\tau;u/2) - R_{k,-r}(\tau;u/2).
 \end{equation}
 More explicitly, $  \wh{R}_{k,r}(\tau; u)$ is
 \begin{equation}\label{eq:Rhat_defn}
 \lp \sum_{n \in r + 2k \ZZ} -  \sum_{n \in -r + 2k \ZZ} \rp
  \lb \mathrm{sgn} \lp n+\frac{1}{2} \rp - 
  E \lp (n +  k u_2/\tau_2) \sqrt{\tau_2/k} \rp    \rb e^{- \pi \i n^2 \tau /2k - \pi \i n u}.
 \end{equation}
\end{definition}
\begin{remark}
For $r = 1, \ldots, k-1$ we can write $ \wh{R}_{k,r}(\tau; 0)$ as
\begin{equation}
 \wh{R}_{k,r}(\tau; 0) = 2 \sum_{n \in r + 2k \ZZ} 
  \lb \mathrm{sgn} \lp n \rp - 
  E \lp n \sqrt{\tau_2/k} \rp    \rb e^{- \pi \i n^2 \tau /2k}.
\end{equation}
Since $E(z) = \mathrm{sgn} \lp z \rp \lb 1 -  \Gamma \lp \frac{1}{2}, \pi z^2  \rp /\sqrt{\pi}\rb$, we see that
\begin{equation}\label{eq:Rhat_Eichler_Int}
\wh{R}_{k,r}(\tau; 0) = \lp \pi k \rp^{-1/2} S^*_{k,r}(\tau),
\end{equation}
where $S^*_{k,r}(\tau)$ is the non-holomorphic Eichler integral 
of $S_{k,r}(\tau)$, which solves the equation
\begin{equation}
\lp 4 \pi \tau_2 \rp^{1/2}  \frac{\partial S^*_{k,r}(\tau)}{\partial \bar{\tau}}  = - 2 \pi \i\, \overline{S_{k,r}(\tau)}.
\end{equation}
\end{remark}

 It is easy to see from this definition that 
\begin{itemize}
\item $R_{k,r}(\tau; u) = R_{k,r+2k}(\tau; u)$ and  $\wh{R}_{k,r}(\tau; u) = \wh{R}_{k,r+2k}(\tau; u)$,
\item $\wh{R}_{k,r}(\tau; u) = -\, \wh{R}_{k,-r}(\tau; u)$,
\item For $r = 0\, (\mathrm{mod } k)$, $\wh{R}_{k,r}(\tau; u) = 0$.
\end{itemize}

\begin{proposition}\label{prop:Rhat_ell}
For $r = 1, \ldots, k-1$
\begin{description}
  \item[(a)] $\wh{R}_{k,r}(\tau; u) = \wh{R}_{k,r}(\tau; -u)$.
  \item[(b)] $ \wh{R}_{k,r}(\tau; u)  + (-1)^{r+1} \,  \wh{R}_{k,r}(\tau; u+1) = 0$.
  \item[(c)] $ \wh{R}_{k,r}(\tau; u)  + e^{-\pi \i k u - \pi \i k \tau/2} \  \wh{R}_{k,k-r}(\tau; u+\tau) =
  2\, e^{-\pi \i r u - \pi \i r^2 \tau/ 2k}$.
\end{description}
\end{proposition}
\begin{proof} \hfill
\begin{description}
\item[(a)] This property follows by changing the dummy summation variable $n \to -n$. 

\item[(b)] The only term in the equation (\ref{eq:Rhat_defn}) that is affected by the $u \to u+1$ transformation is $e^{- \pi \i n u}$, which produces a factor of 
$e^{- \pi \i (\pm r + 2k\ZZ)} = (-1)^r$.

\item[(c)] We write $ e^{-\pi \i k u - \pi \i k \tau/2} \  \wh{R}_{k,k-r}(\tau; u+\tau)$ as
\begin{align}
 \lp \sum_{n \in k - r + 2k \ZZ} -  \sum_{n \in r + k + 2k \ZZ} \rp
  \Big[ &\mathrm{sgn} \lp n+\frac{1}{2} \rp - 
  E \lp (n +  k (u_2+\tau_2)/\tau_2) \sqrt{\tau_2/k} \rp    \Big]  \notag \\
  &\times e^{-\pi \i k u - \pi \i k \tau/2}  e^{- \pi \i n^2 \tau /2k - \pi \i n (u +\tau)}.
 \end{align}
 Shifting the dummy summation variable as $n \to n-k$ we get
 \begin{equation}
  \lp \sum_{n \in - r + 2k \ZZ} -  \sum_{n \in r + 2k \ZZ} \rp
  \lb \mathrm{sgn} \lp n-k+\frac{1}{2} \rp - 
  E \lp (n +  \frac{k u_2}{\tau_2}) \sqrt{\frac{\tau_2}{k}} \rp    \rb e^{- \frac{\pi \i n^2 \tau}{2k} - \pi \i n u}.
 \end{equation}
 We can write this expression as
 \begin{equation}
 - \wh{R}_{k,r}(\tau; u) + \lp \sum_{n \in r + 2k \ZZ} -  \sum_{n \in - r + 2k \ZZ} \rp
 				\lb \mathrm{sgn}\lp n+\frac{1}{2} \rp -  \mathrm{sgn} \lp n-k+\frac{1}{2} \rp \rb
 				e^{- \frac{\pi \i n^2 \tau}{2k} - \pi \i n u}.
 \end{equation}
 $\lb \mathrm{sgn}( n+\frac{1}{2} ) -  \mathrm{sgn} \lp n-k+\frac{1}{2} \rp \rb$ term is only nonzero (and is equal to $2$) for integers in the interval $-\frac{1}{2} < n < k - \frac{1}{2}$ and hence it picks the $n=+r$ term in the sum giving us the desired result.
\end{description}
\end{proof}

\begin{proposition} \label{prop:Rhat_mod}
For $r = 1, \ldots, k-1$
\begin{description}
 \item[(a)] $ \wh{R}_{k,r}(\tau+1; u) = e^{- \pi \i r^2 /2k} \, \wh{R}_{k,r}(\tau+1; u)$.
 \item[(b)] $\wh{R}_{k,r}(\tau; u) + \displaystyle
 \frac{ e^{\pi \i k u^2/2\tau}}{\sqrt{-\i \tau}} \, 
             \sum_{p=1}^{k-1} \wh{{\CS}}^{(k)}_{r,p} \, \wh{R}_{k,p}\lp -\frac{1}{\tau}; \frac{u}{\tau} \rp
             = \wh{h}_{k,r}(\tau; u).$
\end{description}
\end{proposition}
\begin{proof}\hfill

\noindent (a) In equation (\ref{eq:Rhat_defn}), $\tau \to \tau+1$ transformation only affects the
 $e^{- \pi \i n^2 \tau /2k }$ term which produces a factor of 
$e^{- \pi \i (\pm r + 2k \ZZ)^2 /2k } = e^{- \pi \i r^2 /2k }$.

\noindent (b) Let us define $\wt{h}_{k,r}(\tau; u)$ as 
\begin{equation}
\wt{h}_{k,r}(\tau; u) \equiv
\wh{R}_{k,r}(\tau; u) + \frac{e^{\pi \i k u^2/2\tau}}{\sqrt{-\i \tau}}\, 
             \sum_{p=1}^{k-1} \wh{{\CS}}^{(k)}_{r,p} \, 
             \wh{R}_{k,p}\lp -\frac{1}{\tau}; \frac{u}{\tau} \rp.
\end{equation}
Our proof will start by showing that $\wt{h}_{k,r}(\tau; u)$ has the same behavior under $u \to u+1$ and $u \to u+\tau$ as $\wh{h}_{k,r}(\tau; u)$ (see parts (a) and (b) of proposition \ref{prop:hath_ell}). Since $u \to \wh{h}_{k,r}(\tau; u)$ are the unique entire functions having this behavior we will finish our proof by showing that $u \to \wt{h}_{k,r}(\tau; u)$ are entire functions.

We start with $ \wt{h}_{k,r}(\tau; u)  + e^{-\pi \i k u - \pi \i k \tau/2} \  \wt{h}_{k,k-r}(\tau; u+\tau) $.
There are two contributions to this object
which are
\begin{equation}
\wh{R}_{k,r}(\tau; u)  + e^{-\pi \i k u - \pi \i k \tau/2} \  \wh{R}_{k,k-r}(\tau; u+\tau)
\end{equation}
and
\begin{equation}
 \frac{e^{\pi \i k u^2/2\tau}}{\sqrt{-\i \tau}}\,  \sum_{p=1}^{k-1} \wh{{\CS}}^{(k)}_{r,p} \, 
             \wh{R}_{k,p}\lp -\frac{1}{\tau}; \frac{u}{\tau} \rp
   + e^{-\pi \i k u - \pi \i k \tau/2} \, \frac{e^{\pi \i k (u+\tau)^2/2\tau}}{\sqrt{-\i \tau} }
   	\,  \sum_{p=1}^{k-1} \wh{{\CS}}^{(k)}_{k-r,p} \, 
   	\wh{R}_{k,p}\lp -\frac{1}{\tau}; \frac{u+\tau}{\tau} \rp.
\end{equation}
The first factor gives $  2\, e^{-\pi \i r u - \pi \i r^2 \tau/ 2k}$ by part (c) of proposition \ref{prop:Rhat_ell} and the second one is zero by 
$\wh{{\CS}}^{(k)}_{k-r,p} = (-1)^{p+1} \, \wh{{\CS}}^{(k)}_{r,p} $ and using part (b) of proposition \ref{prop:Rhat_ell}.

Similarly, we can separate 
$\wt{h}_{k,r}(\tau; u)  + (-1)^{r+1} \,  \wt{h}_{k,r}(\tau; u+1)$
into the sum of two factors. The first one is $ \wh{R}_{k,r}(\tau; u)  + (-1)^{r+1} \,  \wh{R}_{k,r}(\tau; u+1)$ which is zero by part (b) of proposition \ref{prop:Rhat_ell}.
The second contribution is 
\begin{equation}
\frac{e^{\pi \i k u^2/2\tau}}{\sqrt{-\i \tau}}\,  \sum_{p=1}^{k-1} \wh{{\CS}}^{(k)}_{r,p} \, 
             \wh{R}_{k,p}\lp -\frac{1}{\tau}; \frac{u}{\tau} \rp
   + (-1)^{r+1} \, \frac{e^{\pi \i k (u+1)^2/2\tau}}{\sqrt{-\i \tau} }
   	\,  \sum_{p=1}^{k-1} \wh{{\CS}}^{(k)}_{r,p} \, 
   	\wh{R}_{k,p}\lp -\frac{1}{\tau}; \frac{u+1}{\tau} \rp.
\end{equation}
Using $(-1)^{r+1} \, \wh{{\CS}}^{(k)}_{r,p} = \wh{{\CS}}^{(k)}_{r,k-p}$ and then changing the dummy summation variable of the second factor as $p \to k-p$ we get
\begin{equation}
\frac{e^{\pi \i k u^2/2\tau}}{\sqrt{-\i \tau}}\,  \sum_{p=1}^{k-1} \wh{{\CS}}^{(k)}_{r,p} \, \lb 
             \wh{R}_{k,p}\lp -\frac{1}{\tau}; \frac{u}{\tau} \rp
      				+ e^{\pi \i k u/\tau} e^{\pi \i k/2\tau}  
      							\wh{R}_{k,k-p}\lp -\frac{1}{\tau}; \frac{u}{\tau} + \frac{1}{\tau} \rp
             \rb. 
\end{equation}
The part (c) of proposition \ref{prop:Rhat_ell} then tells us that the factor in square brackets is just $2\,e^{\pi \i p u/\tau} e^{\pi \i p^2/ 2k\tau} $. This finally gives
\begin{equation} 
\wt{h}_{k,r}(\tau; u)  + (-1)^{r+1} \,  \wt{h}_{k,r}(\tau; u+1) = \displaystyle\frac{2}{\sqrt{-\i \tau}} 
      \sum\limits_{p=1}^{k-1} \wh{{\CS}}^{(k)}_{r,p}\  e^{\i \pi k (u+p/k)^2/2\tau}.
\end{equation}

Our final task is to show that $u \to \wt{h}_{k,r}(\tau; u)$ are entire functions. 
It is straightforward to work out
\begin{equation}\label{eq:htilde_analyticity1}
\frac{\partial}{\partial\, \overline{u}} \, 
  \wh{R}_{k,r}\lp \tau; u \rp = 
  \i \sqrt{\frac{k}{\tau_2}} \, e^{- \pi k u_2^{\, 2}/ \tau_2}
  \lb   
  \vth_{k,r} ( -\overline{\tau}, \overline{u}/2 ) -
  \vth_{k,-r} (-\overline{\tau}, \overline{u}/2) 
  \rb
\end{equation}
and
\begin{equation}
\frac{\partial}{\partial\, \overline{u}} \, 
  \wh{R}_{k,p}\lp -\frac{1}{\tau}; \frac{u}{\tau} \rp = 
   \sqrt{\frac{k}{\tau_2}} \, \frac{\sqrt{- \i \, \tau}}{\sqrt{\i \, \overline{\tau}}} 
    \, e^{ \pi k (u \overline{\tau} - \overline{u} \tau)^2/ 4 \tau_2 | \tau |^2}
  \lb   
  \vth_{k,p} \lp \frac{1}{\overline{\tau}}, -\frac{\overline{u}}{2\overline{\tau}} \rp -
  \vth_{k,-p} \lp \frac{1}{\overline{\tau}}, -\frac{\overline{u}}{2\overline{\tau}} \rp 
  \rb.
\end{equation}
Then we rewrite
\begin{equation}
\sum_{p=1}^{k-1} \wh{{\CS}}^{(k)}_{r,p} \, 
\frac{\partial}{\partial\, \overline{u}} \,
             \wh{R}_{k,p}\lp -\frac{1}{\tau}; \frac{u}{\tau} \rp =
 \i \, \sum_{p \in \ZZ_{2k}} {\CS}^{(k)}_{r,p} \, 
\frac{\partial}{\partial\, \overline{u}} \,
             \wh{R}_{k,p}\lp -\frac{1}{\tau}; \frac{u}{\tau} \rp
\end{equation}
as
\begin{equation}
 \i \, \sqrt{\frac{k}{\tau_2}} \, \frac{\sqrt{- \i \, \tau}}{\sqrt{\i \, \overline{\tau}}} 
    \, e^{ \pi k (u \overline{\tau} - \overline{u} \tau)^2/ 4 \tau_2 | \tau |^2}
  \lb   
   \sum_{p \in \ZZ_{2k}} {\CS}^{(k)*}_{-r,p} \, 
  \vth_{k,p} \lp \frac{1}{\overline{\tau}}, -\frac{\overline{u}}{2\overline{\tau}} \rp -
   \sum_{p \in \ZZ_{2k}} {\CS}^{(k)*}_{r,-p} \, 
  \vth_{k,-p} \lp \frac{1}{\overline{\tau}}, -\frac{\overline{u}}{2\overline{\tau}} \rp 
  \rb.
\end{equation}
Using equation (\ref{eq:theta_kr_modS}) this expression is equal to
\begin{equation}\label{eq:htilde_analyticity2}
- \i \, \sqrt{\frac{k}{\tau_2}} \, \sqrt{- \i \, \tau} 
    \, e^{ \pi k (u \overline{\tau} - \overline{u} \tau)^2/ 4 \tau_2 | \tau |^2}
     \, e^{ - \pi \i k {\overline{u}}^2 / 2 \overline{\tau} }
   \lb   
  \vth_{k,r} ( -\overline{\tau}, \overline{u}/2 ) -
  \vth_{k,-r} (-\overline{\tau}, \overline{u}/2) 
  \rb.
\end{equation}
Combining equations (\ref{eq:htilde_analyticity1}) and (\ref{eq:htilde_analyticity2}) we get
\begin{equation}
\frac{\partial}{\partial\, \overline{u}} \wh{R}_{k,r}(\tau; u)
 + \frac{e^{\pi \i k u^2/2\tau}}{\sqrt{-\i \tau}}\, 
             \sum_{p=1}^{k-1} \wh{{\CS}}^{(k)}_{r,p} \, 
          \frac{\partial}{\partial\, \overline{u}}
             \wh{R}_{k,p}\lp -\frac{1}{\tau}; \frac{u}{\tau} \rp = 0
\end{equation}
as was to be shown.     
\end{proof}

We finally relate $\wh{R}_{k,r}$ and $\wh{h}_{k,r}$ to the $\mu^{k,j,j'}$'s modular transformation properties.
\begin{proposition}\label{prop:mu_Sdeficit}
For $j,j' = 1, \ldots, k-1$ and $j = j' \, (\mathrm{mod}\, 2)$
 \begin{align}
\frac{1}{2} \Big[   \mu^{k,j,j'}\lp \tau; v, u, w \rp  
									&+ \mu^{k,k-j,k-j'}\lp \tau; v, u, w \rp  \Big] \notag \\
									&+  
  \frac{e^{\pi \i (u-v)^2 /\tau}\, e^{-\pi \i w^2 (1-2/k)/\tau}}{\sqrt{-\i \tau}}  \sum_{p,p'=1}^{k-1}  
   \wh{{\CS}}^{(k)}_{j,p}  \, \wh{{\CS}}^{(k)}_{j',p'}\,
  \mu^{k,p,p'}\Big( -\frac{1}{\tau}; \frac{v}{\tau}, \frac{u}{\tau}, \frac{w}{\tau} \Big)
  \notag \\
 = &\frac{1}{4} \Big[
 \  \chi^{(k-2)}_{j-1}(\tau; u-v+w) \ \wh{h}_{k,j'}(\tau;u-v+w(1-2/k)) \notag\\
 &+ \chi^{(k-2)}_{k-j-1}(\tau; u-v+w) \ \wh{h}_{k,k-j'}(\tau;u-v+w(1-2/k)) \ 
 \Big].
 \end{align}
\end{proposition}
\begin{remark}
We will find it convenient to define 
\begin{equation}
\mu_{\mathrm{sym}}^{k,j,j'}\lp \tau; v, u, w \rp = 
\frac{1}{2} \Big[   \mu^{k,j,j'}\lp \tau; v, u, w \rp  + \mu^{k,k-j,k-j'}\lp \tau; v, u, w \rp  \Big].
\end{equation}
Since $ \wh{{\CS}}^{(k)}_{j,p} = (-1)^{j+1} \wh{{\CS}}^{(k)}_{j,k-p}$ and since we are considering the $j = j' \, (\mathrm{mod}\, 2)$ case (as otherwise $ \mu^{k,j,j'}$ is zero) we have
\begin{equation} \label{eq:prop_proof3}
 \sum_{p,p'=1}^{k-1}  \wh{{\CS}}^{(k)}_{j,p} \,  \wh{{\CS}}^{(k)}_{j',p'} \,
  \mu_{\mathrm{sym}}^{k,p,p'}\Big( -\frac{1}{\tau}; \frac{v}{\tau}, \frac{u}{\tau}, \frac{w}{\tau} \Big) = 
   \sum_{p,p'=1}^{k-1}  \wh{{\CS}}^{(k)}_{j,p} \,  \wh{{\CS}}^{(k)}_{j',p'}\,
 \mu^{k,p,p'}\Big( -\frac{1}{\tau}; \frac{v}{\tau}, \frac{u}{\tau}, \frac{w}{\tau} \Big).
\end{equation}
Furthermore, since $ \wh{\Omega}^{(k,d)}_{j,j'} =  \wh{\Omega}^{(k,d)}_{k-j,k-j'}$ we have
\begin{equation}
\sum_{j,j'=1}^{k-1} \wh{\Omega}^{(k,d)}_{j,j'} \mu_{\mathrm{sym}}^{k,j,j'}\lp \tau; v, u, w \rp = 
  \sum_{j,j'=1}^{k-1} \wh{\Omega}^{(k,d)}_{j,j'} \mu^{k,j,j'}\lp \tau; v, u, w \rp.
\end{equation}
Using (\ref{eq:mu_Z2symm}) we also see that 
\begin{equation}
\mu_{\mathrm{sym}}^{k,j,j'}\lp \tau; u, u, 0 \rp = \mu^{k,j,j'}\lp \tau; u, u, 0 \rp.
\end{equation}
Then with this definition, the proposition above can be equivalently expressed as
\begin{align}
 \mu_{\mathrm{sym}}^{k,j,j'}\lp \tau; v, u, w \rp  
									&+  \frac{e^{\pi \i (u-v)^2 /\tau}\, 
									e^{-\pi \i w^2 (1-2/k)/\tau}}{\sqrt{-\i \tau}}  \sum_{p,p'=1}^{k-1}  
  									 \wh{{\CS}}^{(k)}_{j,p}  \, \wh{{\CS}}^{(k)}_{j',p'} \,
  \mu_{\mathrm{sym}}^{k,p,p'}\Big( -\frac{1}{\tau}; \frac{v}{\tau}, \frac{u}{\tau}, \frac{w}{\tau} \Big)
  \notag \\
 = &\frac{1}{4} \Big[
 \  \chi^{(k-2)}_{j-1}(\tau; u-v+w) \ \wh{h}_{k,j'}(\tau;u-v+w(1-2/k)) \notag\\
 &+ \chi^{(k-2)}_{k-j-1}(\tau; u-v+w) \ \wh{h}_{k,k-j'}(\tau;u-v+w(1-2/k)) \ 
 \Big].
 \end{align}
\end{remark}

\begin{proof}
Using lemma \ref{lem:mu_z_shift} with $z \to z/\tau$, $v \to v/\tau$, $u \to u/\tau$, 
$w \to w/\tau$ and $\tau \to -1/\tau$ we get
\begin{align}
  \sum_{p,p'=1}^{k-1}   &\wh{{\CS}}^{(k)}_{j,p}  \, \wh{{\CS}}^{(k)}_{j',p'} \,
  \lb \mu_{\mathrm{sym}}^{k,p,p'}\Big( -\frac{1}{\tau}; \frac{v+z}{\tau},
  															 \frac{u+z}{\tau}, \frac{w}{\tau} \Big)
  -  \mu_{\mathrm{sym}}^{k,p,p'}\Big( -\frac{1}{\tau}; \frac{v}{\tau},
  															 \frac{u}{\tau}, \frac{w}{\tau} \Big) \rb   \notag \\
  = \  &\frac{\i \, \eta \lp -\frac{1}{\tau} \rp^3  \,  
  \theta_{11} \lp  -\frac{1}{\tau}; \frac{z}{\tau} \rp  \, 
  \theta_{11}  \lp  -\frac{1}{\tau}; \frac{v+u+z}{\tau} \rp}
 { \theta_{11} \lp  -\frac{1}{\tau}; \frac{v}{\tau} \rp \, 
  \theta_{11} \lp  -\frac{1}{\tau}; \frac{u}{\tau} \rp \,
    \theta_{11} \lp  -\frac{1}{\tau}; \frac{v+z}{\tau} \rp \,
    \theta_{11} \lp  -\frac{1}{\tau}; \frac{u+z}{\tau} \rp} \notag \\
    \label{eq:prop_proof1}
  &\times \frac{1}{2} \,\sum_{p,p'=1}^{k-1}   \wh{{\CS}}^{(k)}_{j,p}  \, \wh{{\CS}}^{(k)}_{j',p'} \,
  \lb  C^{p-1,k}_{p'}\bmat{1 \\ 1} \lp -\frac{1}{\tau};\frac{w}{\tau} \rp
  				 + \  C^{k-p-1,k}_{k-p'}\bmat{1 \\ 1} \lp -\frac{1}{\tau};\frac{w}{\tau} \rp \rb \, .
\end{align}
The first factor
\begin{equation}
\frac{\i \, \eta \lp -\frac{1}{\tau} \rp^3  \,  
  \theta_{11} \lp  -\frac{1}{\tau}; \frac{z}{\tau} \rp  \, 
  \theta_{11}  \lp  -\frac{1}{\tau}; \frac{v+u+z}{\tau} \rp}
 { \theta_{11} \lp  -\frac{1}{\tau}; \frac{v}{\tau} \rp \, 
  \theta_{11} \lp  -\frac{1}{\tau}; \frac{u}{\tau} \rp \,
    \theta_{11} \lp  -\frac{1}{\tau}; \frac{v+z}{\tau} \rp \,
    \theta_{11} \lp  -\frac{1}{\tau}; \frac{u+z}{\tau} \rp} 
\end{equation}
can be rewritten using $\eta$ and $\theta_{11}$'s modular transformation as
\begin{equation}
- \, \sqrt{- \i \tau} \, e^{- \,\pi \i (u-v)^2 /\tau} \,  
\frac{\i \, \eta \lp \tau \rp^3  \,  
  \theta_{11} \lp  \tau; z \rp  \, 
  \theta_{11}  \lp  \tau; v+u+z \rp}
 { \theta_{11} \lp \tau; v \rp \, 
  \theta_{11} \lp  \tau; u \rp \,
    \theta_{11} \lp \tau; v+z \rp \,
    \theta_{11} \lp \tau; u+z \rp} \, .
\end{equation}
We can rewrite the second factor,
\begin{equation}
\frac{1}{2} \,\sum_{p,p'=1}^{k-1}   \wh{{\CS}}^{(k)}_{j,p}  \, \wh{{\CS}}^{(k)}_{j',p'} \,
  \lb  C^{p-1,k}_{p'}\bmat{1 \\ 1} \lp -\frac{1}{\tau};\frac{w}{\tau} \rp
  				 + \  C^{k-p-1,k}_{k-p'}\bmat{1 \\ 1} \lp -\frac{1}{\tau};\frac{w}{\tau} \rp \rb 
\end{equation} 
using equations (\ref{eq:su2u1_prop2}), (\ref{eq:su2u1_prop3}), (\ref{eq:Cfnc_modS}) and properties of $ \wh{{\CS}}^{(k)}_{r,r'}$ and $ {\CS}^{(k)}_{r,r'}$ to get
\begin{align}
\frac{1}{2} \,\sum_{p,p'=1}^{k-1}   &\wh{{\CS}}^{(k)}_{j,p}  \, \wh{{\CS}}^{(k)}_{j',p'} \,
  \lb  C^{p-1,k}_{p'}\bmat{1 \\ 1} \lp -\frac{1}{\tau};\frac{w}{\tau} \rp
  				 - \  C^{p-1,k}_{-p'}\bmat{1 \\ 1} \lp -\frac{1}{\tau};\frac{w}{\tau} \rp \rb 
  				   																			\notag  \\
&= \frac{\i}{2} \,\sum_{p,p'=1}^{k-1}  \wh{{\CS}}^{(k)}_{j,p}  \, 
\lb {\CS}^{(k)}_{j',p'} - {\CS}^{(k)}_{j',- p'} \rb \,
  \lb  C^{p-1,k}_{p'}\bmat{1 \\ 1} \lp -\frac{1}{\tau};\frac{w}{\tau} \rp
  				 - \  C^{p-1,k}_{-p'}\bmat{1 \\ 1} \lp -\frac{1}{\tau};\frac{w}{\tau} \rp \rb 
  				   																			\notag  \\
&= \frac{\i}{2}  \,\sum_{p=1}^{k-1}  \,\sum_{p'=-k+1}^{k}   \wh{{\CS}}^{(k)}_{j,p}  \, 
\lb {\CS}^{(k)}_{j',p'}  \, C^{p-1,k}_{p'}\bmat{1 \\ 1} \lp -\frac{1}{\tau};\frac{w}{\tau} \rp
  	- \, {\CS}^{(k)}_{-\, j',p'} \, C^{p-1,k}_{p'}\bmat{1 \\ 1} \lp -\frac{1}{\tau};\frac{w}{\tau} \rp \rb 
  				   																			\notag  \\
&=\frac{1}{2} \, e^{\pi \i w^2 (1-2/k)/\tau}
\lb \, C^{j-1,k}_{j'}\bmat{1 \\ 1} \lp \tau; w \rp
  	-  \, C^{j-1,k}_{-j'}\bmat{1 \\ 1} \lp \tau; w \rp \rb     \notag \\
&= \frac{1}{2} \, e^{\pi \i w^2 (1-2/k)/\tau} 
\lb \, C^{j-1,k}_{j'}\bmat{1 \\ 1} \lp \tau; w \rp
  	+  \, C^{k-j-1,k}_{k-j'}\bmat{1 \\ 1} \lp \tau; w \rp \rb.
\end{align}
Combining these two factors and using lemma \ref{lem:mu_z_shift} once more we see that
the left hand side of equation (\ref{eq:prop_proof1}) is equal to
\begin{equation}
- \, \sqrt{- \i \tau} \, e^{- \,\pi \i (u-v)^2 /\tau} \,
 e^{\pi \i w^2 (1-2/k)/\tau} \,
 \lb \mu_{\mathrm{sym}}^{k,p,p'} \lp \tau; v+z,  u+z, w \rp
  -  \mu_{\mathrm{sym}}^{k,p,p'}  \lp \tau; v, u, w \rp \rb.
\end{equation}
Therefore, 
\begin{equation}\label{eq:H^kjj_defn}
\mu_{\mathrm{sym}}^{k,j,j'}\lp \tau; v, u, w \rp  
									+  \frac{e^{\pi \i (u-v)^2 /\tau}\, 
									e^{-\pi \i w^2 (1-2/k)/\tau}}{\sqrt{-\i \tau}}  \sum_{p,p'=1}^{k-1}  
  									 \wh{{\CS}}^{(k)}_{j,p}  \, \wh{{\CS}}^{(k)}_{j',p'} \,
  \mu_{\mathrm{sym}}^{k,p,p'}\Big( -\frac{1}{\tau}; \frac{v}{\tau}, \frac{u}{\tau}, \frac{w}{\tau} \Big)
\end{equation}
depends on $u$ and $v$ only through $u-v$. Let us call it $\frac{1}{2} H^{k,j,j'} (\tau; u-v,w)$.

Looking at its definition in equation (\ref{eq:H^kjj_defn}),
 $u \to H^{k,j,j'} (\tau; u-v,w)$ can have poles only at 
$u \in \ZZ \tau + \ZZ$. However, since $H^{k,j,j'} (\tau; u-v,w)$'s dependence on $u$ is only 
through $u-v$, it can not have any poles at all. Therefore, $ z \to  H^{k,j,j'} (\tau; z,w)$ is an entire function. We will now find two properties for $H^{k,j,j'} (\tau; z,w)$ which will turn out 
to characterize these functions completely.

The first one of these two properties is
\begin{align}
e^{-\pi \i \tau - 2 \pi \i z}  H^{k,j,j'} \lp \tau; z+\tau, w \rp
&+  H^{k,j,j'} \lp \tau; z, w \rp
  =
 q^{-{j'}^2/4k} \lp y_z y_w^{1-2/k} \rp^{-{j'}/2} \chi^{(k-2)}_{j-1}(\tau; z+w)  \notag \\
 \label{eq:H^kjj_prop1}
 &+  q^{-{(k-j')}^2/4k} \lp y_z y_w^{1-2/k} \rp^{-{(k-j')}/2} \chi^{(k-2)}_{k-j-1}(\tau; z+w)  \, 
.
\end{align}
To see this, let us use the definition of $\frac{1}{2} H^{k,j,j'} \lp \tau; u-v, w \rp$ as given in equation (\ref{eq:H^kjj_defn}) to write $\frac{1}{2} H^{k,j,j'} \lp \tau; u-v+\tau, w \rp$ as
\begin{equation}
\mu_{\mathrm{sym}}^{k,j,j'}\lp \tau; v, u+\tau, w \rp  
									+  \frac{e^{\pi \i (u-v+\tau)^2 /\tau}\, 
									e^{-\pi \i w^2 (1-2/k)/\tau}}{\sqrt{-\i \tau}}  \sum_{p,p'=1}^{k-1}  
  									 \wh{{\CS}}^{(k)}_{j,p}  \, \wh{{\CS}}^{(k)}_{j',p'} \,
  \mu_{\mathrm{sym}}^{k,p,p'}\Big( -\frac{1}{\tau}; \frac{v}{\tau}, \frac{u}{\tau}+1, \frac{w}{\tau} \Big).
\end{equation}
The first factor, $\mu_{\mathrm{sym}}^{k,j,j'}\lp \tau; v, u+\tau, w \rp$, is equal to
\begin{align}
q^{1/2} y_u^{-1} y_v
\Big[ -\, &\mu_{\mathrm{sym}}^{k,j,j'}\lp \tau; v, u, w \rp
+
\frac{1}{2} \, q^{-j'^2/4k} \lp \frac{y_v}{y_u y_w^{1-2/k}} \rp^{j'/2} \chi_{j-1}^{(k-2)}(\tau; w+u-v) \notag \\
&+
\frac{1}{2} \,  q^{-(k-j')^2/4k} \lp \frac{y_v}{y_u y_w^{1-2/k}} \rp^{(k-j')/2} \chi_{k-j-1}^{(k-2)}(\tau; w+u-v) \Big]
\end{align}
by part (a) of Proposition \ref{prop:mu_algell2}. Then, using part (b) of Proposition \ref{prop:mu_algell1} we see that 
\begin{equation}
 \mu_{\mathrm{sym}}^{k,p,p'}\Big( -\frac{1}{\tau}; \frac{v}{\tau}, \frac{u}{\tau}+1, \frac{w}{\tau} \Big) = 
  - \, \mu_{\mathrm{sym}}^{k,p,p'}\Big( -\frac{1}{\tau}; \frac{v}{\tau}, \frac{u}{\tau}, \frac{w}{\tau} \Big).
\end{equation}
Finally, using that
$e^{\pi \i (u-v+\tau)^2 /\tau} = q^{1/2}\,  y_u^{-1} \, y_v \, e^{\pi \i (u-v)^2 /\tau}$
we rewrite $\frac{1}{2} H^{k,j,j'} \lp \tau; \tau+u-v, w \rp$ as
\begin{align}
q^{1/2} y_u^{-1} y_v
\Big[ -\, &\mu_{\mathrm{sym}}^{k,j,j'}\lp \tau; v, u, w \rp 
+ \frac{1}{2} \, q^{-j'^2/4k} \lp \frac{y_v}{y_u y_w^{1-2/k}} \rp^{j'/2} \chi_{j-1}^{(k-2)}(\tau; w+u-v) \notag \\
&+
\frac{1}{2} \,  q^{-(k-j')^2/4k} \lp \frac{y_v}{y_u y_w^{1-2/k}} \rp^{(k-j')/2} \chi_{k-j-1}^{(k-2)}(\tau; w+u-v)  \notag \\
&-  \frac{e^{\pi \i (u-v)^2 /\tau}\, 
									e^{-\pi \i w^2 (1-2/k)/\tau}}{\sqrt{-\i \tau}}  \sum_{p,p'=1}^{k-1}  
  									 \wh{{\CS}}^{(k)}_{j,p}  \, \wh{{\CS}}^{(k)}_{j',p'} \,
  \mu_{\mathrm{sym}}^{k,p,p'}\Big( -\frac{1}{\tau}; \frac{v}{\tau}, \frac{u}{\tau}, \frac{w}{\tau} \Big) \Big].
\end{align}
Combining the first and fourth terms contained in the brackets as 
$- \, \frac{1}{2} H^{k,j,j'} \lp \tau; u-v, w \rp$ and replacing $u$ with $v+z$ we get to the statement of equation (\ref{eq:H^kjj_prop1}).

The second property we would like to use is
\begin{align}
H^{k,j,j'} &\lp \tau; z+1, w \rp
+  H^{k,j,j'} \lp \tau; z, w \rp \notag\\
  =& \,
\chi^{(k-2)}_{j-1}(\tau; z+w) \, \frac{1}{\sqrt{-\i \tau}} \,
      \sum\limits_{p=1}^{k-1} \wh{{\CS}}^{(k)}_{j',p'}\  e^{\pi \i k (z+w(1-2/k)+p'/k)^2/2\tau}
 \notag \\
 \label{eq:H^kjj_prop2}
 &+  \chi^{(k-2)}_{k-j-1}(\tau; z+w) \, \frac{1}{\sqrt{-\i \tau}} \,
      \sum\limits_{p=1}^{k-1} \wh{{\CS}}^{(k)}_{k-j',p'}\  e^{\pi \i k (z+w(1-2/k)+p'/k)^2/2\tau} 
\, .
\end{align}
To see this, we will again use the definition of $\frac{1}{2} H^{k,j,j'} \lp \tau; u-v, w \rp$
from equation (\ref{eq:H^kjj_defn}) to write $\frac{1}{2} H^{k,j,j'} \lp \tau; u-v+1, w \rp$ as
\begin{equation}\label{eq:prop_proof2}
\mu_{\mathrm{sym}}^{k,j,j'}\lp \tau; v, u+1, w \rp  
									+  \frac{e^{\pi \i (u-v+1)^2 /\tau}\, 
									e^{-\pi \i w^2 (1-2/k)/\tau}}{\sqrt{-\i \tau}}  \sum_{p,p'=1}^{k-1}  
  									 \wh{{\CS}}^{(k)}_{j,p}  \, \wh{{\CS}}^{(k)}_{j',p'} \,
  \mu_{\mathrm{sym}}^{k,p,p'}\Big( -\frac{1}{\tau}; \frac{v}{\tau}, \frac{u+1}{\tau}, \frac{w}{\tau} \Big).
\end{equation}
The first factor $\mu_{\mathrm{sym}}^{k,j,j'}\lp \tau; v, u+1, w \rp$ is simply equal to
$-\, \mu_{\mathrm{sym}}^{k,j,j'}\lp \tau; v, u, w \rp$ by part (b) of Proposition \ref{prop:mu_algell1}.

The second factor is more complicated and we should be very careful applying part (a) of Proposition \ref{prop:mu_algell2}.
$ \delta_{p, p' \, (\mathrm{mod}\, 2)}$ factor in 
\begin{align}
\mu^{k,p,p'}&(\tau;v,u-\tau,w) + y_v y_u^{-1} q^{1/2} \mu^{k,p,p'}(\tau;v,u,w) =  \notag \\
&q^{-p'^2/4k} \, q^{p'/2} \,
\lp \frac{y_v}{y_u y_w^{1-2/k}} \rp^{p'/2} \chi_{p-1}^{(k-2)}(\tau; w+u-v-\tau)
\ \delta_{p, p' \, (\mathrm{mod}\, 2)}
\end{align}
is especially important.
Using equation (\ref{eq:chi_elltau}) we can express this as
\begin{equation}
 y_v \, y_u^{-1} \, q^{1/2} \,
\lb q^{-(k-p')^2/4k} \,
\lp \frac{y_u y_w^{1-2/k}}{y_v} \rp^{(k-p')/2} \chi_{k-p-1}^{(k-2)}(\tau; w+u-v)
\ \delta_{p, p' \, (\mathrm{mod}\, 2)} \rb.
\end{equation}
If we substitute $v \to v/\tau$, $u \to u/\tau$, $w \to w/\tau$ and $\tau \to -1/\tau$ we get
\begin{align}
\mu^{k,p,p'}&\lp -\frac{1}{\tau};\frac{v}{\tau},\frac{u+1}{\tau},\frac{w}{\tau} \rp  = 
e^{- \pi \i / \tau} e^{-2 \pi \i (u-v)/\tau} 
\Big[
- \mu^{k,p,p'}\lp -\frac{1}{\tau};\frac{v}{\tau},\frac{u}{\tau},\frac{w}{\tau} \rp   
+ e^{ \pi \i (k-p')^2/2k \tau} \, \notag \\
&\times  e^{ \pi \i (u - v + w(1-2/k))(k-p')/\tau} \,
 \chi_{k-p-1}^{(k-2)} \lp -\frac{1}{\tau}; \frac{w+u-v}{\tau} \rp
 \, \delta_{p, p' \, (\mathrm{mod}\, 2)} \Big].
\end{align}
Using this expression in (\ref{eq:prop_proof2}) together with equation (\ref{eq:prop_proof3}) we find $\frac{1}{2} \, H^{k,j,j'} \lp \tau; u-v+1, w \rp$ to be equal to
\begin{align}
-\, \mu_{\mathrm{sym}}^{k,j,j'}&\lp \tau; v, u, w \rp 
 - \, \frac{e^{\pi \i (u-v)^2 /\tau}\, 
									e^{-\pi \i w^2 (1-2/k)/\tau}}{\sqrt{-\i \tau}}  
\sum_{p,p'=1}^{k-1}   \wh{{\CS}}^{(k)}_{j,p}  \, \wh{{\CS}}^{(k)}_{j',p'} \,
  \mu_{\mathrm{sym}}^{k,p,p'}\Big( -\frac{1}{\tau}; \frac{v}{\tau}, \frac{u}{\tau}, \frac{w}{\tau} \Big)  \notag\\
  &+ \frac{e^{\pi \i (u-v)^2 /\tau}\, 
									e^{-\pi \i w^2 (1-2/k)/\tau}}{\sqrt{-\i \tau}} 
	\sum_{p,p'=1}^{k-1}   \wh{{\CS}}^{(k)}_{j,p}  \, \wh{{\CS}}^{(k)}_{j',p'} \,
	e^{ \pi \i (k-p')^2/2k \tau} \, 
e^{ \pi \i (u - v + w(1-2/k))(k-p')/\tau} \, \notag\\
&\times \chi_{k-p-1}^{(k-2)} \lp -\frac{1}{\tau}; \frac{w+u-v}{\tau} \rp
 \, \delta_{p, p' \, (\mathrm{mod}\, 2)} 
\end{align}
We recognize the first line as $- \, \frac{1}{2} \, H^{k,j,j'} \lp \tau; u-v, w \rp$. At this point
we replace $u$ with $v+z$ and change the dummy summation variables in the second line as
$p \to k-p$, $p' \to k- p'$. Since $\wh{{\CS}}^{(k)}_{j,k-p} = (-1)^{j+1} \wh{{\CS}}^{(k)}_{j,p}$
$\wh{{\CS}}^{(k)}_{j',k-p'} = (-1)^{j'+1} \wh{{\CS}}^{(k)}_{j',p'}$ and since we are assuming $j = j' \, (\mathrm{mod}\, 2)$ we find
\begin{equation}
H^{k,j,j'} \lp \tau; z+1, w \rp  + H^{k,j,j'} \lp \tau; z, w \rp 
\end{equation}
to be
\begin{align}
 \frac{2\, e^{\pi \i z^2 /\tau}   \,   e^{-\pi \i w^2 (1-2/k)/\tau}}{\sqrt{-\i \tau}} 
	\sum_{p,p'=1}^{k-1}   \wh{{\CS}}^{(k)}_{j,p}  \, \wh{{\CS}}^{(k)}_{j',p'} \,
	&e^{ \pi \i {p'}^2/2k \tau} \, 
e^{ \pi \i (z + w(1-2/k))p'/\tau} \, \notag \\
&\times \chi_{p-1}^{(k-2)} \lp -\frac{1}{\tau}; \frac{w+z}{\tau} \rp
 \, \delta_{p, p' \, (\mathrm{mod}\, 2)}.
\end{align}
We can rearrange the exponential factors to get
\begin{align}
\frac{2}{\sqrt{-\i \tau}} 
	\sum_{p,p'=1}^{k-1}   \wh{{\CS}}^{(k)}_{j,p}  \, \wh{{\CS}}^{(k)}_{j',p'} \,
&e^{ \pi \i k (z + w(1-2/k) + p'/k)^2/2\tau} \,
e^{- \pi \i (k-2) (z + w)^2/2\tau} \, \notag \\
&\times \chi_{p-1}^{(k-2)} \lp -\frac{1}{\tau}; \frac{w+z}{\tau} \rp
 \, \delta_{p, p' \, (\mathrm{mod}\, 2)}.
\end{align}
Next, we use $2\,  \delta_{p, p' \, (\mathrm{mod}\, 2)} = 
1 + (-1)^{p+p'}$, $\wh{{\CS}}^{(k)}_{j,p}  \, \wh{{\CS}}^{(k)}_{j',p'} (-1)^{p+p'} = 
\wh{{\CS}}^{(k)}_{k-j,p}  \, \wh{{\CS}}^{(k)}_{k-j',p'}$ and equation (\ref{eq:chi_modS}) to find
\begin{equation}
\frac{1}{\sqrt{-\i \tau}} 
	\sum_{r,p,p'=1}^{k-1}   
	\lb \wh{{\CS}}^{(k)}_{j,p}  \, \wh{{\CS}}^{(k)}_{j',p'} 
	+ \wh{{\CS}}^{(k)}_{k-j,p}  \, \wh{{\CS}}^{(k)}_{k-j',p'} \rb 
e^{ \pi \i k (z + w(1-2/k) + p'/k)^2/2\tau} \,
\wh{{\CS}}^{(k)}_{p,r} \, 
 \chi_{r-1}^{(k-2)} \lp \tau; w+z \rp.
\end{equation}
Summing over $p$ first, then over $r$ finally proves (\ref{eq:H^kjj_prop2}).

Equations ( \ref{eq:H^kjj_prop1}) and (\ref{eq:H^kjj_prop2}) completely determine
$H^{k,j,j'} \lp \tau; z, w \rp $ because any entire function of $z$ satisfying these two 
equations should be equal to $H^{k,j,j'} \lp \tau; z, w \rp $. To prove that, suppose there are two entire functions $g_1(z)$ and $g_2(z)$ obeying them. Then their difference 
$f(z) = g_1(z) - g_2(z)$ is an entire function satisfying 
\begin{equation}
f(z) + f(z+1) = 0  \qquad \mathrm{and} \qquad
f(z) + e^{- 2 \pi \i z - \pi \i \tau} f(z+\tau) = 0.
\end{equation}
Then for arbitrary integers $m$ and $n$ we get
\begin{equation}\label{eq:prop_proof4}
f(z_0 + m \tau + n) = (-1)^{m+n} e^{\pi \i m^2 \tau + 2 \pi \i m z_0} f(z_0).
\end{equation}
Varying $z_0$ over $0$, $1$, $\tau$, $1+\tau$ parallelogram and $m$ and $n$ over all integers for equation (\ref{eq:prop_proof4}) to see that $f(z)$ is a bounded function and hence has no $z$ dependence at all by Liouville's theorem. Letting $m$ to infinity in (\ref{eq:prop_proof4}) shows that it is in fact zero.

As our final task let us define 
\begin{align}
G^{k,j,j'} \lp \tau; z, w \rp \equiv \frac{1}{2} \Big[
 \  &\chi^{(k-2)}_{j-1}(\tau; z+w) \ \wh{h}_{k,j'}(\tau; z+w(1-2/k)) \notag\\
 &+ \chi^{(k-2)}_{k-j-1}(\tau; z+w) \ \wh{h}_{k,k-j'}(\tau; z+w(1-2/k)) \ 
 \Big]
\end{align}
for $j,j' = 1, \ldots, k-1$ and $j = j' \, (\mathrm{mod}\, 2)$. Since 
$z \to G^{k,j,j'} \lp \tau; z, w \rp$ is an entire function, our proof of (\ref{prop:mu_Sdeficit}) will be complete if we can show  $G^{k,j,j'} \lp \tau; z, w \rp$ satisfies equations (\ref{eq:H^kjj_prop1}) and (\ref{eq:H^kjj_prop2}).
This, in turn, follows from equations (\ref{eq:chi_ell1}), (\ref{eq:chi_elltau} )and parts (a) and (b) of Proposition \ref{prop:hath_ell}.
\end{proof}


\begin{thm}\label{thm:mu_completion}
We define a completion for 
 $\mu_{\mathrm{sym}}^{k,j,j'} = \frac{1}{2} \lb \mu^{k,j,j'} + \mu^{k,k-j,k-j'} \rb$ functions as:
 \begin{align}
  \tilde{\mu}^{k,j,j'}(\tau;v,u,w) 
  \equiv \frac{1}{2} \Big[   \mu^{k,j,j'}&\lp \tau; v, u, w \rp  
									+ \mu^{k,k-j,k-j'}\lp \tau; v, u, w \rp  \Big]
- \, \frac{1}{4} \, \delta_{j, j' \, (\mathrm{mod}\, 2)}  \notag \\ \times
\, \Big[  &\chi^{(k-2)}_{j-1}(\tau; u-v+w) \ \wh{R}_{k,j'}(\tau;u-v+w(1-2/k)) \notag\\
 &+ \chi^{(k-2)}_{k-j-1}(\tau; u-v+w) \ \wh{R}_{k,k-j'}(\tau;u-v+w(1-2/k)) \ 
 \Big].
 \end{align}
for $\tau \in \IH$ and $u,v \in \IC - (\IZ \tau + \IZ)$.
Then this function satisfies:
\begin{description}
 \item [(a)]  $\tilde{\mu}^{k,j,j'}(\tau+1;v,u,w) = e^{-\pi \i /4} \,
 e^{\pi \i (j^2 - {j'}^2) /2k}  \, \tilde{\mu}^{k,j,j'}(\tau;v,u,w)$.
 \item [(b)]  $\displaystyle\sum_{p,p'=1}^{k-1}  \wh{{\CS}}^{(k)}_{j,p}  \, \wh{{\CS}}^{(k)}_{j',p'}  
 \, \tilde{\mu}^{k,p,p'}
 \Big( -\frac{1}{\tau}; \frac{v}{\tau}, \frac{u}{\tau}, \frac{w}{\tau} \Big) = - \displaystyle
 \sqrt{-\i \tau}\,  e^{-\pi \i (u-v)^2 /\tau} \, e^{\pi \i w^2 (1-2/k) /\tau} \,
 \tilde{\mu}^{k,j,j'}(\tau;v,u,w)$.
\end{description}
\end{thm}

\begin{proof}

\noindent (a) This part quickly follows from equation (\ref{eq:chi_modT}), Proposition
\ref{prop:mu_modT} and part (a) of Proposition \ref{prop:Rhat_mod}.

\noindent (b) Let us define 
\begin{align}
  F^{k,j,j'}(\tau;v,u,w) 
  \equiv \frac{1}{4} \, &\delta_{j, j' \, (\mathrm{mod}\, 2)} 
\, \Big[  \chi^{(k-2)}_{j-1}(\tau; u-v+w) \ \wh{R}_{k,j'}(\tau;u-v+w(1-2/k)) \notag\\
 &+ \chi^{(k-2)}_{k-j-1}(\tau; u-v+w) \ \wh{R}_{k,k-j'}(\tau;u-v+w(1-2/k)) \ 
 \Big].
 \end{align}
Our statement will be proven if we can show that
 \begin{align}\label{eq:F_Sdeficit}
F^{k,j,j'}(\tau;v,u,w)   &+  
  \frac{e^{\pi \i (u-v)^2 /\tau}\, e^{-\pi \i w^2 (1-2/k)/\tau}}{\sqrt{-\i \tau}}  
  \sum_{p,p'=1}^{k-1}  
   \wh{{\CS}}^{(k)}_{j,p}  \, \wh{{\CS}}^{(k)}_{j',p'}\,
  F^{k,p,p'}\Big( -\frac{1}{\tau}; \frac{v}{\tau}, \frac{u}{\tau}, \frac{w}{\tau} \Big)
  \notag \\
 = &\frac{1}{4} \, \delta_{j, j' \, (\mathrm{mod}\, 2)}  \Big[
 \  \chi^{(k-2)}_{j-1}(\tau; u-v+w) \ \wh{h}_{k,j'}(\tau;u-v+w(1-2/k)) \notag\\
 &\qquad + \chi^{(k-2)}_{k-j-1}(\tau; u-v+w) \ \wh{h}_{k,k-j'}(\tau;u-v+w(1-2/k)) \ 
 \Big].
 \end{align}
 That is because if  this is the  case,
  the right hand side of equation (\ref{eq:F_Sdeficit}) will cancel the 
 error term that comes from the S transformation of $\mu^{k,j,j'} + \mu^{k,k-j,k-j'}$ part of 
 $\wt{\mu}^{k,j,j'}$ (see
 Proposition \ref{prop:mu_Sdeficit}).
 
 We start by writing 
 $\displaystyle\sum_{p,p'=1}^{k-1}
   \wh{{\CS}}^{(k)}_{j,p}  \, \wh{{\CS}}^{(k)}_{j',p'}\,
  F^{k,p,p'}\lp -\frac{1}{\tau}; \frac{v}{\tau}, \frac{u}{\tau}, \frac{w}{\tau} \rp$ as
 \begin{align}
 \frac{1}{8} \, \sum_{p,p'=1}^{k-1} 
 \Big[   \wh{{\CS}}^{(k)}_{j,p}  \, &\wh{{\CS}}^{(k)}_{j',p'}
      +  \wh{{\CS}}^{(k)}_{k-j,p}  \, \wh{{\CS}}^{(k)}_{k-j',p'}  \Big]
 \Bigg[  \chi^{(k-2)}_{p-1}\lp -\frac{1}{\tau}; \frac{u-v+w}{\tau}\rp \ 
 \wh{R}_{k,p'}\lp -\frac{1}{\tau}; \frac{u-v+w(1-2/k)}{\tau}\rp  \notag\\
 &+ \chi^{(k-2)}_{k-p-1}\lp -\frac{1}{\tau}; \frac{u-v+w}{\tau}\rp  \ 
 \wh{R}_{k,k-p'}\lp -\frac{1}{\tau}; \frac{u-v+w(1-2/k)}{\tau}\rp \ 
 \Bigg].
 \end{align}
We have used $2\,  \delta_{p, p' \, (\mathrm{mod}\, 2)} = 
1 + (-1)^{p+p'}$ and $\wh{{\CS}}^{(k)}_{j,p}  \, \wh{{\CS}}^{(k)}_{j',p'} (-1)^{p+p'} = 
\wh{{\CS}}^{(k)}_{k-j,p}  \, \wh{{\CS}}^{(k)}_{k-j',p'}$ to obtain this form.
Changing the dummy variables $p,p' \to k-p, k-p'$ for the 
$\chi^{(k-2)}_{k-p-1}  \wh{R}_{k,k-p'}$ term yields a $(-1)^{j+j'}$ factor to give 
\begin{align}
 \frac{\delta_{j, j' \, (\mathrm{mod}\, 2)}}{4} \, \sum_{p,p'=1}^{k-1} 
 &\Big[   \wh{{\CS}}^{(k)}_{j,p}  \, \wh{{\CS}}^{(k)}_{j',p'}
      +  \wh{{\CS}}^{(k)}_{k-j,p}  \, \wh{{\CS}}^{(k)}_{k-j',p'}  \Big] \notag \\
 &\times \Bigg[  \chi^{(k-2)}_{p-1}\lp -\frac{1}{\tau}; \frac{u-v+w}{\tau}\rp \ 
 \wh{R}_{k,p'}\lp -\frac{1}{\tau}; \frac{u-v+w(1-2/k)}{\tau}\rp  \Bigg].
 \end{align}

Using equation (\ref{eq:chi_modS}) then and combining exponential factors we can rewrite
\begin{equation}
 \frac{e^{\pi \i (u-v)^2 /\tau}\, e^{-\pi \i w^2 (1-2/k)/\tau}}{\sqrt{-\i \tau}}  
  \sum_{p,p'=1}^{k-1}  
   \wh{{\CS}}^{(k)}_{j,p}  \, \wh{{\CS}}^{(k)}_{j',p'}\,
  F^{k,p,p'}\Big( -\frac{1}{\tau}; \frac{v}{\tau}, \frac{u}{\tau}, \frac{w}{\tau} \Big)
\end{equation}
as
\begin{align}
 &\frac{e^{\pi \i k (u - v + w(1-2/k))^2 / 2 \tau}}{4 \sqrt{-\,\i\,\tau}}
 \Bigg[   \chi^{(k-2)}_{j-1}(\tau; u-v+w)  \, 
 \sum_{p'=1}^{k-1}  \wh{{\CS}}^{(k)}_{j',p'}  
 								\wh{R}_{k,p'}\lp -\frac{1}{\tau}; \frac{u-v+w(1-2/k)}{\tau}\rp
       \notag \\
      &+\chi^{(k-2)}_{k-j-1}(\tau; u-v+w)  \, 
      \sum_{p'=1}^{k-1}  \wh{{\CS}}^{(k)}_{k-j',p'} 
      					 \wh{R}_{k,p'}\lp -\frac{1}{\tau}; \frac{u-v+w(1-2/k)}{\tau}\rp \Bigg] \,
      					 \delta_{j, j' \, (\mathrm{mod}\, 2)}.
 \end{align}
Using part (b) of Proposition \ref{prop:Rhat_mod} this is equal to
\begin{align}
-\,F^{k,j,j'}(\tau;v,u,w)   &+ \frac{\delta_{j, j' \, (\mathrm{mod}\, 2)}}{4}   \Big[
 \  \chi^{(k-2)}_{j-1}(\tau; u-v+w) \ \wh{h}_{k,j'}(\tau;u-v+w(1-2/k)) \notag\\
 &\qquad + \chi^{(k-2)}_{k-j-1}(\tau; u-v+w) \ \wh{h}_{k,k-j'}(\tau;u-v+w(1-2/k)) \ 
 \Big].
 \end{align}
 proving equation (\ref{eq:F_Sdeficit}).   \end{proof}

 We also define
\begin{equation}
\wt{\mu}^{(k,d)}(\tau;v,u,w) = 
		\sum_{j,j'=1}^{k-1} \wh{\Omega}^{k,d}_{j,j'} \ \wt{\mu}^{k,j,j'}(\tau;v,u,w)
\end{equation}
for $d | k$ and
\begin{equation}
\wt{\mu}^{Y}(\tau;v,u,w) = 
		\sum_{j,j'=1}^{k-1} \wh{\Omega}^Y_{j,j'} \ \wt{\mu}^{k,j,j'}(\tau;v,u,w)
\end{equation}
for a simply laced root system $Y$ with Coxeter number $k$.
We notice that because of the 
$\wh{\Omega}^{(k,d)}_{j,j'} = \wh{\Omega}^{(k,d)}_{k-j,k-j'}$ property and because
$\wh{\Omega}^{(k,d)}_{j,j'}  = 0$ unless $j = j' \, (\mathrm{mod}\, 2)$, these combinations have simple forms such as
\begin{equation}
 \tilde{\mu}^{Y}(\tau;v,u,w) = \mu^{Y}(\tau; v,u,w) - 
 \, \frac{1}{2} \,
  \sum_{j,j'=1}^{k-1} \wh{\Omega}^{Y}_{j,j'} \,
 \chi^{(k-2)}_{j-1}(\tau; u-v+w) \ \wh{R}_{k,j'}(\tau;u-v+w(1-2/k))  \, .
\end{equation}

Lastly, $u=v$ and $w=0$ case is specifically important for this work. We note that
\begin{equation}
 \chi^{(k-2)}_{j-1}(\tau; 0) = \frac{S_{k,j}(\tau)}{\eta(\tau)^3}\, ,
 \quad \mathrm{and} \quad
 \wh{R}_{k,j'}(\tau;0) = \lp \pi k \rp^{-1/2} S^*_{k,j'}(\tau)\, .
\end{equation}

\begin{corollary}\label{cor:mu_mod}
 For $\tau \in \IH$ and $z \in \IC - (\IZ \tau + \IZ)$ we have 
 \begin{align}
  \tilde{\mu}^{Y}(\tau;z,z,0) &= \mu^{Y}(\tau; z,z,0) - 
  \frac{1}{\eta(\tau)^3} \, \frac{1}{\sqrt{4 \pi k}} \,
  \sum_{j,j'=1}^{k-1} \wh{\Omega}^{Y}_{j,j'} \,
 S_{k,j}(\tau) \, S^*_{k,j'}(\tau)\, \\
 &= \mu^{Y}(\tau; z,z,0) - 
  \frac{1}{\eta(\tau)^3} \, \frac{1}{\sqrt{16 \pi k}} \,
  \sum_{j,j' \in \ZZ_{2k}} {\Omega}^{Y}_{j,j'} \,
 S_{k,j}(\tau) \, S^*_{k,j'}(\tau)\, .
 \end{align}
which obeys:
\begin{description}
 \item [(a)]  $\tilde{\mu}^{Y}(\tau+1;z,z,0) = e^{-\pi \i /4} \,  \tilde{\mu}^{Y}(\tau;z,z,0)$.
 \item [(b)]  $\tilde{\mu}^{Y}\Big( -\frac{1}{\tau}; \frac{z}{\tau}, \frac{z}{\tau}, 0 \Big) = - \displaystyle\,
 \sqrt{-\i \tau}\  \tilde{\mu}^{Y}(\tau;z,z,0)$.
 \item [(c)]  $\lb \eta(\tau+1) \rb^3\  \tilde{\mu}^{Y}(\tau+1; z,z,0) = \lb \eta(\tau) \rb^3 \ \tilde{\mu}^{Y}(\tau; z,z,0)$.
 \item [(d)]  $\lb \eta(-1/\tau) \rb^3 \  \tilde{\mu}^{Y}\Big( -\frac{1}{\tau}; \frac{z}{\tau}, \frac{z}{\tau}, 0 \Big) = \tau^2  \, \lb \eta(\tau) \rb^3 \  \tilde{\mu}^{Y}(\tau; z,z,0)$.
\end{description}
\end{corollary}
\begin{proof}
These properties quickly follow from modular transformation properties of $\eta(\tau)$, 
Theorem \ref{thm:mu_completion} and equations (\ref{eq:Omega_commutingT}) and 
(\ref{eq:Omega_commutingS}).
\end{proof}

\subsection{Decomposing \texorpdfstring{$\varphi_{0,1}$}{phi_01}} \label{decomp}
In this subsection we will prove the identity
\begin{equation}\label{eq:phi01_decompose}
\varphi_{-2,1}(\tau;z) \,  \eta(\tau)^3 \,
\Big[ - \mu^{Y}(\tau;z,z,0)  + \frac{1}{3} \sum_{w \in \Pi_2} \mu^{Y}(\tau;w,w,0) \Big]
 = \frac{\mathrm{rk}(Y)}{12} \, \varphi_{0,1}(\tau;z),
\end{equation}
where $\Pi_2  = \{ \frac{1}{2}, \frac{\tau}{2}, \frac{\tau+1}{2} \}$
and $\varphi_{0,1}, \varphi_{-2,1}$ are the weak Jacobi forms defined in (\ref{phzeroone}), (\ref{phtwoone}). 
At $k=2$ with $Y=A_1$ this reduces to the relation given in \cite{Eguchi:2008gc}, \cite{Eguchi:2009cq} and employed in \cite{Harvey:2013mda} (eqn. A.24) for the evaluation of the helicity supertrace. 

Our proof is an easy application of the Lemma \ref{lem:mu_z_shift} which gives
 \begin{align}
 \mu^{k,j,j'}(\tau;w,w,0) &- \mu^{k,j,j'}(\tau;z,z,0)  \notag \\
 &= \frac{\i \,\eta(\tau)^3 \  C^{j-1,k}_{j'}\bmat{1 \\ 1}(\tau;0) \  \theta_{11}(\tau;w-z) \ \theta_{11}(\tau;w+z)}
 {\theta_{11}(\tau;z)^2 \  \theta_{11}(\tau;w)^2}
\end{align}
for each $w \in \Pi_2$. Since $C^{j-1,k}_{j'}\bmat{1 \\ 1}(\tau;0) 
= \i \delta_{j',j }$ in the range $j,j' = 1, \ldots, k-1$ and $\displaystyle\sum_{j=1}^{k-1} \wh{\Omega}^Y_{j,j} = \mathrm{rk}(Y)$, we can rewrite the left hand side of the equation (\ref{eq:phi01_decompose}) as
\begin{equation}
\frac{\mathrm{rk}(Y)}{3} \sum_{w \in \Pi_2} \frac{\theta_{11}(\tau;w-z) \ \theta_{11}(\tau;w+z)}{\theta_{11}(\tau;w)^2}
\end{equation}
which is just $\frac{\mathrm{rk}(Y)}{12} \, \varphi_{0,1}(\tau;z)$.

We can also generalize this identity  by replacing $\mu^Y$'s with $\mu^{(k,d)}$ functions. Similar arguments apply to this case; the only change is replacing  $\mathrm{rk}(Y)$
with $\displaystyle\sum_{j=1}^{k-1} \wh{\Omega}^{(k,d)}_{j,j} $ which is just 
$\frac{k}{d} - d$. We get
\begin{equation}
\varphi_{-2,1}(\tau;z) \,  \eta(\tau)^3 \,
\Big[ - \mu^{(k,d)}(\tau;z,z,0)  + \frac{1}{3} \sum_{w \in \Pi_2} \mu^{(k,d)}(\tau;w,w,0) \Big]
 = \frac{\frac{k}{d} - d}{12} \, \varphi_{0,1}(\tau;z).
\end{equation}

\subsection{Riemann Relations}\label{sec:Riemann_reln}
We start by defining 
\begin{equation}
B^{k,j,j'}_{a b}(\tau; v,u,w) \equiv 
\frac{\th_{ab}(\tau; v)\, \th_{ab}(\tau;u)}{\eta(\t)^3} \   
 \mu^{k,j,j'} \lp \tau; v + \t_{ab}, u +  \t_{ab}, w \rp,
\end{equation}
where $\tau_{ab} \equiv (a-1) \t/2 + (b-1)/2$, $k \geq 2$ is an integer, 
$j,j' = 1, \ldots, k-1$ and $a,b \in \{ 0,1 \}$.
Using equation (\ref{eq:mu_v_residue}) we find
\begin{equation}
B^{k,j,j'}_{1 1}(\tau; 0,u,w) = - \, \i \,  C^{j-1,k}_{j'}\bmat{1 \\ 1}(\tau;w).
\end{equation}
Since $C^{j-1,k}_{j'}\bmat{1 \\ 1}(\tau;0) 
= \i \,  \lp \delta_{j',j \, (\mathrm{mod}\, 2k) } -  \delta_{- j',j \,(\mathrm{mod}\, 2k) } \rp$, in the range $j,j' = 1, \ldots, k-1$ (and also for $j' = 0$ or $k$ as we will need in the main text) we get
\begin{equation}
B^{k,j,j'}_{1 1}(\tau; 0,u,0) = \delta_{j,j'}.
\end{equation}

In this section we will prove some identities which are similar to the Riemann relations satisfied by theta functions. In particular, we will show that
\begin{align}
 \th_{00}&(\t;x) \, \th_{00}(\t;z) \,  B^{k,j,j'}_{00}(\t;v,u,w)
 - \th_{01}(\t;x) \, \th_{01}(\t;z) \, B^{k,j,j'}_{01}(\t;v,u,w)  \notag\\
 &- \th_{10}(\t;x) \, \th_{10}(\t;z) \, B^{k,j,j'}_{10}(\t;v,u,w)
 + \th_{11}(\t;x) \, \th_{11}(\t;z) \, B^{k,j,j'}_{11}(\t;v,u,w) \notag \\
 \label{eq:mu_riemann}
 &= 2 \,\th_{11}(\t;x_0) \, \th_{11}(\t;z_0) \, B^{k,j,j'}_{11}(\t;v_0,u_0,w),
\end{align}
where
\begin{align}
 & x_0 = \frac{1}{2}(x+z+v+u)  \, ,  \qquad \qquad z_0 = \frac{1}{2}(x+z-v-u) \, , \cr
 & v_0 = \frac{1}{2}(x-z+v-u)  \, , \qquad \qquad u_0 = \frac{1}{2}(x-z-v+u) .
\end{align}

We will use Lemma \ref{lem:mu_z_shift} to get
 \begin{align}
  \mu^{k,j,j'} ( \tau; v + \t_{ab}, &u +  \t_{ab}, w ) - \mu^{k,j,j'}(\tau;v,u,w)  \notag\\ 
 &= \frac{\i \eta(\tau)^3 \  C^{j-1,k}_{j'}\bmat{1 \\ 1}(\tau;w) \  \theta_{ab}(\tau;0) \ \theta_{ab}(\tau;v+u)}
 {\theta_{11}(\tau;v) \ \theta_{11}(\tau;u) \ \theta_{ab}(\tau;v) \  \theta_{ab}(\tau;u)}.
\end{align}
If we use this on the left hand side of (\ref{eq:mu_riemann}), we find
\begin{align}
&\sum_{a,b\,=\,0,1} (-1)^{a+b} 
\th_{ab}(\t;x) \, \th_{ab}(\t;z) \,  B^{k,j,j'}_{ab}(\t;v,u,w)
\notag \\
& \qquad =
\frac{ \mu^{k,j,j'} \lp \tau; v, u, w \rp}{\eta(\t)^3} 
\, \sum_{a,b\,=\,0,1} (-1)^{a+b} 
\th_{ab}(\t;x) \, \th_{ab}(\t;z) \, \th_{ab}(\t;v) \, \th_{ab}(\t;u)  \notag \\
& \qquad \qquad +
\frac{ \i \, C^{j-1,k}_{j'}\bmat{1 \\ 1}(\tau;w)}{\th_{11}(\t;v) \, \th_{11}(\t;u) } 
\, \sum_{a,b\,=\,0,1} (-1)^{a+b} 
\th_{ab}(\t;x) \, \th_{ab}(\t;z) \, \th_{ab}(\t;0) \, \th_{ab}(\t;v+u) \, .
\end{align}
Using the Riemann theta relation (R5) this is equal to
\begin{align}
&\frac{2\, \mu^{k,j,j'} \lp \tau; v, u, w \rp}{\eta(\t)^3} 
\, \th_{11}(\t;x_0) \, \th_{11}(\t;z_0) \, \th_{11}(\t;v_0) \, \th_{11}(\t;u_0)  \notag \\
& \qquad \qquad +
\frac{2 \i \, C^{j-1,k}_{j'}\bmat{1 \\ 1}(\tau;w)}{\th_{11}(\t;v) \, \th_{11}(\t;u) } \,
\th_{11}(\t;x_0) \, \th_{11}(\t;z_0) \, \th_{11}(\t;v_0-v) \, \th_{11}(\t;v_0+u) \, 
\end{align}
which we can rearrange as
\begin{align}
&2 \, \th_{11}(\t;x_0) \, \th_{11}(\t;z_0) \, 
\frac{\th_{11}(\t;v_0) \, \th_{11}(\t;u_0)}{\eta(\t)^3}
\Big[ \mu^{k,j,j'} \lp \tau; v, u, w \rp  \notag \\
& \qquad \qquad 
+
\frac{ \i \, \eta(\t)^3 \, C^{j-1,k}_{j'}\bmat{1 \\ 1}(\tau;w)
\, \th_{11}(\t;v_0-v) \, \th_{11}(\t;(v_0 - v) +u + v)}
{\th_{11}(\t;v) \, \th_{11}(\t;u) \,
\th_{11}(\t;v_0) \, \th_{11}(\t;u_0)}  \Big] \, .
\end{align}
Noting that $u_0 = (v_0 - v)+u$, $v_0 = (v_0 -v) +v$ and employing Lemma \ref{lem:mu_z_shift} once more we obtain
\begin{equation}
2 \, \th_{11}(\t;x_0) \, \th_{11}(\t;z_0) \ 
\frac{\th_{11}(\t;v_0) \, \th_{11}(\t;u_0)}{\eta(\t)^3}  \ 
 \mu^{k,j,j'} \lp \tau; v_0, u_0, w \rp
\end{equation}
which is just 
\begin{equation}
 2 \,\th_{11}(\t;x_0) \, \th_{11}(\t;z_0) \, B^{k,j,j'}_{11}(\t;v_0,u_0,w)
\end{equation}
as we wanted to show.

Shifting $x,z$ in equation (\ref{eq:mu_riemann}) by various factors of $\pm 1/2, \pm \tau/2, \pm (1+\tau)/2$ we obtain
\begin{align} \nonumber
(\wt{R5}): + \theta_{00} \theta_{00} B^{k,j,j'}_{00}
- \theta_{01} \theta_{01} B^{k,j,j'}_{01}
- \theta_{10}\theta_{10} B^{k,j,j'}_{10} 
+ \theta_{11} \theta_{11} B^{k,j,j'}_{11} & = 
2 \theta_{11} \theta_{11} B^{k,j,j'}_{11}, \\  \nonumber
(\wt{R8}):- \theta_{01} \theta_{01} B^{k,j,j'}_{00} 
+\theta_{00} \theta_{00} B^{k,j,j'}_{01}
 -\theta_{11} \theta_{11} B^{k,j,j'}_{10}
 + \theta_{10} \theta_{10}  B^{k,j,j'}_{11}  &=  
 -2 \theta_{11} \theta_{11} B^{k,j,j'}_{10}, \\ \nonumber 
(\wt{R9}): - \theta_{01} \theta_{01} B^{k,j,j'}_{00}
+ \theta_{00} \theta_{00} B^{k,j,j'}_{01} 
 +\theta_{11} \theta_{11} B^{k,j,j'}_{10}
- \theta_{10} \theta_{10}  B^{k,j,j'}_{11} &= 
 -2 \theta_{10} \theta_{10}  B^{k,j,j'}_{11}, \\ \nonumber 
(\wt{R11}): - \theta_{10} \theta_{10} B^{k,j,j'}_{00}
- \theta_{11} \theta_{11} B^{k,j,j'}_{01}
+\theta_{00} \theta_{00} B^{k,j,j'}_{10} 
+ \theta_{01} \theta_{01}  B^{k,j,j'}_{11}   &= 
2 \theta_{01} \theta_{01}  B^{k,j,j'}_{11}, \\ \nonumber
(\wt{R13}): - \theta_{10} \theta_{10} B^{k,j,j'}_{00}
 +\theta_{11} \theta_{11} B^{k,j,j'}_{01}
+\theta_{00} \theta_{00} B^{k,j,j'}_{10} 
- \theta_{01} \theta_{01}  B^{k,j,j'}_{11} &= 
   2 \theta_{11} \theta_{11} B^{k,j,j'}_{01}, \\ \nonumber   
(\wt{R15}): -\theta_{11} \theta_{11} B^{k,j,j'}_{00}
- \theta_{10} \theta_{10} B^{k,j,j'}_{01}
+ \theta_{01} \theta_{01} B^{k,j,j'}_{10}
+\theta_{00} \theta_{00}  B^{k,j,j'}_{11}   &= 
2 \theta_{00} \theta_{00}  B^{k,j,j'}_{11}, \\ \nonumber
(\wt{R16}):  -\theta_{11} \theta_{11} B^{k,j,j'}_{00}
+ \theta_{10} \theta_{10} B^{k,j,j'}_{01}
- \theta_{01} \theta_{01} B^{k,j,j'}_{10} 
+\theta_{00} \theta_{00}  B^{k,j,j'}_{11}  &= 
 -2 \theta_{11} \theta_{11} B^{k,j,j'}_{00}.
\end{align}
In these relations, the arguments of the $\th_{ab}$'s and $B^{k,j,j'}_{ab}$
 on the left hand side are 
$(\tau;x)$, $(\tau;z)$, $(\tau;v,u,w)$, in that order, and the arguments 
for $\th_{ab}$'s and $B^{k,j,j'}_{ab}$ on the right hand
side are $(\tau;x_0)$, $(\tau;z_0)$, $(\tau;v_0,u_0,w)$, again in that order.


\bibliography{fbrane_refs}
\bibliographystyle{JHEP}

\end{document}